\documentclass{CSML}
\pdfoutput=1

\usepackage{lastpage}
\newcommand{\dOi}{14(2:7)2018}
\lmcsheading{}{1--\pageref{LastPage}}{}{}%
{Oct.~13, 2016}{May~02, 2018}{}


\ACMCCS{[{\bf Theory of computation}]: Logic---Proof Theory}

\usepackage[LGR,T1]{fontenc}
\usepackage[utf8]{inputenc}
\usepackage[english]{babel}

\usepackage{amsfonts,amsmath,latexsym,amssymb,amstext}
\usepackage{mathtools}
\usepackage{cmll}
\usepackage{stmaryrd}

\usepackage{thmtools}

\usepackage{tikz}

\usepackage{STYLES/proofCF}

\usepackage{graphicx}

\def\eg{{\em e.g.}}
\def\ie{{\em i.e.}}
\def\Ie{{\em I.e.}}

\usepackage{xspace}

\colorlet{couleurCourante}{black}
\ifx\SansRouge\undefined
\newcommand{\ROUGE}[1]{{
\colorlet{oldcouleurCourante}{couleurCourante}\color{red}\colorlet{couleurCourante}{red}%
#1%
\color{oldcouleurCourante}\colorlet{couleurCourante}{oldcouleurCourante}\xspace
}}%
\else
\newcommand{\ROUGE}[1]{\ignorespaces
}%
\fi

\ifx\SansRougebis\undefined
\newcommand{\ROUGEbis}[1]{{
\colorlet{oldcouleurCourante}{couleurCourante}\color{violet}\colorlet{couleurCourante}{violet}%
#1%
\color{oldcouleurCourante}\colorlet{couleurCourante}{oldcouleurCourante}\xspace
}}%
\else
\newcommand{\ROUGEbis}[1]{\ignorespaces
}%
\fi

\ifx\SansBleu\undefined
\else
\fi

\ifx\SansMagenta\undefined
\newcommand{\MAGENTA}[1]{{
\colorlet{oldcouleurCourante}{couleurCourante}\color{magenta}\colorlet{couleurCourante}{magenta}%
#1%
\color{oldcouleurCourante}\colorlet{couleurCourante}{oldcouleurCourante}\xspace
}}%
\else
\newcommand{\MAGENTA}[1]{\ignorespaces
}%
\fi

\ifx\SansCyan\undefined
\else
\fi

\newcommand\bigshpos{
\mathinner{\mathchoice%
{\raise-.9pt\hbox{\scalebox{1.1}{$\displaystyle\shpos$}}}
{\raise-.9pt\hbox{\scalebox{1.1}{$\textstyle\shpos$}}}
{\raise-.3pt\hbox{\scalebox{1}{$\scriptstyle\shpos$}}}
{\raise.1pt\hbox{\scalebox{1}{$\scriptscriptstyle\shpos$}}}
}%
}
\newcommand\bigshneg{
\mathinner{\mathchoice%
{\raise-.9pt\hbox{\scalebox{1.1}{$\displaystyle\shneg$}}}
{\raise-.9pt\hbox{\scalebox{1.1}{$\textstyle\shneg$}}}
{\raise-.3pt\hbox{\scalebox{1}{$\scriptstyle\shneg$}}}
{\raise.1pt\hbox{\scalebox{1}{$\scriptscriptstyle\shneg$}}}
}%
}

\def\boldvdashdisplay{%
\mkern4mu%
\raise-.3pt\hbox{$\vdash$}\mkern-10.5mu\hbox{$\vdash$}\mkern-10.5mu\raise.3pt\hbox{$\vdash$}%
\mkern2mu%
}
\def\boldvdashscript{%
\mkern4mu%
\raise-.4pt\hbox{$\vdash$}\mkern-13mu\hbox{$\vdash$}\mkern-13mu\raise.4pt\hbox{$\vdash$}%
\mkern2mu%
}

\def\boldtopdisplay{%
\raise-.3pt\hbox{$\top$}\mkern-14.5mu\hbox{$\top$}\mkern-14.5mu\raise.3pt\hbox{$\top$}
}
\def\boldtopscript{%
\mkern4mu
\raise-.2pt\hbox{$\scriptstyle\top$}\mkern-14mu\hbox{$\scriptstyle\top$}\mkern-14mu\raise.2pt\hbox{$\scriptstyle\top$}
\mkern2mu%
}
\newcommand{\boldtop}{%
\mathord{\mathchoice%
{\boldtopdisplay}
{\boldtopdisplay}
{\boldtopscript}
{\boldtopscript}
}%
}

\newcommand{\boldbot}{%
\mathord{%
\raise-.3pt\hbox{$\bot$}\mkern-13.5mu\hbox{$\bot$}\mkern-13.5mu\raise.3pt\hbox{$\bot$}
}%
}

\usepackage{tensor}

\newcommand{\fullotensor}{\mathbin{\mathchoice%
{\hbox to 7.85 pt{\hss$\fullotensortextstyle$\hss}}
{\hbox to 7.85 pt{\hss$\fullotensortextstyle$\hss}}
{\hbox to 6.37 pt{\hss$\fullotensorscriptstyle$\hss}}
{\hbox to 5.4 pt{\hss$\fullotensorscriptscriptstyle$\hss}}
}}

\newcommand{\fullotensortextstyle}{\mathbin{%
\begin{tikzpicture}[line width=.2pt,baseline=-.65ex]
\draw[clip] (0,0) circle(3.3pt);
\fill[black] (-.12,0) -- (.12,-.11) -- (.12,.11) -- (-.12,0);
\end{tikzpicture}
}}

\newcommand{\fullotensorscriptstyle}{\mathbin{%
\begin{tikzpicture}[line width=.2pt,baseline=-.45ex]
\draw[clip] (0,0) circle(2.6pt);
\fill[black] (-.095,0) -- (.12,-.10) -- (.12,.10) -- (-.095,0);
\end{tikzpicture}
}}

\newcommand{\fullotensorscriptscriptstyle}{\mathbin{%
\begin{tikzpicture}[line width=.2pt,baseline=-.33ex]
\draw[clip] (0,0) circle(2pt);
\fill[black] (-.072,0) -- (.12,-.09) -- (.12,.09) -- (-.072,0);
\end{tikzpicture}
}}

\newcommand{\shuffle}{\mathbin{\mathchoice%
{\hbox to 7.85 pt{\hss$\shuffletextstyle$\hss}}
{\hbox to 7.85 pt{\hss$\shuffletextstyle$\hss}}
{\hbox to 6.4 pt{\hss$\shufflescriptstyle$\hss}}
{\hbox to 5.4 pt{\hss$\shufflescriptscriptstyle$\hss}}
}}

\newcommand{\shuffletextstyle}{\mathbin{%
\begin{tikzpicture}[line width=.5pt,baseline=-.08ex]
\draw[cap=round] (0,0) -- (.28,0);
\draw[cap=round] (0,0) -- (0,.19);
\draw[cap=round] (.14,0) -- (.14,.19);
\draw[cap=round] (.28,0) -- (.28,.19);
\end{tikzpicture}
}}

\newcommand{\shufflescriptstyle}{\mathbin{%
\begin{tikzpicture}[line width=.4pt,baseline=-.08ex]
\draw[cap=round] (0,0) -- (.20,0);
\draw[cap=round] (0,0) -- (0,.13);
\draw[cap=round] (.10,0) -- (.10,.13);
\draw[cap=round] (.20,0) -- (.20,.13);
\end{tikzpicture}
}}

\newcommand{\shufflescriptscriptstyle}{\mathbin{%
\begin{tikzpicture}[line width=.35pt,baseline=-.08ex]
\draw[cap=round] (0,0) -- (.16,0);
\draw[cap=round] (0,0) -- (0,.095);
\draw[cap=round] (.08,0) -- (.08,.095);
\draw[cap=round] (.16,0) -- (.16,.095);
\end{tikzpicture}
}}

\newcommand{\daimon}{{\scriptstyle \maltese}}
\newcommand{\chronicle}[1]{{\mathfrak{#1}}}
\newcommand{\design}[1]{{\mathfrak{#1}}}

\newcommand{\designset}[1]{{\mathrm{#1}}}
\newcommand{\behaviour}[1]{{\mathbf{#1}}}

\newcommand{\Bincarnation}[1]{|{\designset{#1}}|}

\newcommand{\Dincarnation}[2]{|{\design{#1}}|_{\designset{#2}}}

\newcommand{\normalisation}[1]{[\![#1]\!]}
\newcommand{\psdes}[2]{[\![ #1 , #2 ]\!]}
\newcommand{\proj}[2]{
\mathinner{\mathchoice%
{#1\!\!\upharpoonright\!\!#2}
{#1\!\!\upharpoonright\!\!#2}
{#1\mkern1mu\upharpoonright\mkern1mu #2}
{#1\!\!\upharpoonright\!\!#2}
}%
}

\newcommand{\normalisationSeq}[2]{\left<#1\!\!\leftarrow\!\!#2\right>}
\newcommand{\normalisationDes}[2]{\fullview{\normalisationSeq{\design{#1}}{\design{#2}}}}

\DeclareFontFamily{OT1}{pzc}{}
\DeclareFontShape{OT1}{pzc}{m}{it}{<-> [1.1] pzcmi8t}{} 
\DeclareMathAlphabet{\mathpzc}{OT1}{pzc}{m}{it}
\newcommand{\pathLL}[1]{\mathpzc{#1}}

\newlength{\Viewheight}
\newlength{\ulcornerheight}
\newcommand{\view}[1]{%
\settoheight{\Viewheight}{$#1$}%
\settoheight{\ulcornerheight}{$\ulcorner$}%
\addtolength{\Viewheight}{-\ulcornerheight}%
\addtolength{\Viewheight}{1pt}%
\raisebox{\Viewheight}{$\ulcorner$}{#1}\raisebox{\Viewheight}{$\urcorner$}}

\newlength{\Fullviewheight}
\newcommand{\fullview}[1]{%
\settoheight{\Fullviewheight}{$#1$}%
\settoheight{\ulcornerheight}{$\ulcorner$}%
\addtolength{\Fullviewheight}{-\ulcornerheight}%
\addtolength{\Fullviewheight}{1pt}%
\raisebox{\Fullviewheight}{$\ulcorner\mkern-6mu\ulcorner\mkern-2mu$}{#1}\raisebox{\Fullviewheight}{$\mkern-2mu\urcorner\mkern-6mu\urcorner$}}

\newcommand{\dai}{\design{Dai}}

\usepackage{slashed}

\newlength{\dualwidth}
\newlength{\dualheight}
\newcommand{\dual}[1]{
\mathinner{\mathchoice%
{
\dddual{#1}
}%
{
\dddual{#1}
}%
{
\dddual{\scriptstyle #1}
}%
{
\dddual{\scriptscriptstyle #1}
}%
}
}

\newcommand{\dddual}[2][.8]{
\settowidth{\dualwidth}{$#2$}%
\settoheight{\dualheight}{$#2$}%
\makebox[\dualwidth][c]{\mbox{\rule{0cm}{#1\dualheight}$\Widetilde[#1]{#2}$}}
}

\newlength{\ddualwidth}
\newlength{\ddualheight}
\newcommand{\ddual}[2][1]{
\settowidth{\ddualwidth}{$#2$}%
\settoheight{\ddualheight}{$#2$}%
\makebox[\ddualwidth][c]{\mbox{\rule{0cm}{#1\ddualheight}$\WWidetilde[#1]{#2}$}}
}

\def\lshuffledef#1{{\mkern-2.05mu\hbox{$\shuffle$}\mkern-11mu \hbox{\raisebox{2.5pt}{$\scriptscriptstyle l$}}\mkern2.5mu \hbox{\raisebox{2.5pt}{$\scriptscriptstyle l$}}}\mkern1mu}
\newcommand*{\lshuffle}{\mathbin{\mathchoice%
{\lshuffledef\displaystyle}
{\lshuffledef\textstyle}
{\lshuffledef\scriptstyle}
{\lshuffledef\scriptscriptstyle}}}

\newcommand{\PosOrder}{\olessthan}

\def\notPosOrderdef#1{\hbox{\hbox to -2pt{$#1/$\hss}$#1\PosOrder$}}

\newlength{\Overlineheight}
\newlength{\Overlinewidth}

\newcommand{\Overline}[2][\Overlinestretch]{%
\settowidth{\Overlinewidth}{$#2$}%
\settoheight{\Overlineheight}{$#2$}%
\setlength{\Overlineheight}{#1\Overlineheight}%
\rlap{$
\overline{
	\makebox[\Overlinewidth][c]{
		\rule{0cm}{\Overlineheight}
		}
	}$
}
#2
}

\makeatletter
\newlength{\Widetildeheight}
\newlength{\Widetildewidth}
\newcommand{\Widetildestretch}{1}
\newcommand*\Widetilde[2][\Widetildestretch]{%
        \begingroup
        \mathchoice{\Widetilde@helper{#1}{#2}{\displaystyle}{\textfont}}
                   {\Widetilde@helper{#1}{#2}{\textstyle}{\textfont}}
                   {\Widetilde@helper{#1}{#2}{\scriptstyle}{\scriptfont}}
                   {\Widetilde@helper{#1}{#2}{\scriptscriptstyle}{\scriptscriptfont}}%
        \endgroup
}

\newcommand*\Widetilde@helper[4]{%
\settowidth{\Widetildewidth}{$#2$}%
\settoheight{\Widetildeheight}{$#2$}%
\setlength{\Widetildeheight}{#1\Widetildeheight}%
\rlap{\raisebox{\Widetildeheight}{$\resizebox{\Widetildewidth}{\height}{$\sim$}$}}{#2}
}
\makeatother

\makeatletter
\newlength{\WWidetildeheight}
\newlength{\WWidetildewidth}
\newcommand{\WWidetildestretch}{1}
\newcommand*\WWidetilde[2][\WWidetildestretch]{%
        \begingroup
        \mathchoice{\WWidetilde@helper{#1}{#2}{\displaystyle}{\textfont}}
                   {\WWidetilde@helper{#1}{#2}{\textstyle}{\textfont}}
                   {\WWidetilde@helper{#1}{#2}{\scriptstyle}{\scriptfont}}
                   {\WWidetilde@helper{#1}{#2}{\scriptscriptstyle}{\scriptscriptfont}}%
        \endgroup
}

\newcommand*\WWidetilde@helper[4]{%
\settowidth{\WWidetildewidth}{$#2$}%
\settoheight{\WWidetildeheight}{$#2$}%
\setlength{\WWidetildeheight}{#1\WWidetildeheight}%
\rlap{\raisebox{\WWidetildeheight}{$\resizebox{\WWidetildewidth}{\height}{$\ssim$}$}}{#2}
}
\makeatother

\newcommand*{\ssim}{%
\mathrel{\vcenter{\offinterlineskip
\hbox{$\sim$}\vskip-.35ex\hbox{$\sim$}}}}

\usepackage{hyperref}
\hypersetup{hidelinks}

\begin{document}

\title[Ludics Behaviours]{Study of Behaviours via Visitable Paths}

\author[C.~Fouquer\'e]{Christophe Fouquer\'e}
\address{Universit\'e Paris 13, Sorbonne Paris Cité, LIPN, CNRS, France.}
\email{christophe.fouquere@lipn.univ-paris13.fr}

\author[M.~Quatrini]{Myriam Quatrini}	
\address{I2M, Aix-Marseille Universit\'e, CNRS, France.}	
\email{myriam.quatrini@univ-amu.fr}

\begin{abstract}
  \noindent Around 2000, J.-Y. Girard developed a logical theory, called Ludics. This theory was a step in his program of Geometry of Interaction, the aim of which being to account for the dynamics of logical proofs.  In Ludics, objects called designs keep only what is relevant for the cut elimination process, hence the dynamics of a proof: a design is an abstraction of a formal proof. The notion of behaviour is the counterpart in Ludics  of the notion of type or the logical notion of formula. Formally a behaviour is a closed set of designs.
Our aim is to explore the constructions of behaviours and to analyse their properties.
In this paper a design is viewed as a set of coherent paths. We recall or give variants of properties concerning visitable paths, where a visitable path is a path in a design or a set of designs that may be traversed by interaction with a design of the orthogonal of the set. 
We are then able to answer the following question: which properties should satisfy a set of paths for being exactly the set of visitable paths of a behaviour? Such a set and its dual should be prefix-closed, daimon-closed and satisfy two saturation properties. This allows us to have a means for defining the whole set of visitable paths of a given set of designs without closing it explicitly, that is without computing the orthogonal of this set of designs.
We finally apply all these results for making explicit the structure of a behaviour generated by constants and multiplicative/additive connectives. We end by proposing an oriented tensor for which we give basic properties.
\end{abstract}

\maketitle

\section{Introduction}
%

\subsection{Context}
At the beginning of the decade 2000, J.-Y. Girard developed a logical theory, called Ludics~\cite{DBLP:journals/mscs/Girard01}. This theory was a step in his program of Geometry of Interaction, the aim of which being to account for the dynamics of logical proofs.  
In Ludics, objects called designs keep only what is relevant for the cut elimination process, hence the dynamics of a proof: a design is an abstraction of a formal proof.
Some notable successes have been achieved with Ludics, although it has limitations. This theory succeeds in providing several new concepts and tools for manipulating logic as a theory of interaction.  It enables a new formulation of useful properties  like stability, associativity, monotonicity, traditionally required for a model of computation,   in Ludics in a unique language instead of the usual dichotomy syntax/semantics. Original concepts are also present: separation (a design is completely defined by its interacting counterdesigns), incarnation (the relevant part of a design with respect to a set of designs). Furthermore additive connectives are satisfactorily handled. One of its most important results is a full completeness theorem for second-order multiplicative additive (and affine) Linear Logic. Nevertheless, Ludics presents also limitations. A first one is that objects are rather unfamiliar and seem difficult to manipulate. A more serious limitation is due to the construction itself. The focalisation property of Linear Logic makes it possible to consider only proofs respecting a very constrained procedure, making proof search considerably easier. This enables to grasp cut elimination by means of proofs built with very few rules. The counterpart is that such proofs are strictly sequential. A still more serious limitation, when interested in computation,  is the difficulty to grasp the duplication/contraction phenomenon as Ludics is strictly linear.

Even if works in Geometry of Interaction focus more on proof nets (\cite{DBLP:journals/tcs/Girard87} as seminal paper, \cite{dicosmo_kesner_polonovski_2003,Hughes:2005:PNU:1094622.1094629,GUERRINI20111958,DBLP:conf/lics/Seiller16} for some recent ones), Ludics provides an interesting setting for applications. For example, the authors with other researchers have developed in Ludics a modelling of natural language dialogues and of several other linguistic aspects~\cite{LQwollic,Lbook,FQlacl}. 
For such purposes, the specificities of Ludics compared to other theoretical frameworks are until now essential.
In particular, Ludics is developed on  an ontological reversal, in the sense that primitive objects in Ludics are not formulas or types but their inhabitants, called designs. 
Hence, for example, non-unicity of typing is given for free and not an overlay of a logical theory.

\subsection{Ludics and Game Semantics}

The approach developed in Ludics is closely related to the Game Semantics approach where execution, \ie, interaction, occurs between player strategies in a game. Game Semantics has been extremely fruitful for studying various fragments of Linear Logic or Classical Logic  in order to obtain 
full completeness results (among others, \cite{HylandOng93, DBLP:journals/jsyml/AbramskyJ94, DBLP:conf/lics/Loader94,DBLP:conf/lics/Lamarche95, DBLP:conf/lics/AbramskyM99,Laurent10}).
Basic concepts of Ludics may in fact be expressed in terms of Game Semantics~\cite{DBLPconf/csl/Faggian02, DBLP:conf/csl/FaggianH02, faggiancurien,  DBLP:journals/corr/abs-cs-0501039,faggianbasaldella}. In a few words:
\begin{itemize}
\item an action is a {\em move}, the  abstraction of the application of a rule,
\item the sequence of actions used during interaction is a {\em play}, the cut-elimination steps,
\item a design is an {\em innocent strategy}, it is also a frame of a sequent calculus derivation.
\end{itemize}
However there is a fundamental difference between Game Semantics as it is generally used and Ludics.
Strategies are typed, while designs are {\it a priori} untyped.
More concretely, a game comes with a set of {\em plays}, \ie, sequences of {\em moves} that satisfy particular conditions, a {\em strategy} is nothing else but such a set of plays. In Ludics, a play is what results from the interaction between two designs, and a game, what denotes a formula,
 is interpreted as a behaviour, \ie, a set of designs which is closed under bi-orthogonality. 
Notice that, considered as an element of a behaviour, only part of a design may be travelled during interactions with designs in the orthogonal of the behaviour, and this part has to be considered as a strategy in terms of Game Semantics. C. Faggian~\cite{DBLPconf/csl/Faggian02,DBLP:journals/tcs/Faggian06} studied which part of a design may be travelled during an interaction. In this paper we go further by studying in which extent a design may be travelled, when this design is part of a behaviour: in that case, other designs in the behaviour constrain what can really be travelled by interaction.

\subsection{Extensions of (the original) Ludics}
Further works extended original Ludics and corrected its initial limitations, in particular strict linearity and sequentiality. 
Ludics nets were developed as a game model for concurrent interaction by F. Maurel and
C. Faggian~\cite{faggianmaurel} and more thoroughly analyzed by P.-L. Curien and C. Faggian in~\cite{faggiancurien}. In their work, Ludics nets generalize Ludics designs (or innocent strategies in Game Semantics) by allowing less sequentiality than required in the focalisation procedure. This is similar to what is at stake when transforming proof structures into proof nets.
M. Basaldella and C. Faggian in~\cite{faggianbasaldella} extended Ludics with non-linear terms. More precisely they introduce specific actions (or moves) that are neutral, hence not polarized. The interaction procedure is modified in such a way that neutral actions may be reused. They obtain a full completeness result for a variant of Multiplicative-Exponential Linear Logic.
K. Terui proposed in~\cite{terui}  a reformulation of Ludics by means of a notion of $c$-design, more suitable from a computational point of view. In particular, a $c$-design is presented in a $\lambda$-calculus style, and may not be linear. Whereas actions in Ludics (or moves) have fixed locations, \ie, are constants, actions have relative addresses in a $c$-design: variables are introduced in the model. This work was further developed with M. Basaldella in~\cite{DBLP:journals/corr/abs-1011-1625} to prove a full completeness result with respect to polarized Linear Logic. To go back to their formal framework, the forest presentation of Ludics is replaced by a linear one in such a way that the interaction rule becomes an elegant generalization of the $\beta$-reduction rule of $\lambda$-calculus. However the characterization of {\em interaction paths}, \ie, sequences of actions that may be followed in an interaction, is not as simple as in the original presentation of Ludics (see~\cite{PavauxCSL17} for a work in this direction). For that reason we stick to the original presentation of Ludics.

 \subsection{Our aim}
 We are interested in the study of {\it behaviours}.
The notion of behaviour is the counterpart in Ludics  of the notion of type or the logical notion of formula.
Typed programming language may be an important domain of application. Whereas operational semantics is often developed before defining a correct and complete type system, our purpose is to define a language for behaviours that may be used for conceiving a type system that extends traditional type systems. The operational semantics and computational rules for a language with this type system are (quite) given for free: the corresponding concept of cut-elimination is given {\it ab initio} in Ludics, not only for behaviours representing multiplicative-additive connectives but for the whole system. In this direction, we propose in this paper a definition of an oriented tensor, different from what has been done for example by M. Churchill, J. Laird and G. McCusker~\cite{DBLP:journals/apal/ChurchillLM13} in the framework of Game Semantics. 
 
 Our aim is therefore to explore the constructions of behaviours and to analyse their properties in order to define a language of behaviours.

 \subsection{Our methodology}
 
We study behaviours from two complementary approaches:
\begin{itemize}
\item Study of behaviours by means of their {\em incarnation}. 
Intuitively the incarnation of a behaviour is its greatest subset with each element fully used by interaction with counter-designs, thus sufficient for recovering the behaviour.
Furthermore, incarnation is the core of the original notion of {\em internal completeness}. Whereas a behaviour is formally defined as the closure by bi-orthogonality of a set of designs, tensor or sum of behaviours may be already defined by construction, without the necessity of bi-orthogonality.  What is at stake is to be able to study a behaviour with respect to its designs, without making explicit neither its counter-designs nor the whole set of designs of the behaviour. 
\item Study of operations on behaviours. 
  We already mentioned the internal completeness properties for additive and multiplicative connectives. However what characterizes a behaviour generated from these connectives and logical constants was not known. We prove that such a behaviour should be {\em regular}, \ie, designs in the incarnation are made of interaction paths. We hope this sheds light on the structure of such connectives. The characterization we obtain in this paper may be viewed as a generalization to the non-intuitionistic case of what was given in Game Semantics in the intuitionistic case. In particular the tensor is semantically obtained as (quite) a shuffle. This characterization provides also a means for unveiling the structure of a non-regular behaviour. What is at stake is then to specify which constructions may be defined, which ones are logically justified and enable the decomposition of all behaviours.
\end{itemize}
Both these two approaches are complementary in the sense that incarnation, and therefore internal completeness, is a guideline for defining relevant constructions on behaviours.
 
For these two approaches the concept of {\em visitable path} is central. A visitable path is a path, inside a design, which is traversed during interaction between this design and a counter design. This notion is necessarily defined with respect to a fixed behaviour. Nevertheless the definition may be relaxed using a non necessarily closed set of designs instead of a behaviour. The notion of visitable path is equivalent to the one of legal play in Game Semantics when we forget that the path is part of a design of some fixed behaviour. Among the results we obtain studying visitable paths, let us mention the following ones:
\begin{itemize}
\item Characterization of the set of visitable paths of a given set of designs without making explicit its dual.
\item Characterization of the incarnation of a behaviour generated by a set of designs without making explicit its dual. 
\item Characterization of the kind of sets of paths that are sets of visitable paths of some behaviour, without making explicit neither the behaviour nor its dual.  
\end{itemize}

\subsection{Content of the paper}

The reader may find in the annex (section~\ref{annex:basicLudics}) the original concepts of Ludics. 
In section~\ref{section:visitablePaths} we depart from this presentation recalling that a design may equivalently be viewed as a set of coherent paths. We define next what is the shuffle of paths. Finally we recall or give variants of properties concerning visitable paths, where a visitable path is a path in a design of a set of designs that may be traversed by interaction with a design of the orthogonal of the set. This last notion is important as a behaviour, where a behaviour is a closure of a set of designs with respect to interaction, is fully characterized by its set of visitable paths.

Section~\ref{section:behaviours} is devoted to answering the following question: which properties should satisfy a set of paths for being the set of visitable paths of a behaviour? Such a set of paths is said to be {\em ludicable}. A set is ludicable when this set and its dual are prefix-closed, daimon-closed and satisfy two saturation properties. We explain by means of examples why these properties are required. This allows us to have a means for defining the whole set of visitable paths of a given set of designs (not necessarily a behaviour) without computing explicitly the orthogonal of this set of designs.

In section~\ref{section:grammar}, we apply all these results for proving properties concerning behaviours obtained by means of (focalised) linear logic connectives, mainly the tensor `$\otimes$' and the plus `$\oplus$'. We can then make explicit the structure of a behaviour generated by constants and multiplicative/additive connectives. 

We begin in a last section (section~\ref{sec:SimpleOrientedTensor}) the study of non multiplicative-additive behaviours. For this purpose we consider three cases of non-commutative tensors. We prove that the one defined by J.-Y. Girard is non-commutative but may be expressed by means of the commutative tensor. We analyze the non-commutative tensor defined by Churchill et al. in terms of visitable paths. We end by defining a new non-commutative tensor and state a few properties.

\section{Ludics in Terms of Visitable Paths}\label{section:visitablePaths}
The precise definition of basic objects of Ludics is set in J.-Y. Girard's seminal paper~\cite{DBLP:journals/mscs/Girard01}. Roughly speaking,   the main objects are {\em designs} together with a notion of interaction between them.  As in Game Semantics,
the basic steps of interaction are {\em actions} (moves). In Ludics, an action is either a special one called {\em daimon}, written $\daimon$, or a proper action written $(+/-,\xi,I)$, where the {\em polarity} is relative to  the point of view one adopts (one side or the other of interaction), the  {\em focus} $\xi$ determines the position (or address) where this action may occur in an interaction. 
The daimon $\daimon$ will be considered as a positive (non-proper) action.
An action either terminates the interaction (if it is the daimon) or creates new addresses (a set $I$) on which the interaction may continue. Therefore, the (dynamics of the) interaction consists in following two dual alternate sequences of actions, one in each design.  Two kinds of sequences of actions, paths and chronicles, may equivalently be used to define designs. The latter is used by J.-Y. Girard to define a design as a set of chronicles, the former closer to the notion of {\em play}  is used in~\cite{faggianbasaldella} as an alternate presentation of Ludics, making explicit the link between Ludics and Game Semantics. More details on the notions and proofs of above results may be found in~\cite{MQ-CF-LMCS11}, in particular for the notion of paths and the computation of the incarnation of a behaviour generated from a set of designs.
\\
The reader may find in section~\ref{annex:basicLudics} complementary definitions.


\subsection{Designs as Sets of Chronicles or Sets of Paths}
 The notions of paths and chronicles are closely linked. From paths, we obtain chronicles by means of an operation of {\em view}, 
 while from chronicles we obtain paths by means of an operation of {\em shuffle}.
A design may then be viewed either as a set of paths or a set of chronicles. 
Furthermore a design has a {\em base}, \ie, the specification of initial addresses and the polarity of the actions that have such foci.
Finally an address $\xi.i$ (resp.\ an action $\kappa$ with focus $\xi.i$) is either initial if in the base or justified by the address $\xi$ (resp.\ an action $\kappa'$ of focus $\xi$ and of polarity opposite to the one of $\kappa$).

Throughout this paper, we note $\kappa^+$ a positive action, $\kappa^-$ a negative action, $\kappa$ when the polarity may be positive or negative.

\begin{defi}[Base]
A {\tt base} is  a non-empty finite set of sequents: $\Gamma_1 \vdash \Delta_1, \dots, \Gamma_n \vdash \Delta_n$ such that each $\Delta_j$ is a finite set of addresses, at most one $\Gamma_i$ may be empty and the other $\Gamma_i$ contain each exactly one address.
Furthermore if an address appears twice then one occurrence is in one of $\Gamma_i$ of a sequent and the other in one of $\Delta_j$ of another sequent, otherwise an address appears only once.
\end{defi}

An address may appear twice in a base when this address is a cut in terms of Sequent Calculus: interaction corresponds to travelling through such pairs.

\begin{defi}[Based Sequence]
A sequence of actions is based on $\beta = \Gamma_1 \vdash \Delta_1, \dots, \Gamma_n \vdash \Delta_n$ if each proper action of the sequence which is not justified by one of the previous actions in the sequence, \ie, is initial, has its focus in one of $\Gamma_i$ (resp.\ $\Delta_i$) if the action is negative (resp.\ positive).
\end{defi}

\begin{defi}[View]
Let $\pathLL{s}$ be a sequence of actions based on $\beta$, the {\tt view} $\view{\pathLL{s}}$ is the subsequence of $\pathLL{s}$ defined as follows:
\begin{itemize}
\item $\view{\epsilon}:=\epsilon$ where $\epsilon$ is the empty sequence;
\item $\view{\kappa}:=\kappa$ where $\kappa$ is an action;
\item $\view{w\kappa^+}:=\view{w}\kappa^+$ where $\kappa^+$ is a positive action;
\item $\view{w\kappa^-}:=\view{w_0}\kappa^-$ where $\kappa^-$ is a negative action and $w_0$ either is empty if $\kappa^-$ is initial or is the prefix of $w$ ending with the positive action $\kappa^+$ which justifies $\kappa^-$, \ie, the focus of $\kappa^-$ is built from the one of $\kappa^+$.
\end{itemize}
\end{defi}

\begin{defi}[Path]
A {\tt path} $\pathLL{p}$ based on $\beta = \Gamma_1 \vdash \Delta_1, \dots, \Gamma_n \vdash \Delta_n$ is a finite alternated sequence of actions based on $\beta$ such that:
\begin{itemize}
\item  Let $w\kappa^+$ be a prefix of  $\pathLL p$, if $\kappa^+$ is not initial and is justified by $\kappa^-$ then $\kappa^-\in\view{w\kappa^+}$. Roughly speaking, ``there is no view change or jump on a positive action''.
\item Actions in $\pathLL p$ have distinct foci.
\item One of the $\Gamma_i$ is empty iff $\pathLL p$ is non-empty with first action positive.
\item A daimon can only occur as the last action of $\pathLL p$.
\end{itemize}
\end{defi}

We say a path is positive-ended (resp.\ negative-ended) if its last action is positive (resp.\ negative), a path is positive (resp.\ negative) if its first action is positive (resp.\ negative). Chronicles are particular paths: namely non-empty paths such that each non-initial negative action is justified by the immediately previous (positive) action and there is at most one initial negative action.
Then, by construction, a view of a path is a chronicle. 
More generally, if we consider  the set of all prefixes $\pathLL{q}$ of a path $\pathLL{p}$, we obtain the set written $\fullview{\pathLL{p}}$ of chronicles $\view{\pathLL{q}}$ induced by $\pathLL{p}$. Conversely,  it is possible to rebuild the path $\pathLL{p}$ from the set of chronicles $\fullview{\pathLL{p}}$. The relevant operation to build paths from chronicles is the operation of shuffle. The shuffle operation may more generally be defined on paths. The standard shuffle operation consists in interleaving sequences keeping each element and respecting the order. We depart from this definition first by imposing that alternation of polarities should be satisfied, second by taking care of the daimon that should only appear at the end of a path. By this way, being a path is preserved by the shuffle operation.

\begin{defi}[Shuffle of paths]\label{shuffle}~
\begin{itemize}
\item Let $\pathLL{p}$ and $\pathLL{q}$ be two positive-ended negative paths on disjoint bases $\beta$ and $\gamma$,   and such that at least one path does   not end   on a daimon.  The {\tt shuffle} of $\pathLL{p}$ and $\pathLL{q}$, noted $\pathLL{p} \shuffle \pathLL{q}$, is the set of sequences $
\pathLL{p}_1\pathLL{q}_1\dots\pathLL{p}_n\pathLL{q}_n
$, based on $\beta \cup \gamma$ such that:\\
- each sequence $\pathLL{p}_i$ and $\pathLL{q}_i$ is either empty or a positive-ended negative path,\\
- $\pathLL{p}_1\dots\pathLL{p}_n=\pathLL{p}$ and $\pathLL{q}_1\dots\pathLL{q}_n=\pathLL{q}$,\\
- if $\pathLL{p}_n$ ends with $\daimon$ then $\pathLL{q}_n$ is empty.
\item The definition is extended to paths $\pathLL{p}\kappa_1\daimon$ and $\pathLL{q}\kappa_2\daimon$ where  $\pathLL{p}$ and $\pathLL{q}$ are two positive-ended negative paths on disjoint bases: 
$$ \pathLL{p}\kappa_1\daimon\shuffle\pathLL{q}\kappa_2\daimon
	:=
	(\pathLL{p}\kappa_1\daimon\shuffle\pathLL{q})
	\cup 
	(\pathLL{p}\shuffle\pathLL{q}\kappa_2\daimon)$$
\item The definition is extended to paths $\pathLL{r}\pathLL{p}$ and $\pathLL{r}\pathLL{q}$ where $\pathLL{r}$ is a positive-ended path and  $\pathLL{p}$ and $\pathLL{q}$ are two positive-ended negative paths on disjoint bases:  
$$ \pathLL{r}\pathLL{p}\shuffle\pathLL{r}\pathLL{q}:=\pathLL{r}(\pathLL{p}\shuffle\pathLL{q})$$
\end{itemize}
\end{defi}

It is possible to build paths from a given set of chronicles, provided that these chronicles are pairwise  coherent.
Informally, coherence  ensures that, after a common positive-ended prefix, paths are made of negative paths either on disjoint bases or with first actions of same focus. 
 Let $\pathLL{p}$ be a path, it follows from the definition of a view and the coherence relation that $\fullview{\pathLL{p}}$ is a clique of chronicles, \ie, a set of pairwise coherent chronicles. Then we may extend the coherence relation to paths: $\pathLL p$ and $\pathLL q$ are two coherent paths when $\fullview{\pathLL p} \cup \fullview{\pathLL q}$ is a clique of chronicles.

Finally a {\it design} is a clique of chronicles, or equivalently a clique of paths, satisfying supplementary conditions.
Formal definitions of coherence relation as well as design are given in annex~\ref{annex:basicLudics}.

We proved in~\cite{MQ-CF-LMCS11} that a non-empty clique of non-empty paths may give rise to a net of designs. Furthermore, when $\design{R}$ is a net of designs, the closure by shuffle of $\design{R}$, denoted  $\design{R}^{\shuffle}$ is a set of coherent paths (and we say that $\pathLL{p}$ is a positive-ended path {\em of} a net $\design{R}$ whenever $\pathLL{p}$ is in $\design{R}^{\shuffle}$).
That makes explicit the link between paths and chronicles  of a design, and more generally of a net of designs, hence justifies the switch from/to the reading of designs or nets as cliques of chronicles to/from the reading of designs or nets as cliques of paths.


\subsection{Duality and Legal Paths, Interaction}

\begin{defi}[Duality, Legality]
Let $\pathLL p$ be a positive-ended alternate sequence.
\begin{itemize}
\item The {\tt dual} of $\pathLL{p}$ (possibly empty) is the positive-ended alternate sequence of actions $\dual{\pathLL{p}}$ (possibly empty) such that\footnote{The notation $\overline{\kappa}$ is simply $\Overline{(\pm,\xi,I)} := (\mp,\xi,I)$ and may be extended on sequences by $\Overline{\epsilon} := \epsilon$ and $\Overline{w\kappa} := \Overline{w}\,\Overline{\kappa}$.}:
\begin{itemize}
\item If $\pathLL{p} = w\daimon$ then $\dual{\pathLL{p}} := \Overline{w}$.
\item Otherwise $\dual{\pathLL{p}} := \Overline{\pathLL{p}}\daimon$.
\end{itemize}
\item When $\pathLL p$ and $\dual{\pathLL p}$ are positive-ended paths, we say that $\pathLL p$ is {\tt legal}.
\end{itemize}
\end{defi}
There exist paths such that their duals are not paths, as illustrated in example~\ref{exa:paths_dualnotpath}. Nevertheless, the dual of a chronicle is a path.
\begin{exa}\label{exa:paths_dualnotpath}
Let us consider the following design. The reader may find in annex~\ref{annex:basicLudics} the way a design may be drawn as a tree of sequents. For ease of reading, in examples, we present designs as trees of sequents.
 
\noindent\begin{minipage}{.3\textwidth}
$$
\infer{\vdash \xi, \sigma}
	{
	\infer{\xi 0 \vdash \sigma}
		{
		\infer{\vdash \xi 00, \sigma}
			{
			\infer{\sigma 0 \vdash \xi 00}
				{
				\infer{\vdash \sigma 00, \xi 00}
					{
					\xi 000 \vdash \sigma 00
					}
				}
			}
		}
	&
	\infer{\xi 1 \vdash}
		{
		\infer{\vdash \xi 10}
			{
			 \xi 100 \vdash
			}
		}
	}
$$
\end{minipage}
\begin{minipage}{.70\textwidth}
The sequence $s = \scriptstyle (+,\xi,\{0,1\}) (-,\xi 0,\{0\}) (+, \sigma, \{0\}) (-,\xi 1,\{0\}) (+, \xi 10,\{0\})$ is a path based on $\vdash \xi, \sigma$. On the contrary its dual $\dual{s} = \scriptstyle (-,\xi,\{0,1\}) (+,\xi 0,\{0\}) (-, \sigma, \{0\}) (+,\xi 1,\{0\}) (-, \xi 10,\{0\}) \daimon$ is not a path based on the net $\xi \vdash ~ \sigma \vdash$: it does not satisfy the `no positive jump' condition.
\end{minipage}
\end{exa}

Interaction between two designs (or more generally two nets of designs) is at the heart of Ludics: Two designs are orthogonal when interaction converges. Interaction is realized by following a path on one design and its dual on the other, hence interaction stops when encountering a daimon. 
Recall that a path $\pathLL p$ is a path of a net of designs (that may be a single design) $\design D$ if $\pathLL p$ is in $\design D^{\shuffle}$.
We give below a definition equivalent to the seminal one~\cite{DBLP:journals/mscs/Girard01}:
\begin{defi}
Two (nets of) designs ${\design D}$ and ${\design E}$ are orthogonal if there exists a path $\pathLL p$ of ${\design D}$ such that $\dual{\pathLL p}$ is a path of ${\design E}$.
\end{defi}

The path followed by interaction may be defined as an abstract machine in the following way:
\begin{defi}[Interaction path]
Let $(\design{D},\design{R}$) be a convergent, \ie, orthogonal, closed cut-net such that all the cut loci belong to the base of $\design{D}$. The {\tt interaction path} of $\design{D}$ with $\design{R}$, denoted $\normalisationSeq{\design{D}}{\design{R}}$, is the sequence of actions of $\design{D}$ visited during the normalization. The construction goes as follows where $n$ is the number of normalization steps so far obtained:
Let $\kappa_1\dots\kappa_n$ be the prefix of $\normalisationSeq{\design{D}}{\design{R}}$ already defined (or the empty sequence if $n=0$).
\begin{itemize}
\item Either the interaction stops: if the main design is a subdesign of $\design{R}$ then $\normalisationSeq{\design{D}}{\design{R}}=\kappa_1\dots\kappa_n$, otherwise the main design is a subdesign of $\design{D}$ then $\normalisationSeq{\design{D}}{\design{R}}=\kappa_1\dots\kappa_n\daimon$.
\item Or, let $\kappa^+$ be the first proper action of the closed cut-net obtained after step $n$, $\normalisationSeq{\design{D}}{\design{R}}$ begins with $\kappa_1\dots\kappa_n\overline{\kappa^+}$ if the main design is a subdesign of $\design{R}$, or it begins with $\kappa_1\dots\kappa_n\kappa^+$ if the main design is a subdesign of $\design{D}$.
\end{itemize}
\noindent We note $\normalisationSeq{\design{R}}{\design{D}}$ the sequence of actions visited in $\design{R}$ during the normalization with $\design{D}$.
\end{defi}

It follows from the definition that $\normalisationSeq{\design D}{\design R} = \dual{\normalisationSeq{\design R}{\design D}}$.
Abusively, we also call interaction path the sequence of actions followed by interaction, even if the cut-net is not convergent, \ie, the definition is followed as long as divergence does not occur.

 The closure by bi-orthogonality of a set of designs allows to recover the notion of type, called in Ludics {\em behaviour}.
The study of these behaviours is in some aspects more graspable when interaction is defined on designs presented as cliques of paths. 
A visitable path in a behaviour, \ie, in a design of this behaviour, is a path that may be traversed by interaction. 

 \begin{defi}[Visitability]
 Let $E$ be a set of designs of same base.
Let $\pathLL p$ be a path, $\pathLL p$ is {\tt visitable} in $E$ if there exist a design $\design D$ in $E$ and a net $\design R$ in ${E}^\perp$ such that $\pathLL p = \normalisationSeq{\design{D}}{\design{R}}$. 
\\
We write $V_{E}$ the set of visitable paths of $E$.
 \end{defi}
 
Remark that a visitable path is necessarily a positive-ended path as it is the result of a normalization. The following proposition follows immediately from the definitions of visitability and orthogonality:
\begin{prop}\label{prop:dual_visitable}~\cite{MQ-CF-LMCS11}
Let $\behaviour{A}$ be a behaviour, $V_\behaviour{A} = \dual{V_{\behaviour{A}^\perp}}$.
\end{prop}


One of the complex aspects of Ludics is to be able to characterize the incarnation of a behaviour. This incarnation is in some sense the essence of a behaviour: it contains all visitable paths of the behaviour. However not all paths in the incarnation are visitable. Hence being able to identify visitable paths is a real challenge.
Conversely, there exist sets of paths that are not sets of visitable paths of behaviours. In terms of Game Semantics, a set of plays/paths is not necessarily the set of plays/visitable paths of a type/behaviour. The next section will be devoted to such a characterization.
We set here simple properties relating designs, visitable paths and incarnation of behaviours.

\begin{defi}[Completion of designs]\label{defi:closure}
Let $\design{D}$ be a design of base $\beta$, the {\tt completion} of $\design{D}$, noted $\design{D}^c$, is the design of base $\beta$ obtained from $\design{D}$ in the following way:\\
\centerline{
	$\design{D}^c := \design D \cup \{\chronicle c \kappa^-\daimon ~;~ \chronicle c \in \design D, \chronicle c \kappa^-\not\in \design D, \chronicle c\kappa^-\daimon$ is a chronicle of base $\beta\}$
}\\
Let $\design R$ be a net of designs, the {\em completion} of $\design R$ also written $\design R^c$ is the net of completions of designs of $\design R$.
\end{defi}

\noindent Note that $\design{D}^c$ is a design:
\begin{itemize}
\item An action $\kappa^-$ is either initial or justified by the last action of $\chronicle c$ (in $\chronicle c \kappa^-$) hence linearity is satisfied.
\item As chronicles $\chronicle c$ are in $\design D$ and $\kappa^-$ are negative actions, chronicles $\chronicle c \kappa^-\daimon$ are pairwise coherent and coherent with chronicles of $\design D$.
\end{itemize}

\begin{prop}\label{prop:completed_subdesign}
Let $\design{D}$ be a design in a behaviour $\behaviour{A}$, consider a design $\design{C} \subset \design{D}$ then $\design{C}^c \in \behaviour{A}$.
\end{prop}
\begin{proof}
Let $\design{E} \in \behaviour{A}^\perp$. Hence $\design{E} \perp \design{D}$. Let $\pathLL{p}$ be the longest positive-ended path in the design $\design{C}$ that is a prefix of $\normalisationSeq{\design{D}}{\design{E}}$. 
Either $\pathLL{p} = \normalisationSeq{\design{D}}{\design{E}}$, hence $\design{E} \perp \design{C}$, and also $\design{E} \perp \design{C}^c$. 
Or there exist actions $\kappa^-$, $\kappa^+$ and a sequence $w$ such that $\normalisationSeq{\design{D}}{\design{E}} = \pathLL{p}\kappa^-\kappa^+w$. 
Consider the chronicle $\chronicle{c}$ such that $\view{\pathLL{p}\kappa^-} = \chronicle{c}\kappa^-$. By construction, $\chronicle{c} \in \design{C}$. Either $\chronicle{c}\kappa^- \in \design{C}$ hence also $\chronicle{c}\kappa^-\kappa^+ \in \design{C}$ as $\design{C} \subset \design{D}$ and there is a unique positive action after a negative action, contradiction as $\pathLL{p}$ is then not maximal. 
Or $\chronicle{c}\kappa^-\daimon \in \design{C}^c$ hence $\design{E} \perp \design{C}^c$.
\end{proof}

The proposition~\ref{prop:completed_subdesign} is also true when we have nets of designs instead of designs.

\begin{prop}\label{prop:visitable-carac}
Let $E$ be a set of designs of same base, let $\pathLL p$ be a path of a design of $E$, then
$\pathLL p$ is visitable in $E$ iff $\fullview{\,\dual{\pathLL p}\,}^c \in E^{\perp}$.
\end{prop}
\begin{proof}
Suppose that $\pathLL p$ is visitable in $E$. Then there exist $\design D \in E$ and $\design R \in E^\perp$ such that $\pathLL p =  \normalisationSeq{\design{D}}{\design{R}}$. 
Furthermore $\dual{\pathLL p}$ is a path in $\design R$, hence $\fullview{\,\dual{\pathLL p}\,} \subset \design R$. It follows from proposition~\ref{prop:completed_subdesign} that $\fullview{\,\dual{\pathLL p}\,}^c \in E^\perp$. \\
Suppose that $\fullview{\,\dual{\pathLL p}\,}^c \in E^{\perp}$, let $\design D$ be the design in $E$ such that $\pathLL p$ is a path of $\design D$. Note that $\design D \perp \fullview{\,\dual{\pathLL p}\,}^c$ and that $\pathLL p = \normalisationSeq{\design{D}}{\fullview{\,\dual{\pathLL p}\,}^c}$.
Hence $\pathLL p$ is visitable in $E$.
\end{proof}

\begin{cor}\label{cor:chemvisitable}
Let $\behaviour E$ be a behaviour, let $\pathLL p$ be a path of a design of $\behaviour E$, then
$\pathLL p$ is visitable in $\behaviour E$ iff $\fullview{\pathLL p}^c \in \behaviour E$.
\end{cor}
\begin{proof}
The path $\pathLL p$ is visitable in $\behaviour E$ iff $\dual{\pathLL p}$ is visitable in ${\behaviour E}^\perp$, and we know by the previous proposition that the path 
$\dual{\pathLL p}$  is visitable in ${\behaviour E}^\perp$ iff $\fullview{\,\ddual{\pathLL p}\,}^c \in \behaviour E^{\perp\perp}=\behaviour E$.  
\end{proof}

Proposition~\ref{prop:visitable-carac} gives a means to compute the set of visitable paths of a set $E$ of designs of the same base when $E$ is a finite set of finite designs: take each positive-ended path $\pathLL{p}$ of some design of $E$, test if for all design $\design D$ in $E$ we have $\design D \perp \fullview{\dual{\pathLL p}}^c$.

\section{To Grasp Behaviours from Visitable Paths}\label{section:behaviours}
A behaviour is fully characterized by its incarnation, and this latter by the set of visitable paths of the behaviour: we proved in~\cite{MQ-CF-LMCS11} that designs in the incarnation coincide with maximal cliques of visitable paths satisfying extra conditions.

The question now is: In which extent a given set of paths may 
characterize
a behaviour? More precisely, what conditions should satisfy a  set of paths $S$ to be the set of visitable paths of some behaviour $\behaviour B_S$, $\behaviour B_S$ not given {\it a priori}? Such a set of paths will be called {\em ludicable}. Obviously the paths of $S$ should be legal. Moreover, 
$S$ has to be prefix and daimon closed. Then, maximal cliques of $S$ are relevant candidates for retrieving designs of $|\behaviour B_S|$. 
However, the extra conditions mentioned in~\cite{MQ-CF-LMCS11} assume already given the behaviour! On the contrary, we give in this paper necessary and sufficient conditions for $S$ to be ludicable without the behaviour given {\em a priori}.  
First, a path $\pathLL p\kappa^+$  obtained as a shuffle of two coherent elements of $S$ may be visited during an interaction as soon as $\pathLL p\daimon$ belongs to $S$. Such a constraint is satisfied if $S$ is {\it positively saturated}: positive saturation ensures that there are enough specific maximal cliques (called {\em positively saturated} cliques) that are really relevant candidates for building a behaviour $\behaviour B_S$.     
 Second, if   a negative-ended path  $\pathLL r$ is such that, for   every design generated by such a relevant clique, either    $\pathLL r$ is a path of this design or exits from it on a positive action, then the path $\pathLL r\daimon$ should necessarily be visited during an interaction. The condition of {\it negative saturation} expresses that such paths $\pathLL r\daimon$ should belong effectively to $S$. Last, these two constraints should be satisfied by the set $S$ and its dual $\dual S$. This guarantees that no unexpected designs may be generated by cliques of $\dual S$, hence no additional path may be visitable. 

In order to fully answer the question of ludicability, \ie, of characterizing sets of paths that give rise to behaviours (subsection~\ref{subsec:lucability}), we begin with a study of positively saturated cliques (subsection~\ref{subsec:pos_saturated_cliques}) where we prove that designs in the incarnation of a behaviour are generated by cliques that satisfy this constraint. This is a necessary first step towards a study of the two main constraints that should be satisfied for ludicability: positive and negative saturation (subsection~\ref{subsec:PosNegSaturation}). We end this section by providing a constructive way for closing a set of paths with respect to ludicability (subsection~\ref{subsec:LudicableClosures}).

Let us notice that such a question: characterizing sets of paths which are  sets of visitable paths of some behaviour seems meaningless in terms of Game Semantics. Indeed, as Game Semantics is defined, the notion of plays ontologically depends on the primitive notion of arena, which itself is relative to a given type. The question makes sense in a realisability approach, where programs and formulas come together, or even more where programs, as paths or designs in Ludics, may arise before types. Nevertheless, such a question (and its answer)  may suggest extensions in Game semantics:  in order to be able to define new types, why not consider plays  before defining  types?
Indeed, in subsection~\ref{subsec:sequoid}, we give in Ludics an elegant definition of the (linear) sequoid connective previously defined in terms of Game Semantics by Churchill et al~\cite{DBLP:journals/apal/ChurchillLM13}. An extension of our approach to exponentials is however necessary to fully address the issue.


\subsection{Positively saturated cliques}\label{subsec:pos_saturated_cliques}

We already proved in~\cite{MQ-CF-LMCS11} that a design in the incarnation of a behaviour $\behaviour B$ is a maximal clique $C$ of visitable paths such that $\dual C$ is finite-stable and saturated. 
$\dual C$ finite-stable means that if a strictly increasing sequence $(\pathLL p_n)$ of paths in $\dual C$ is such that $(\bigcup \fullview{\pathLL p_n})^c \in \behaviour B^\perp$ then this sequence $(\pathLL p_n)$ is finite.
$\dual C$ saturated means that if $\pathLL p$ is a prefix of an element of $C$ and $\pathLL p\kappa^-\daimon \in V_{\behaviour B}$ then $\pathLL p\kappa^-$ is a prefix of an element of $C$. While finite-stability is not really a constraint, since such strictly increasing sequences of paths generate designs either in a behaviour or in its dual\footnote{A more complete explanation is given in subsection~\ref{subsec:lucability}.}, saturation is not completely satisfying for our issue, as it presupposes an already given behaviour. We replace the condition of saturation on the dual of a clique by a condition of positive saturation on the clique itself. Below we first define positive saturation for cliques, then we give  an example that explains this notion. Last, we prove that this new property still enables to characterize designs of the incarnation of a behaviour.

\begin{defi}[Positively saturated clique]
Let $S$ be a set of positive-ended paths. A clique $C$ of $S$ is {\tt positively saturated for} $S$ when: for all $\pathLL m\in C$, $\pathLL n\kappa^-\kappa^+\in C$, if $\pathLL m\kappa^-\daimon\in S$ then $\pathLL m\kappa^-\kappa^+\in S$ (hence, if $C$ is a maximal clique, $\pathLL m\kappa^-\kappa^+\in C$).
 \end{defi} 

Let us give now an example that explains this notion of positive saturation.

\begin{exa}
Let $\behaviour B = \{\design D, \design E\}^{\perp\perp}$ where the designs $\design D$ and $\design E$ are given below.
$$
\design D = 
	\infer[\kappa^+]{\vdash \xi}
		{
		\infer[\kappa_1^-]{\xi1 \vdash}
			{
			\infer[\kappa_1^+]{\vdash \xi11}{\xi111 \vdash}
			}
		&
		\infer[\kappa_2^-]{\xi2 \vdash}
			{
			\infer[\kappa_{12}^+]{\vdash \xi22}{\xi221 \vdash}
			}
		&
		\infer[\kappa_3^-]{\xi3 \vdash}
			{
			\infer[\daimon]{\vdash \xi33}{}
			}
		}
~~~~~
\design E = 
	\infer[\kappa^+]{\vdash \xi}
		{
		\infer[\kappa_1^-]{\xi1 \vdash}
			{
			\infer[\daimon]{\vdash \xi11}{}
			}
		&
		\infer[\kappa_2^-]{\xi2 \vdash}
			{
			\infer[\kappa_{32}^+]{\vdash \xi22}{\xi223 \vdash}
			}
		&
		\infer[\kappa_3^-]{\xi3 \vdash}
			{
			\infer[\kappa_3^+]{\vdash \xi33}{\xi333 \vdash}
			}
		}
$$
The following paths (in particular) are visitable in $\behaviour B$:
	$\pathLL m = \kappa^+\kappa_1^-\kappa_1^+\kappa_2^-\kappa_{12}^+$
	and
	$\kappa^+\kappa_3^-\kappa_3^+$.
Let $C$ be the set containing these two paths and their positive-ended prefixes.
Then $C$ is a maximal clique of visitable paths of $\behaviour B$.
Remark that  $\pathLL m\kappa_3^-\daimon \in V_{\behaviour B}$ however $\pathLL m\kappa_3^-\kappa_3^+ \not\in V_{\behaviour B}$: $C$ is not positively saturated for $V_{\behaviour B}$. Note that $\fullview{C} \not\in \behaviour{B}$.
\end{exa}

We are able now to characterize designs in the incarnation:

\begin{prop}\label{prop:caracDessinIncarnation}
Let $\behaviour B$ be a behaviour,
$\design D$ is a design in the incarnation $|\behaviour B|$ iff $\design D = \fullview{C}$ where $C \subset V_{\behaviour B}$ is a maximal clique, positively saturated for $V_{\behaviour B}$ and such that $\dual C$ is finite-stable for $\behaviour B^\perp$.
\end{prop} 
\begin{proof}
($\Rightarrow$) We prove in~\cite{MQ-CF-LMCS11} that there exists a maximal clique $C$ of $V_{\behaviour B}$ such that $\dual C$ is finite-stable and saturated for $\behaviour B^\perp$ and $\fullview{C}=\design D$. 
Let  $\pathLL m$ and $\pathLL n\kappa^-\kappa^+$ be two elements of $C$, and suppose that  $\pathLL m\kappa^-\daimon \in V_{\behaviour B}$.
Since $\pathLL m$ is an element of $C$, and since $\pathLL m\kappa^-\daimon \in V_{\behaviour B}$, $\pathLL m\kappa^-$ is a prefix of an element of $C$ as $\dual C$ is saturated, \ie, there is a path $\pathLL m\kappa^-\kappa_0^+\pathLL w \in C$, hence $\pathLL m\kappa^-\kappa_0^+ \in C$ as $C$ is a maximal clique and $V_{\behaviour B}$ is prefix-closed. Since  $\pathLL m\kappa^-\kappa^+$ is a path of $\fullview C$, the only possibility is $\kappa^+ = \kappa_0^+$ as $C$ is a clique, \ie, $\pathLL m\kappa^-\kappa^+\in C$.

($\Leftarrow$) Let $\design E \in \behaviour B^\perp$. We consider the interaction between $\design D$ and $\design E$. First we prove that positive-ended prefixes of $\normalisationSeq{\design D}{\design E}$ are elements of $C$. 
If $C = \{\daimon\}$ the result is immediate. 
Otherwise note that if $\design D$ is a positive design then its first (hence positive) action is an element of $C$ (as $C$ is a clique), if $\design D$ is a negative design then the empty sequence is an element of $C$. 
Thus $\design D$ being of positive or negative base, there exists a prefix of $\normalisationSeq{\design D}{\design E}$ that is an element of $C$. Suppose that $\pathLL p$ is a prefix of an element of $C$ whereas $\pathLL p\kappa$ is not a prefix of an element of $C$ and $\pathLL p\kappa$ is a prefix of $\normalisationSeq{\design D}{\design E}$:
		\begin{itemize}
		\item Either $\kappa$ is a positive action. The path $\pathLL p$ is negative-ended. The path $\pathLL p$ cannot be empty otherwise $\kappa$, as a prefix of $\normalisationSeq{\design D}{\design E}$, is the first action of $\design D$ hence also of paths of $C$, contradicting the hypothesis. Then $\pathLL p$ is of the form $\pathLL p'\kappa^-$. As $\pathLL p$ is an element of $C$ and elements of $C$ are positive-ended paths, there exists $\pathLL q \in C$ of the form $\pathLL p'\kappa^-\kappa^+\pathLL p''$. As $C$ is a clique and as normalization is deterministic, $\kappa^+ = \kappa$. Thus $\pathLL p'\kappa^-\kappa \in C$, contradiction.
		\item Or $\kappa$ is a negative action. Note that $\fullview{\pathLL p}^c \in \behaviour B$ as $\pathLL p \in V_{\behaviour B}$ and that $\pathLL p\kappa^-\daimon = \normalisationSeq{\fullview{\pathLL p}^c}{\design E}$ thus $\pathLL p\kappa^-\daimon \in V_{\behaviour B}$. Furthermore $\kappa^- = \kappa$ as normalization is deterministic. As $\kappa^-$ appears in $\fullview{C}$, there exists a path $\pathLL n\kappa^-\kappa'^+ \in C$. Hence as $C$ is positively saturated for $V_{\behaviour B}$, $\pathLL p\kappa^-\kappa'^+ \in V_{\behaviour B}$. 
Hence $\pathLL p\kappa^-\kappa'^+ \in C$ as $C$ is a maximal clique.
		\end{itemize}
We prove now by contradiction that $\design D \perp \design E$ analysing each possible case. Suppose $\design D \not\perp \design E$:
		\begin{itemize}
		\item Either normalization goes on infinitely: there is a strictly increasing sequence $(\pathLL p_n)$ of $C$ such that 
$\pathLL p_n$ is a path of $\design D$ and $\overline{\pathLL p_n}$ is a path of $\design E$.
In particular for all $n$ $\fullview{\dual{p_n}} \subset \design E$. Then as $\dual C$ is finite-stable, the sequence $(\pathLL p_n)$ is finite. Contradiction.
		\item Or 
if $\pathLL p\kappa^+$ is a path of $\design D$, \ie, $\pathLL p\kappa^+ \in C$, $\overline{\pathLL p}$ is a path of $\design E$ but $\overline{\pathLL p\kappa^+}$ is not a path of $\design E$.
We have that $\pathLL p\kappa^+ \in V_{\behaviour B}$, thus $\fullview{\pathLL p\kappa^+}^c \in \behaviour B$ (corollary~\ref{cor:chemvisitable}). Thus contradiction as $\design E \in \behaviour B^\perp$.
		\item Or 
if $\overline{\pathLL p\kappa^-}$ is a path of $\design E$, $\pathLL p$ is a path of $\design D$ but $\pathLL p\kappa^-$ is not a path of $\design D$.
Hence in particular $\kappa^-$ is not present in $\design D$.
Furthermore $\pathLL p$ is a (positive-ended) prefix of $\normalisationSeq{\design D}{\design E}$ thus $\pathLL p \in C$ as a consequence of the first part of this proof. Then $\fullview{\pathLL p}^c \in \behaviour B$ (corollary~\ref{cor:chemvisitable}).
But $\pathLL p\kappa^-\daimon = \normalisationSeq{\fullview{\pathLL p}^c}{\design E}$ thus $\pathLL p\kappa^-\daimon \in V_{\behaviour B}$. 
Thus $\pathLL p\kappa^-\daimon$ could be added to $C$, contradiction with the fact that $C$ is a maximal clique.
		\end{itemize}
Hence $\design D \in \behaviour B$. Suppose that $\design D \not\in |\behaviour B|$, then there exists a design $\design E \in \behaviour B$ such that $\design E \subsetneq \design D$, \ie, there exists a path $\pathLL p\kappa^-\kappa^+ \in C$ such that $\pathLL p$ is a path in $\design E$, $\pathLL p\kappa^-\kappa^+$ is not a path in $\design E$. As $\pathLL p\kappa^-\kappa^+ \in C$, there exists a design $\design F \in \behaviour B^\perp$ such that $\dual{\pathLL p\kappa^-\kappa^+}$ is a path in $\design F$. Hence $\design E \not\perp \design F$, contradiction.
\end{proof}


\subsection{Positive Saturation, Negative Saturation of a set of paths}\label{subsec:PosNegSaturation}Let us go
back to the initial question: let $S$ be a set of legal paths, is $S$ the set of visitable paths of a behaviour? We know from proposition~\ref{prop:caracDessinIncarnation} that we should base our analysis on the set of positively saturated maximal cliques of $S$. As suggested in example~\ref{exa:PosNegSaturation}, each path in $S$ should be in some positively saturated maximal clique of $S$. Another condition, called {\em negative saturation}, should also be satisfied: a path $\pathLL p\kappa^-\daimon$ contained in the design generated by some positively saturated maximal clique should be in $S$ as soon as the design $\fullview{\overline{\pathLL p\kappa^-}}$ is orthogonal to all positively saturated maximal clique of $S$ (see example~\ref{exa:PosNegSaturation}). Before that, we recall the following points: let $\pathLL m\pathLL n$ be a visitable path with $\pathLL m$ positive-ended then $\pathLL m$ is also a visitable path, let $\pathLL p\kappa^+$ be a visitable path then $\pathLL p\daimon$ is also a visitable path. Hence a set $S$ should be prefix and daimon closed:

\begin{defi}[Prefix closure, daimon closure]
Let $S$ be a set of legal paths of same base:
	\begin{itemize}
	\item $S$ is {\tt prefix-closed} if all positive-ended prefix of an element of $S$ is an element of $S$.
	\item $S$ is {\tt daimon-closed} if for all path $\pathLL{p}\kappa^+ \in S$, $\pathLL{p}\daimon \in S$.
	\end{itemize}
\end{defi}

\begin{exa}\label{exa:PosNegSaturation}
Let us first remark that each path in $S$ should be in some positively saturated maximal clique of $S$.
Let $S$ be the prefix and daimon closure of the two paths $\kappa^+\kappa_1^-\kappa_1^+\kappa_2^-\kappa_2^+$ and $\kappa^+\kappa_2^-\daimon$. We only have the following maximal cliques of $S$:
	\begin{itemize}
	\item $C_0 = \{\daimon\}$
	\item $C_1 = \{\kappa^+, \kappa^+\kappa_1^-\daimon, \kappa^+\kappa_2^-\daimon\}$
	\item $C_2 = \{\kappa^+, \kappa^+\kappa_1^-\kappa_1^+,\kappa^+\kappa_1^-\kappa_1^+\kappa_2^-\daimon,\kappa^+\kappa_2^-\daimon\}$
	\item $C_3 = \{\kappa^+, \kappa^+\kappa_1^-\kappa_1^+,\kappa^+\kappa_1^-\kappa_1^+\kappa_2^-\kappa_2^+\}$
	\end{itemize}
Note that $C_3$ is not positively saturated: $\kappa^+ \in C_3$, $\kappa^+\kappa_1^-\kappa_1^+\kappa_2^-\kappa_2^+ \in C_3$, $\kappa^+\kappa_2^-\daimon \in S$ but $\kappa^+\kappa_2^-\kappa_2^+ \not\in S$. Thus there is no positively saturated maximal clique containing the path $\kappa^+\kappa_1^-\kappa_1^+\kappa_2^-\kappa_2^+$. By the way $S$ is not the set of visitable paths of some behaviour.

Let us now motivate the need for negative saturation. Let $S$ be the set of paths $\{\daimon, \kappa^+, \kappa^+\kappa_1^-\kappa_1^+, \kappa^+\kappa_1^-\daimon, \kappa^+\kappa_2^-\daimon\}$: $S$ is daimon and positive prefix closed. We have the following maximal cliques of $S$:
	\begin{itemize}
	\item $C_0 = \{\daimon\}$
	\item $C_1 = \{\kappa^+, \kappa^+\kappa_1^-\daimon, \kappa^+\kappa_2^-\daimon\}$
	\item $C_2 = \{\kappa^+, \kappa^+\kappa_1^-\kappa_1^+, \kappa^+\kappa_2^-\daimon\}$
	\end{itemize}
Notice that these three cliques are positively saturated.
However the design $\fullview{C_2}$ contains the path $\kappa^+\kappa_1^-\kappa_1^+\kappa_2^-\daimon$.
Remark that the design $\fullview{\overline{\kappa^+\kappa_1^-\kappa_1^+\kappa_2^-}}$ is orthogonal to all three designs $\fullview{C_i}$ but the path $\kappa^+\kappa_1^-\kappa_1^+\kappa_2^-\daimon$ is not in $S$.
Here too, $S$ is not the set of visitable paths of some behaviour.
\end{exa}

\begin{defi}[Positive saturation condition]
 A set $S$ of legal paths of same base is {\tt positively saturated} if for each path $\pathLL p\in S$, there exists a positively saturated maximal clique $C$ of $S$ such that $\pathLL p\in C$.
\end{defi}

In the following proposition, we remark that, if it exists, there is a standard positively saturated clique. This clique $C_{\pathLL p}$ is in fact the smallest positively saturated maximal clique containing a given path $\pathLL p$.

 \begin{prop}\label{prop:Cp}
Let $S$ be a set of legal paths of same base, prefix and daimon closed.
The three conditions are equivalent:
\begin{itemize}
\item[$i)$] $S$ is positively saturated.
\item[$ii)$] For each path $\pathLL p\in S$, the clique $C_{\pathLL p}$ defined below is positively saturated for $S$:
$$C_{\pathLL p}=\{ \pathLL q ~;~ \pathLL q \in S, \fullview{\pathLL q}\subset\fullview{\pathLL p}\}\cup\{\pathLL w\kappa^-\daimon ~;~ \pathLL w\kappa^-\daimon\in S,  \fullview{\pathLL w}\subset\fullview{\pathLL p}, \fullview{\pathLL w\kappa^-}\not\subset\fullview{\pathLL p}\}$$
\item[$iii)$] For each path $\pathLL p\in S$, there exists a positively saturated maximal clique $C$ of $S$ without infinite strictly increasing sequence and such that $\pathLL p\in C$.
\end{itemize}
 \end{prop}
 
\begin{proof}
$i) \Rightarrow ii)$ Suppose that there exists a path $\pathLL p\in S$ such that $C_{\pathLL p}$ is not positively saturated. 
This means that there exist two paths  $\pathLL m\in C_{\pathLL p}$ and  $\pathLL n\kappa^-\kappa^+\in C_{\pathLL p}$, such that $\pathLL m\kappa^-\daimon\in S$ but  $\pathLL m\kappa^-\kappa^+\notin S$.
Note that $\kappa^+ \neq \daimon$. Note also that $\pathLL m$ cannot end with a daimon, otherwise $\pathLL m\kappa^-\kappa^+$ is not a path. Therefore $\fullview{\pathLL m}\subset\fullview{\pathLL p}$ and $\fullview{\pathLL n\kappa^-\kappa^+}\subset\fullview{\pathLL p}$.
As $S$ is positively saturated, there exists a positively saturated maximal clique $C$ containing $\pathLL p$, hence also $\pathLL m$ and  $\pathLL n\kappa^-\kappa^+$ thus $\pathLL m\kappa^-\kappa^+ \in S$. Contradiction.

$ii) \Rightarrow iii)$ Let us remark that the clique $C_{\pathLL p}$ is a maximal clique of $S$. Indeed, let $\pathLL m$ be an element of $S$ which is coherent with $\pathLL p$. Either $\fullview{\pathLL m}\subset\fullview{\pathLL p}$ then $\pathLL m\in C_{\pathLL p}$, or there exists a prefix $\pathLL{m}_0\kappa^-$ of $\pathLL m$ such that $\fullview{\pathLL{m}_0}\subset\fullview{\pathLL p}$ and $\kappa^-$ does not occur in $\pathLL p$. In such a case, $\pathLL{m}_0\kappa^-\daimon\in C_{\pathLL p}$ and either $\pathLL m=\pathLL{m}_0\kappa^-\daimon$ which belongs to $C_{\pathLL p}$, or $\pathLL m\not\coh\pathLL{m}_0\kappa^-\daimon$, \ie, $\pathLL m\not\coh C_{\pathLL p}$. 
Finally the path $\pathLL p$ has a finite length, hence a path in $\fullview{\pathLL p}$ has a length bounded by the length of $\pathLL p$ (no possibility of using twice an action by linearity condition). A path of $C_{\pathLL p}$ is either a path in $\fullview{\pathLL p}$, hence its length is bounded by the length of $\pathLL p$, or extended by a sequence of the form $\kappa^-\daimon$ as $\daimon$ should end a path, thus its length is bounded by the length of $\pathLL p$ plus $2$.
Thus there cannot be infinite strictly increasing sequences in $C_{\pathLL p}$.

$iii) \Rightarrow i)$ This result is obvious.
\end{proof}

Let us now define the negative saturation condition.
The negative saturation condition ensures that, as soon as a path $\pathLL p$ belongs to $S$, and $\overline{\pathLL p\kappa^-}$ generates a design which is orthogonal with all positively satured maximal cliques of $S$, then $\pathLL p\kappa^-\daimon\in S$.
 
\begin{defi}[Negative saturation condition]
Let $S$ be a set of legal paths of same base, prefix-closed and daimon-closed.
\\
$S$ satisfies {\tt negative saturation} if for all path $\pathLL p \in S$ such that $\pathLL p\kappa^-\daimon$ is a legal path and for all positively saturated maximal clique $C$ of $S$ we have that $\fullview{C} \perp \fullview{\overline{\pathLL p\kappa^-}}^c$,
then the path $\pathLL p\kappa^-\daimon\in S$.
\end{defi}

The following lemma gives rise to equivalent formulations for the negative saturation condition.

\begin{lem}\label{lem:negativesat}
 Let $S$ be a set of legal paths of same base, prefix-closed and daimon-closed. Let $\pathLL p$ be a path belonging to $S$, $\pathLL p\kappa^-\daimon$ be a legal path. Then, the three following conditions are equivalent:
	\begin{itemize}
	\item[$i)$] For all positively saturated maximal clique $C$ of $S$, we have $\fullview C\perp\fullview{\overline{\pathLL p\kappa^-}}^c$.
	\item[$ii)$] For all positively saturated maximal clique $C$ of $S$ without infinite strictly increasing sequence, we have $\fullview C\perp\fullview{\overline{\pathLL p\kappa^-}}^c$.
	\item[$iii)$] For all positively saturated maximal clique $C$ of $S$, for all legal path $\pathLL m$ in the design $\fullview C$, for all negative action $\kappa_0^-$ such that $\overline{\pathLL m\kappa_0^-}$ is a path in $\fullview{\overline{\pathLL p\kappa^-}}$ then $\pathLL m\kappa_0^-$ is a path of $\fullview C$. 
	\end{itemize}
 \end{lem}
 
 \begin{proof} 
$i)\Rightarrow ii)$ The result is obvious.

$ii)\Rightarrow i)$ Let $C$ be a positively saturated maximal clique of $S$ and suppose that $\fullview C\not\perp\fullview{\overline{\pathLL p\kappa^-}}^c$. As paths in $\fullview{\overline{\pathLL p\kappa^-}}^c$ have finite length, divergence occurs in a finite number of steps, say $n$. Because $\fullview{\overline{\pathLL p\kappa^-}}^c$ has all negative actions, we know that the action that causes divergence is a positive one in $\fullview{\overline{\pathLL p\kappa^-}}^c$, \ie, the dual negative action is not available in $C$.
Let us define a design $\design D$ to be the set of chronicles $\chronicle c\kappa'^- \daimon$ as soon as $\chronicle c$ has length $n+1$ and there exist actions $\kappa'^-$ and $\kappa'^+$ such that $\chronicle c\kappa'^-\kappa'^+ \in \fullview C$. Let $C'$ be the set of paths $\pathLL q$ such that $\pathLL q$ is a path in $\design D$ and either $\pathLL q\in C$ or $\pathLL q = \pathLL q_0\daimon$ and there exists a positive action $\kappa''^+$ such that $\pathLL q_0\kappa''^+ \in C$. Remark that $C'$ is a clique as all paths are in a design, and that $C'$ is maximal and positively saturated as $C$ is maximal and positively saturated. Hence by hypothesis $\fullview{C'}\perp\fullview{\overline{\pathLL p\kappa^-}}^c$: $\normalisationSeq{\fullview{C'}}{\fullview{\overline{\pathLL p\kappa^-}}^c} = \pathLL m\kappa_0^-\pathLL m'$. Hence contradiction as $\pathLL m\kappa_0^-$ should also be the prefix of a path in $\fullview C$. 

$iii)\Rightarrow i)$ Let $C$ be a positively saturated maximal clique of $S$.  Suppose that $\fullview C\not\perp\fullview{\overline{\pathLL p\kappa^-}}^c$. Let $\pathLL m$ be the longest path in $\fullview C$ such that $\overline{\pathLL m}$ is a path in $\fullview{\overline{\pathLL p\kappa^-}}^c$ (as $\pathLL p$ has a finite length, paths in $\fullview{\overline{\pathLL p\kappa^-}}^c$ have also a finite length, thus $\pathLL m$ is well defined). The path $\pathLL m$ is necessarily daimon-free otherwise the interaction converges. As $\fullview{\overline{\pathLL p\kappa^-}}^c$ is complete with respect to negative action, there is a negative action $\kappa_0^-$ such that $\pathLL m\kappa_0^-$ is not a path in $\fullview C$ whereas $\overline{\pathLL m\kappa_0^-}$ is a path in $\fullview{\overline{\pathLL p\kappa^-}}^c$. In that case we should have in fact $\overline{\pathLL m\kappa_0^-}$ to be a path in $\fullview{\overline{\pathLL p\kappa^-}}$. Thus $\pathLL m\kappa_0^-$ is a path in $\fullview C$ (because of hypothesis $iii$). Contradiction. Thus $\fullview C\perp\fullview{\overline{\pathLL p\kappa^-}}^c$.

$i)\Rightarrow iii)$ Let $C$ be a positive saturated maximal clique of $S$, let $\pathLL m$ be a legal path in the design $\fullview C$, let $\kappa_0^-$ be a negative action such that $\overline{\pathLL m\kappa_0^-}$ is a path in $\fullview{\overline{\pathLL p\kappa^-}}$. Then, since the interaction path between $\fullview C$ and $\fullview{\overline{\pathLL p\kappa^-}}^c$ begins with the path $\pathLL m$, it should go on with $\overline{\kappa_0^-}$ in $\fullview{\overline{\pathLL p\kappa^-}}^c$, then $\pathLL m\kappa_0^-$ is a path in $\fullview{C}$. 
 \end{proof}


\subsection{Ludicable sets}\label{subsec:lucability}
We are now ready for fully characterizing a behaviour in terms of visitable paths:
\begin{itemize}
\item The set of visitable paths of a behaviour is {\em ludicable} (proposition~\ref{prop:visitableTOludicable}).
\item A {\em ludicable} set of paths is the set of visitable paths of some behaviour (proposition~\ref{prop:ludicableTObehaviour}).
\end{itemize}

\begin{defi}[Ludicable set]
A set $S$ of legal paths of same base is {\tt pre-ludicable} when $S$ is prefix-closed, daimon-closed, positively and negatively saturated.
\\
A set $S$ is {\tt ludicable} if $S$ and $\dual S$ are pre-ludicable.
\end{defi}


\begin{prop}\label{prop:visitableTOludicable}
Let $\behaviour E$ be a behaviour, then its set of visitable paths $V_{\behaviour E}$ is ludicable.
\end{prop}
%
\proof
 We just need to check that $V_{\behaviour E}$ is pre-ludicable: When $\behaviour E$ is a behaviour we have  $\dual{V_{\behaviour E}}=V_{{\behaviour E}^\perp}$, hence $\dual{V_{\behaviour E}}$ is also pre-ludicable.
\begin{itemize}
\item {\it $ V_{\behaviour E}$ is both daimon and positive prefix closed.} This is a direct consequence of the following fact: When a design $\design{D}$ is obtained from a design ${\design E}$ by replacing chronicles $\chronicle c\kappa^+\chronicle{c'}$ of $\design E$ by the chronicle $\chronicle c\daimon$ then ${\design E}^\perp\subset{\design D}^\perp$.
\item {\it Positive saturation.} Let $\pathLL p$ be a visitable path of $V_{\behaviour E}$. Then the design $\fullview{\pathLL p}^c$ belongs to $\behaviour E$. Let us consider the  clique 
 $$C_{\pathLL p}=
	\{ \pathLL q ~;~ \pathLL q \in V_{\behaviour E}, \fullview{\pathLL q}\subset\fullview{\pathLL p}\}
	\cup
	\{\pathLL w\kappa^-\daimon ~;~ \pathLL w\kappa^-\daimon\in V_{\behaviour E}, \fullview{\pathLL w}\subset\fullview{\pathLL p}, \fullview{\pathLL w\kappa^-}\not\subset\fullview{\pathLL p}\}$$
We remark that $C_{\pathLL p}$ is a maximal clique of $V_{\behaviour E}$. Furthermore the design $\fullview{C_{\pathLL p}}$ is the incarnation of $\fullview{\pathLL p}^c$ with respect to $\behaviour E$,  hence $\fullview{C_{\pathLL p}}$ belongs to $\behaviour E$. 
\\
Let  $\pathLL m$ and $\pathLL n\kappa^-\kappa^+$ be two paths belonging to  $C_{\pathLL p}$, and suppose that $\pathLL m\kappa^-\daimon\in V_{\behaviour E}$, so $\fullview{\overline{m\kappa^-}}^c\in {\behaviour E}^\perp$. 
Note that the path $\pathLL m\kappa^-$ should be a prefix of $\normalisationSeq{\fullview{C_{\pathLL p}}}{\fullview{\overline{\pathLL m\kappa^-}}^c}$ and that $\fullview{C_{\pathLL p}} \perp \fullview{\overline{\pathLL m\kappa^-}}^c$. Hence there exists a positive action $\kappa^+$ such that $\pathLL m\kappa^-\kappa^+$ is a prefix of $\normalisationSeq{\fullview{C_{\pathLL p}}}{\fullview{\overline{\pathLL m\kappa^-}}^c}$.
In other words, $\pathLL m\kappa^-\kappa^+\in V_{\behaviour E}$.
\\
Conclusion: $C_{\pathLL p}$  is positively saturated.
\item {\it Negative saturation.} Suppose that $\pathLL p\kappa^-\daimon$ is a legal path such that $\pathLL p \in V_{\behaviour E}$ and for all positively saturated maximal clique $C$ of $V_{\behaviour E}$ we have that $\fullview{C} \perp \fullview{\overline{\pathLL p\kappa^-}}^c$. 
Let $\design D \in |\behaviour E|$, it follows from proposition~\ref{prop:caracDessinIncarnation} that there exists a positively saturated maximal clique $C$ of $V_{\behaviour E}$ such that $\fullview C =  \design D$, therefore $\fullview{\overline{\pathLL p\kappa^-}}^c \perp \design D$. Thus $\fullview{\overline{\pathLL p\kappa^-}}^c \in \behaviour E^\perp$. Finally $\fullview{\pathLL p}^c \in \behaviour E$ and $\pathLL p\kappa^-\daimon = \normalisationSeq{\fullview{\pathLL p}^c}{\fullview{\overline{\pathLL p\kappa^-}}^c}$, thus $\pathLL p\kappa^-\daimon \in V_{\behaviour E}$.
\qed
\end{itemize}


Our issue now is to prove that if a set $S$ of paths is ludicable then it is the set of visitable paths of some behaviour $\behaviour A$. 
Previously, it is important to notice that such a behaviour $\behaviour A$ may not be unique. In some cases, in particular when $S$ is finite, the incarnation of $\behaviour A$ is exactly defined from maximal cliques of $S$, hence $\behaviour A$ is unique. However, when $S$ is infinite and, more precisely when it contains at least one infinite strictly increasing sequence of paths, several choices of behaviour are possible.
\begin{exa}\label{example:ludicableBehaviour}
Let us consider for example a set $S$ generated by an infinite strictly increasing sequence $\chronicle c_0 = \kappa_0^+$ and $\chronicle c_{i} = \chronicle c_{i-1}\kappa_i^-\kappa_i^+$, for $i \geq 1$, where each action except $\kappa_0^+$ is justified by the previous one: $\kappa_i^+$ is justified by $\kappa_i^-$ and $\kappa_{i+1}^-$ is justified by $\kappa_i^+$. Formally,
$$
S= \{\daimon\} \cup \{ \chronicle c_i ~;~ i\geq 0\} \cup \{ \chronicle c_{i-1}\kappa_i^-\daimon ~;~ i \geq 1\}
$$
Note that the elements of $S$ are in fact chronicles.
Let us consider the two following behaviours:
$$
\begin{array}{rcl}
\behaviour A_1 &=& \{ \fullview C ~;~ C\mbox{ is a maximal clique of } S\}^{\perp\perp}\\
\behaviour A_2 &=& \{ \fullview D ~;~ D\mbox{ is a maximal clique of } \dual S\}^{\perp}
\end{array}
$$
It is worth noticing that $S$ is ludicable and that $S = V_{\behaviour{A}_1} = V_{\behaviour{A}_2}$, although $\behaviour{A}_1 \neq \behaviour{A}_2$.
The difference between $\behaviour{A}_1$ and $\behaviour{A}_2$ is that $\bigcup_{i} \chronicle c_i$ is a design of $\behaviour{A}_1$ but not of $\behaviour{A}_2$. Whereas $\behaviour{A}_2$ contains only designs generated by bounded maximal cliques of $S$ of the form $\{ \chronicle c_i ~;~ i\leq N\} \cup \{\chronicle c_N\kappa_N^-\daimon\}$. On the opposite, $\bigcup_{i}\overline{ \chronicle c_{i}}$ belongs to ${\behaviour{A}_2}^\perp$.
\end{exa}


\begin{prop}\label{prop:ludicableTObehaviour}
Let $S$ be a non-empty set of legal paths of same base, if $S$ is ludicable then there exists a behaviour $\behaviour E$ such that $S = V_{\behaviour E}$.
\end{prop}

\begin{proof}
We set:\\
- $E=\{\fullview{C} ~;~ C$ is a positively saturated maximal clique of $S\}$, \\
- $F=\{\fullview{D} ~;~ D$ is a positively saturated maximal clique of $\dual S$ with no infinite strictly increasing sequences of paths$\}$.\\
We shall prove that $S = V_{F^{\perp}}$. There may be different choices for $E$ and $F$ with the same result as it may be deduced from example~\ref{example:ludicableBehaviour}: for each infinite strictly increasing sequence of $S$, either $E$ or $F$ should have a restriction. \\
Note that with conditions as stated on $E$ and $F$, normalization between a design of $E$ and a design of $F$ cannot continue infinitely: either it converges or it stops because a positive action in one side has no negative counterpart on the other side.
 
The proof sketch is as follows:
\begin{itemize}
\item[i)] We prove first that $E \subset F^\perp$ (hence $F\subset F^{\perp\perp}\subset E^\perp$). 
\item[ii)] Next we prove that $S \subset V_{F^\perp}$.
\item[iii)] Then, we prove that $V_{F^\perp}\subset S$.
\item[iv)] We conclude from ii) and iii) that $S = V_{F^\perp}$ and the fact that $F^\perp$ is a behaviour.
 \end{itemize}

\noindent i) We prove that $E\subset {F}^\perp$.
We will show that the interaction between a design of $E$ and a design of   $F$ does not diverge. Therefore, since by construction  the interaction cannot go on infinitely, it converges, so $E\subset {F}^\perp$. 
Namely, let $\fullview C \in E$ and $\fullview D \in F$, we show by induction on the length of $\pathLL p$, a strict prefix of the interaction path of $\fullview C$ with $\fullview D$, that:\\
{\it  either $\pathLL p\daimon \in C$, or $\overline{\pathLL p}\daimon \in D$, or there exists a proper action $\kappa$ such that $\pathLL{p\kappa}$ is a prefix of a path of $C$ and $\overline{\pathLL{p\kappa}}$ is a prefix of a path of $D$}.
\\
It follows not only that the interaction path continues but it remains a prefix of an element of $S$ as $\pathLL p$ (resp. $\overline{\pathLL p}$) is a prefix of an element of $C$ (resp. of $D$).
	\begin{itemize}
\item The path $\pathLL p$ is the empty sequence. Either $\fullview C$ or $\fullview D$ is the set $\{\daimon\}$ hence the property. Or suppose $\fullview C$ is a positive design with first (proper) action $\kappa^+$, then $\kappa^+ \in C$ ($C$ is positive-prefix closed) thus 
$\overline{\kappa^+}\daimon \in \dual S$. As $D$ is a maximal clique of $\dual S$, there should exist a path beginning with the action $\overline{\kappa^+}$, \eg, there exists a positive action $\kappa_0^+$ such that $\overline{\kappa^+}\kappa_0^+ \in D$.
And the property is satisfied. If $\fullview C$ is a negative design, then $\fullview D$ is a positive design and the same reasoning applies exchanging $C$ and $D$.

	\item   
Let $\pathLL p = \pathLL p_0\kappa^+$ and $\pathLL p_0$ satisfies the property.
Note that $\pathLL p_0$ is negative-ended and that $\pathLL p_0\daimon \not\in C$ otherwise $\pathLL p$ is not a strict prefix as $\normalisationSeq{\fullview{C}}{\fullview{D}} = \pathLL p_0\daimon$.
Thus there exists a proper action $\kappa_0^+$ such that $\pathLL p_0\kappa_0^+$ is a prefix of a path of $C$ (in fact $\pathLL p_0\kappa_0^+ \in C$) and $\overline{\pathLL p_0\kappa_0^+}$ is a prefix of a path of $D$. As $C$ is a clique and $\pathLL p$ is a prefix of $\normalisationSeq{\fullview{C}}{\fullview{D}}$, we should have $\kappa_0^+ = \kappa^+$. As $\overline{\pathLL p_0\kappa_0^+}$ is a prefix of a path of $D$ and $\overline{\pathLL p_0\kappa_0^+}$ is negative-ended and $D$ is a clique of paths, there should exist a positive action $\kappa_1^+$ such that $\overline{\pathLL p_0\kappa_0^+}\kappa_1^+$ is a prefix of a path of $D$.
If $\kappa_1^+ = \daimon$, the property is satisfied. 
Otherwise remark that $\overline{\pathLL p_0\kappa_0^+}\kappa_1^+ \in \dual S$, hence $\pathLL p_0\kappa_0^+\overline{\kappa_1^+}$ is a prefix of an element of $S$.
As $\pathLL p_0\kappa_0^+ \in C$ and $C$ is a maximal clique of $S$, $\pathLL p_0\kappa_0^+\overline{\kappa_1^+}$ should be a prefix of an element of $C$. Hence the result.

	\item The case of  $\pathLL p$ negative-ended is symmetrical.
	\end{itemize}

\noindent ii) We prove that {\bf $S\subset V_{F^\perp}$}.
Let $\pathLL p$ be an element of $S$, 
let $C$ be a positively saturated maximal clique of $S$ containing $\pathLL p$, 
let $D$ be a positively saturated maximal clique of $\dual S$ containing $\dual{\pathLL p}$ and with no infinite strictly increasing sequences of paths. 
Remark that $\normalisationSeq{\fullview C}{\fullview D} = \pathLL p$. 
Remark also that $\fullview{C} \in E$ and $\fullview{D} \in F$, and, since $E \subset {F}^\perp$,  $\fullview{C}$   belongs to ${F}^\perp$.
Hence the path $\pathLL p$ is the interaction path between a design of $F$ and a design of ${F}^\perp$, \ie, $\pathLL p \in V_{F^\perp}$.

\noindent iii) We prove now that $V_{F^\perp}\subset S$. 
 Let $\pathLL q$ be a visitable path of $F^\perp$ and let us denote by $\pathLL p$ the longest prefix of $\pathLL q$ that is a prefix of an element of $S$. If $\pathLL q$ is empty then the base of $F^\perp$ hence of $S$ is negative then the empty sequence is in $S$ as $S$ is prefix closed. We suppose now that $\pathLL q$ is not empty. Notice that $\pathLL p$ cannot be empty. Indeed, if $\pathLL q=\daimon$, then $\pathLL p = \daimon \in S$ as $S$ is daimon-closed, otherwise $\pathLL q=\kappa\pathLL{q'}$ and either $\kappa$ is positive  and  $\kappa\in S$ or $\kappa$ is negative and $\overline{\kappa}\in\dual S$, \ie, $\kappa\daimon \in S$. \\
Suppose that $\pathLL p$ is a strict prefix of $\pathLL q$, this means that  there is an action $\kappa$ such that $\pathLL q=\pathLL{p\kappa}\pathLL{p'}$  and neither $\pathLL p\kappa$  belongs to $S$ nor $\overline{\pathLL p\kappa}$ belongs to $\dual S$. 
 \begin{enumerate}
 \item If $\kappa$ is negative. By hypothesis, $\overline{\pathLL p}\daimon\in\dual S$. Since  $\pathLL{p\kappa}\pathLL{p'}$  is visitable in ${F^\perp}$, so is $\pathLL{p\kappa}\daimon$. This means that there exist an element $\fullview{D_0}$ in $F$ and a design ${\design D_0}\in{F^\perp}$ such that $\pathLL p\kappa\daimon= \normalisationSeq{\design D_0}{\fullview{D_0}}$. Then we have $\overline{\pathLL p\kappa}$ is a path in $\fullview{D_0}$. 
This is only possible when there exists a path $\pathLL m\overline{\kappa} \in D_0\subset\dual S$.
Therefore, by positive saturation of $D_0$, $\overline{\pathLL p\kappa} \in \dual S$, hence $\pathLL p\kappa \in S$ against the hypothesis: $\pathLL p$ is not the longest prefix of $\pathLL q$ that is a prefix of an element of $S$.

\item If $\kappa$ is positive. By hypothesis $\overline{\pathLL p} \in \dual S$. 
Since  $\pathLL{p \kappa}$ is visitable in ${F^\perp}$ and $F^\perp$ is a behaviour, the design $\fullview{\pathLL p\kappa}^c \in F^\perp$.
Therefore, for all positively saturated maximal clique without infinite strictly increasing sequence $D$ of $\dual S$, 
  $\fullview{\pathLL p\kappa}^c\perp\fullview D$. Then, since $\dual S$ satisfies the negative saturation condition (and with lemma~\ref{lem:negativesat}), $\overline{\pathLL p\kappa}\daimon\in \dual S$, that is $\pathLL p\kappa \in S$, against the hypothesis that $\pathLL p$ is the longest prefix of $\pathLL q$ that is a prefix of an element of $S$.
 
   \end{enumerate}

We conclude from ii) and iii) that $S = V_{F^\perp}$ and the fact that $F^\perp$ is a behaviour.
%
\end{proof}


 \subsection{Ludicable closures}\label{subsec:LudicableClosures}
In the previous subsection, propositions~\ref{prop:visitableTOludicable} and~\ref{prop:ludicableTObehaviour} give necessary and sufficient conditions for relating sets of paths and behaviours. In this subsection we give a constructive means for closing a set of paths $S$, \ie, for obtaining the minimal ludicable set containing $S$.

\begin{defi}[Ludicable Closure]
Let $S$ be a set of legal paths of the same base, we define the pre-ludicable closure, written $S'$, of $S$ to be the smallest set of paths including $S$ that is prefix-closed, daimon-closed, and that satisfies negative saturation and positive saturation.\\
The {\tt ludicable closure} of $S$, written $LC(S)$, is $\bigcup_n S_n$ such that the family $(S_n)$ is defined inductively in the following way:
	\begin{itemize}
	\item $S_0 = S$.
	\item $S_{n+1} = S_n \cup S_n'\cup \dddual[.9]{({\dddual[.9]{S_n}})'}$
	\end{itemize}
\end{defi}

\begin{lem}\label{lem:LC_properties}
Let $S$ be a set of legal paths of the same base, $LC(S)$ is well-defined and pre-ludicable.
\end{lem}
\begin{proof}
We show that the function $f: S \rightarrow S \cup S'\cup \dddual[.9]{({\dddual[.9]{S}})'}$, which domain $\mathcal D$ is the set of sets of legal paths of same base as $S$ and including $S$, is Scott-continuous. Hence, Kleene fixed-point theorem applies: $LC(S)$ is well-defined and is the least fixed-point of $f$, \ie, $LC(S)$ is pre-ludicable.
\\
We have that $S \subset f(S)$. 
It is also immediate that $\mathcal D$ is a complete partial order for the inclusion of sets.
Let $\mathcal E$ be a directed subset of $\mathcal D$, note that $\bigvee_{S \in \mathcal E} S' = (\bigvee_{S \in \mathcal E} S)'$ and that $\bigvee_{S \in \mathcal E} \dual S = \dual{(\bigvee_{S \in \mathcal E} S)}$. Hence also $\bigvee_{S \in \mathcal E} f(S) = f(\bigvee_{S \in \mathcal E} S)$.
\end{proof}

\begin{prop}
A ludicable closure is a ludicable set of paths.
\end{prop}
\begin{proof}
We know with lemma~\ref{lem:LC_properties} that $LC(S)$ is pre-ludicable.
Let $T = \dual{S}$ and $(T_n)$ be the family defining $LC(T)$. We prove by induction on $n$ that $T_n = \dual {S_n}$. By definition $T_0 = \dual S = \dual {S_0}$. Suppose that $T_n = \dual {S_n}$, then 
		\begin{eqnarray*}
		T_{n+1} &=& T_n \cup T_n'\cup \dddual[.9]{{\dddual[.9]{T_n}}'}	\\
			&=& \dual {S_n} \cup \dual {S_n}' \cup \dddual[.9]{\dddual[.9]{\dddual[.9]{S_n}}'} \\
			&=& \dual {S_n} \cup \dddual[.9]{\dddual[.9]{\dddual[.9]{S_n}'}} \cup \dual{S_n'}	\\
			&=& \dual{ S_n \cup \dddual[.9]{\dddual[.9]{S_n}'} \cup S_n'} \\
			&=& \dual{S_{n+1}}
		\end{eqnarray*}
It follows that $LC(T) = \bigcup_n T_n = \bigcup_n \dual{S_n} = \dual{\bigcup_n S_n} = \dual{LC(S)}$ thus $\dual{LC(S)}$ is also pre-ludicable.
So $LC(S)$ is ludicable.
\end{proof}

\begin{lem}\label{lem:cliqueMacDeLC}
Let $E$ be a set of designs of same base, let $C$ be a positively saturated maximal clique of $LC(V_E)$ such that if there is an infinite increasing sequence of paths $(\pathLL p_i)$ in $C$ then $\bigcup \pathLL p_i$ is already a path in a design of $E$, then $\fullview{C} \in E^{\perp\perp}$.
\end{lem}

\begin{proof}
Let $\design D\in |E^\perp|$.
\begin{enumerate}
\item The interaction path between $\fullview C$ and the design $\design D$ cannot be infinite:
For each infinite increasing sequence of paths $(\pathLL p_i)$ in $C$, the path $\bigcup \pathLL p_i$ is already included in a design of $E$, hence there is an index $i_0$ such that $\fullview{\overline{\pathLL{p_{i_0}}}\daimon} \subset \design D$.

\item We prove by contradiction that  the interaction between $\fullview C$ and $\design D$ cannot diverge finitely.
 Let $\pathLL p$ be the longest interaction sequence of $\fullview{C}$ with $\design D$. Note that $\pathLL p$ cannot end with a daimon otherwise the two designs are orthogonal.
Let $\pathLL q$ be the longest prefix of $\pathLL p$ such that either $\pathLL q\daimon$ or $\pathLL q$ is an element of $LC(V_E)$ (depending on the polarity of $\pathLL q$). Note that $\pathLL q$ is not the empty sequence.
\begin{itemize}
\item If $\overline{\pathLL q}\kappa^+$ is a path in $\design D$. $\kappa^+$ cannot be the daimon otherwise the two designs are orthogonal. Then there exists a path $\pathLL n\kappa^+ \in \dual{V_{E}}$ that is a path of $\design D$. Thus $\pathLL n\kappa^+ \in \dual{LC(V_E)}$. Furthermore we can write $\overline{\pathLL q} = \pathLL m\kappa^-$ and $\pathLL m\kappa^-\daimon \in \dual{LC(V_E)}$ by hypothesis. Finally $\pathLL m \coh \pathLL n \kappa^+$. Thus, as $\dual{LC(V_E)}$ satisfies positive saturation, $\overline{\pathLL q}\kappa^+ = \pathLL m\kappa^-\kappa^+ \in \dual{LC(V_E)}$, hence $\pathLL q\overline{\kappa^+}\daimon \in LC(V_E)$. Contradiction.
\item If $\pathLL q\kappa^+$ is a path in $\fullview{C}$. Thus $\pathLL q\daimon \in LC(V_E)$. Furthermore there exists a path $\pathLL n\kappa^+ \in LC(V_E)$ in the design $\fullview C$. We can write $\pathLL q = \pathLL m \kappa^-$, thus as $LC(V_E)$ is prefix-closed, $\pathLL m \in LC(V_E)$. Finally, $\pathLL m \coh \pathLL n\kappa^+$ (as they are paths of the same design $\fullview C$). So, by positive saturation of $LC(V_E)$, we have $\pathLL q\kappa^+ \in LC(V_E)$. Contradiction.
\end{itemize}
\end{enumerate}
Thus $\fullview C\perp\design D$. Hence $\fullview C \in E^{\perp\perp}$.
\end{proof}

\begin{prop}
Let $E$ be a set of designs of the same base, let $\behaviour A$ be the behaviour generated by $LC(V_E)$ and such that, for each increasing sequence $(\pathLL{p}_i)$ of paths belonging to some positively saturated maximal clique $C$ of $LC(V_E)$, if $\bigcup\pathLL{p}_i$ is already a path in $E$ then $\fullview C\in \behaviour A$, otherwise  $\fullview{\dual C}\in {\behaviour A}^\perp$. 
Then the behaviour $\behaviour A$ is equal to ${E}^{\perp\perp}$. 
\end{prop}

\begin{proof}
The lemma~\ref{lem:cliqueMacDeLC} enables to affirm that $\behaviour A\subset E^{\perp\perp}$.
On the other side, if $\design{D}$ belongs to $|E|$, then $\design D\in\behaviour A$. Indeed, the set of paths of $\design D$  which are visitable in $E$ is a maximal clique of $V_E$. And this maximal clique of $V_E$ gives rise to a positively saturated maximal clique of $LC(V_E)$ (the new paths are only the one obtained by positive saturation and are already some paths in $\design D$). 
\end{proof}

\section{A Characterization of MALL behaviours}\label{section:grammar}

This section is devoted to a characterization of MALL behaviours in terms of their visitable paths. A MALL behaviour is built from constants $\bf 1$, $\bf 0$, $\boldtop$, $\boldbot$ and connectives multiplicative tensor `$\otimes$', additive sum `$\oplus$', shift `$\bigshpos$' and their duals. We first recall definitions and properties of these connectives. We give an example showing that there exist behaviours that are not the denotation of a MALL behaviour. We prove then that a MALL behaviour satisfies two criteria: {\em finiteness}, \ie, the incarnation has a finite number of designs and these designs are finite, and {\em regularity}, \ie, a legal path built from actions in designs of the incarnation should be visitable.

\subsection{Linear Connectives}
Before establishing properties concerning the behaviours which may be associated with linear formulas, 
we recall below the main linear operations on behaviours: multiplicative tensor `$\otimes$', additive sum `$\oplus$' and also the shift `$\bigshpos$' operation. The shift operation is required as the logic is polarized: it allows for switching from/to a positive behaviour to/from a negative behaviour. Dual operations are defined in a standard way: $\behaviour{A}\parr \behaviour{B}= (\behaviour{A}\otimes\behaviour{B})^\perp$, $\behaviour{A}\with\behaviour{B} = (\behaviour{A}\oplus\behaviour{B})^\perp$ and $\bigshneg \behaviour{A}= (\bigshpos \behaviour{A})^\perp$.

\begin{defi}~
\begin{itemize}
\item Let $(\behaviour{G}_k)$ be a family of positive behaviours pairwise disjoint, 
$\bigoplus_k \behaviour{G}_k = (\bigcup_k \behaviour{G}_k)^{\perp\perp}$
\item Let $\design{A}$ and $\design{B}$ be two positive alien designs\footnote{Two positive designs $(+,\xi,I)\design A$ and $(+,\xi,J)\design B$ are {\em alien} when $I \cap J = \emptyset$. A positive design $\design A$ and the design $\{\daimon\}$ are alien. Two positive behaviours are alien when their designs are pairwise alien.}:\\
- If $\design{A}$ or $\design{B}$ is $\{\daimon\}$, then $\design{A} \otimes \design{B} = \{\daimon\}$.\\
- Otherwise $\design{A} = \scalebox{.9}{$(+,\xi,I)$}\design{A}'$ and $\design{B} = \scalebox{.9}{$(+,\xi,J)$}\design{B}'$
then $\design{A} \otimes \design{B} = \scalebox{.9}{$(+,\xi,I \cup J)$}(\design{A}' \cup \design{B}')$.
\item Let $\behaviour{G}$ and $\behaviour{H}$ be two positive alien behaviours,
$\behaviour{G} \otimes \behaviour{H} = \{\design{A} \otimes \design{B} ~;~ \design{A} \in \behaviour{G}, \design{B} \in \behaviour{H}\}^{\perp\perp}$
\item Let $\behaviour G$ be a negative behaviour of base $\xi i\vdash$, $\bigshpos \behaviour{G} = ((+,\xi,\{i\}) \behaviour G)^{\perp\perp}$.
\end{itemize}
\end{defi}

\begin{thm}[internal completeness~\cite{DBLP:journals/mscs/Girard01}]~
\begin{itemize}
\item Let $K \neq \emptyset$, ~~~~ $\bigoplus_{k\in K} \behaviour{G}_k = \bigcup_{k \in K} \behaviour{G}_k$
\item A behaviour of positive base is always decomposable as a $\bigoplus$ of connected behaviours.
\item Let $\behaviour{G}$ and $\behaviour{H}$ be two alien positive behaviours, $\behaviour{G} \otimes \behaviour{H} = \{\design{A} \otimes \design{B} ~;~ \design{A} \in \behaviour{G}, \design{B} \in \behaviour{H}\}$.
\item Let $\behaviour G$ be a negative behaviour of base $\xi i\vdash$, $\bigshpos \behaviour{G} = \{\{\daimon\}\} \cup (+,\xi,\{i\}) \behaviour G$.
\end{itemize}
\end{thm}

To be complete, let us recall the definition of multiplicative and additive Ludics constants (on base $\vdash \xi$ for $\bf 1$):
\begin{defi}
$\bf 1 = \{\{\daimon\}, (+,\xi,\emptyset)\}$,
$\boldbot = \bf 1^\perp$,
$\bf 0 = \{\{\daimon\}\}$,
$\boldtop = \bf 0^\perp$.
\end{defi}

We intend to relate MALL logical connectives, in fact `$\otimes$', `$\oplus$', `$\bigshpos$', with operations on visitable paths. Obtaining the set of visitable paths for `$\oplus$' and `$\bigshpos$' is quite immediate. The behaviour $\behaviour{A} \oplus \behaviour{B}$ is the union of the two behaviours $\behaviour A$ and $\behaviour B$, hence visitable paths of the result should be the union of the two sets of visitable paths.
The behaviour $\bigshpos \behaviour{A}$ is built by adding to each design the same action at the root, hence visitable paths of the resulting behaviour should be built by adding this action as prefix of visitable paths of $\behaviour{A}$. As the operation `$\otimes$' models a kind of concurrency, it is natural to consider that the set of visitable paths of a tensor should be in some way the shuffle of the sets of visitable paths. Without surprise, it is  necessary to consider only {\em legal} paths among the ones obtained by shuffle of visitable paths.
However, being a legal path in the shuffle is not sufficient as shown in example~\ref{example-nonregular}.
\begin{exa}\label{example-nonregular}
Let us consider the two behaviours $\behaviour A$ which incarnation is the daimon-closure of the set $\{\design A_1, \design A_2\}$ and $\behaviour B$ which incarnation is the daimon-closure of the set $\{\design B_1, \design B_2\}$:
\begin{center}
\begin{tabular}{cccc}
$\design A_1 = \!\!\!\!\!\!
\scalebox{.8}{
\infer[\gamma_0^+]{\vdash \langle\rangle}{
\infer[\kappa_1^-]{\alpha \vdash}
	{
	\infer[\kappa_1^+]{\vdash \alpha1, \alpha2}
		{
		\infer[\kappa_2^-]{\alpha11 \vdash \alpha2}
			{
			\infer[\kappa_2^+]{\vdash \alpha111, \alpha2}
				{
				\infer[\kappa_3^-]{\alpha22 \vdash \alpha111}
					{
					\infer[\daimon]{\vdash \alpha111, \alpha222}
						{
						}
					}
				}
			}
		}
	}
}
}
$
&
$\design A_2 = \!\!\!\!\!\!
\scalebox{.8}{
\infer[\gamma_0^+]{\vdash \langle\rangle}{
\infer[\kappa_1^-]{\alpha \vdash}
	{
	\infer[\kappa_2^+]{\vdash \alpha1, \alpha2}
		{
		\infer[\kappa_3^-]{\alpha22 \vdash \alpha1}
			{
			\infer[\kappa_1^+]{\vdash \alpha1, \alpha222}
				{
				\infer[\kappa_2^-]{\alpha11 \vdash \alpha222}
					{
					\infer[\daimon]{\vdash \alpha111, \alpha222}
						{
						}
					}
				}
			}
		}
	}
}
}
$
&
$\design B_1 =\!\!\!\!\!\!\!\!\!\!\!\!\!\!\!\!
\scalebox{.8}{
\infer[\gamma_1^+]{\vdash \langle\rangle}{
\infer[\lambda_0^-]{\beta \vdash}
	{
	\infer[\lambda_0^+]{\vdash \beta0}
		{
		\infer[\lambda_1^-]{\beta01 \vdash}
			{
			\infer[\lambda_1^+]{\vdash \beta011}
				{
				\beta0111 \vdash
				}
			}
		&
		\infer[\lambda_2^-]{\beta02 \vdash}
			{
			\infer[\lambda_2^+]{\vdash \beta022}
				{
				\beta0222 \vdash
				}
			}
		}
	}
}
}
$
&
$\design B_2 =\!\!\!\!\!\!\!\!\!\!\!\!\!\!
\scalebox{.8}{
\infer[\gamma_1^+]{\vdash \langle\rangle}{
\infer[\lambda_0^-]{\beta \vdash}
	{
	\infer[\lambda_0^+]{\vdash \beta0}
		{
		\infer[\lambda_1^-]{\beta01 \vdash}
			{
			\infer[\daimon]{\vdash \beta011}
				{
				}
			}
		&
		\beta02 \vdash
		}
	}
}
}
$
\end{tabular}
\end{center}
Let us consider the three following paths (where $\gamma^+ =(+,\langle\rangle,\{\alpha,\beta\})$):\\
\begin{minipage}{.6\textwidth}
\begin{itemize}
	\item $\pathLL p = \kappa_1^-\kappa_1^+\kappa_2^-\kappa_2^+$
	\item $\pathLL q = \lambda_0^-\lambda_0^+\lambda_1^-\lambda_1^+\lambda_2^-\lambda_2^+$
	\item $\pathLL r = \gamma^+\lambda_0^-\lambda_0^+\kappa_1^-\kappa_1^+\lambda_1^-\lambda_1^+\kappa_2^-\kappa_2^+\lambda_2^-\lambda_2^+$
\end{itemize}
Remark that $\gamma_0^+\pathLL p \in V_{\behaviour A}$ and $\gamma_1^+\pathLL q \in V_{\behaviour B}$.
Furthermore $\pathLL r \in \gamma^+(\pathLL p \shuffle \pathLL q)$ is a legal path.
However $\pathLL r \not\in V_{\behaviour A \otimes \behaviour B}$.
Indeed $\fullview{\,\dual{r}\,}^c \not\in (\behaviour A \otimes \behaviour B)^\perp$
as $\design A_2 \otimes \design B_2 \not\perp \fullview{\,\dual{r}\,}^c$.
\end{minipage}
\begin{minipage}{.4\textwidth}
\begin{center}
$
\fullview{\,\dual{r}\,} = \!\!\!\!\!\!\!\!\!\!\!\!\!\!\!\!\!\!\!\!\!\!\!\!
\scalebox{.8}{
\infer[\gamma^+]{\vdash \langle\rangle}
{
\infer[\overline{\lambda_0^-}]{\vdash \alpha, \beta}
	{
	\infer[\overline{\lambda_0^+}]{\beta0 \vdash \alpha}
		{
		\infer[\overline{\kappa_1^-}]{\vdash \alpha, \beta01, \beta02}
			{
			\infer[\overline{\kappa_1^+}]{\alpha1 \vdash \beta01}
				{
				\infer[\overline{\lambda_1^-}]{\vdash \alpha11, \beta01}
					{
					\infer[\overline{\lambda_1^+}]{\beta011 \vdash \alpha11}
						{
						\infer[\overline{\kappa_2^-}]{\vdash \alpha11, \beta0111}
							{
							\alpha111 \vdash \beta0111
							}
						}
					}
				}
			&
			\infer[\overline{\kappa_2^+}]{\alpha2 \vdash \beta02}
				{
				\infer[\overline{\lambda_2^-}]{\vdash \alpha22, \beta02}
					{
					\infer[\overline{\lambda_2^+}]{\beta022 \vdash \alpha22}
						{
						\infer[\daimon]{\vdash \alpha22, \beta0222}
							{
							}
						}
					}
				}
			}
		}
	}
}
}
$
\end{center}
\end{minipage}

\end{exa}

In example~\ref{example-nonregular}, the behaviour $\behaviour B$ is not decomposable by means of linear connectives. We present in the following section a very simple characterization of behaviours which may be recursively decomposable by means of linear connectives. For such an issue we introduce the notion of {\em regular} behaviour.

\subsection{Regular Behaviours}\label{sec:reg_beh}

We prove in this subsection that {\em finiteness} and {\em regularity} as defined later characterize a behaviour for being the denotation of a MALL formula.

\begin{defi}[Trivial view/chronicle]
A legal path which is equal to its view is a chronicle. When a chronicle (view) $\chronicle c$ is such that $\dual{\chronicle c}$ is also a chronicle, or, in other words, each action of $\chronicle c$ is justified by the immediate previous one, except the first one which is initial, this chronicle is  said to be {\tt trivial}.
\\
If $w\kappa w'$ is a path with $\kappa$ a proper action, the {\tt trivial chronicle of $\kappa$ for $w\kappa w'$} is a trivial chronicle with last action $\kappa$ that is a subsequence of $w\kappa$.
\end{defi}

Note that an action occurs only once in a path hence there is only one trivial chronicle of a given proper action that is a subsequence of this path.

\begin{prop}\label{prop:trivialchronique}
Let $V$ be a set of legal paths and $\dual V$ be its dual.  Let $w\kappa w' \in V$ and $\chronicle c$ be the trivial chronicle of $\kappa$ for $w\kappa w'$. If $V$ and $\dual V$ are both prefix, daimon and view-closed, then either $\chronicle c \in V$ or $\chronicle c\daimon \in V$. 
\end{prop}

\proof
The proof is done with respect to $\kappa$.
\begin{itemize}
	\item Suppose $\kappa$ initial. 
	Either $\kappa$ is negative, hence $\chronicle c = \kappa$ and $w = \epsilon$. As $\kappa w' \in V$ and $V$ daimon-closed then $\chronicle c\daimon = \kappa \daimon \in V$.
	Or $\kappa$ is positive hence $w\kappa\in V$ as $V$ is prefix-closed. Therefore $\dual{w\kappa}\in\dual V$. Since $\overline{\kappa}$ is a negative initial action, $\view{\overline{w\kappa}\daimon}=\overline{\kappa}\daimon\in\dual V$.  Thus $\chronicle c = \kappa \in V$.

	\item Suppose $\kappa$ non initial and  positive, hence $w\kappa\in V$ as $V$ is prefix-closed. Therefore $\dual{w\kappa}\in\dual V$. Since $\overline{\kappa}$ is a negative action and $\dual V$ is view-closed, $\view{\overline{w\kappa}\daimon}=\view{w_0}\overline{\kappa}\daimon\in\dual V$ where $w_0$ is the prefix of $w$ ending on the justifier $\kappa_1$ of $\kappa$. 
Therefore, we may write $\view{\overline{w\kappa}\daimon}= w_1\kappa_1\overline{\kappa}\daimon\in\dual V$. And since $\overline{w_1\kappa_1\overline{\kappa}}\in V$, we have that $\view{\overline{w_1\kappa_1}\kappa}\in V$.  Since $\overline{\kappa_1}$ is a negative action, $\view{\overline{w_1\kappa_1}\kappa}=w_2\kappa_2\overline{\kappa_1}\kappa\in V$ where $w_2\kappa_2$ is the view of the prefix of $w_1$ ending on the justifier $\kappa_2$ of $\kappa_1$. The chain of justifiers of $\kappa$ is finite, therefore we obtain in a finite number of steps that the trivial chronicle $\chronicle c$ of $\kappa$ belongs to $V$. 

	\item Suppose $\kappa$ non initial and negative, hence $\overline{w\kappa}\in \dual V$ as $\dual V$ is prefix-closed. The same reasoning as in the previous item applies exchanging $V$ and $\dual V$. Thus the result.
\qed
\end{itemize}

In the following we consider mainly legal paths. For that purpose we introduce a specific definition.

\begin{defi}[Shuffle of two sets of legal paths]
Let $V$ and $S$ be two sets of legal paths respectively based on $\sigma\vdash$ and $\tau\vdash$. The set $V\lshuffle S$ contains the legal paths $\pathLL r$ such that there are some paths $\pathLL{p}_V\in V$ and $\pathLL{p}_S\in S$ and $\pathLL r\in\pathLL{p}_V\shuffle\pathLL{p}_S$. 

The definition is extended to  couples of sets of legal paths $V$ and $S$, both  based on $\vdash\xi$, by considering  the set of  legal paths $(+,\xi,I \cup J)\pathLL r$ such that there are $(+,\xi,I)\pathLL{p}_V\in V$ and $(+,\xi,J)\pathLL{p}_S\in S$ and $\pathLL r\in\pathLL{p}_V\shuffle\pathLL{p}_S$ (when $I \cap J = \emptyset$). 
\end{defi}

\noindent Remark: the notation $\lshuffle$ is a remainder for a shuffle ($\shuffle$) limited to legal ($l$) paths.

\begin{defi}[Regular set of legal paths]
Let $V$ be a set of legal paths and $\dual V$ be its dual.  The data $V/\dual V$ is a {\tt regular data} when $V$ and $\dual V$ are both sets of legal paths, prefix closed, daimon closed, view-closed  and moreover, both are shuffle closed, that is:\\
 if $\pathLL p$ and $\pathLL q$ belong to $V$ (resp. $\dual V$), then $\pathLL p\lshuffle\pathLL q \subset V$ (resp. $\dual V$).
\end{defi}

We prove in propositions~\ref{tousleschemins} and~\ref{touschemins} that a data $V/\dual V$ is regular iff {\em all} legal paths made of actions occurring in $V$ (resp.\ $\dual V$) are in $V$ (resp.\ $\dual V$).
These two propositions concern the non-additive case: a data $V/\dual V$ is {\em non-additive} if there is no pair of actions occurring in $V/\dual V$ with same focus but different ramifications, \ie, if $(+/-,\xi,I) \in V$ and $(+/-,\xi,J) \in V$ then $I = J$. In such a case, a focus occurs only in one action (that may be present several times in paths of $V$ or $\dual V$). 
It follows also that a trivial chronicle of an action does no more depend on a path: it is unique being given a data $V/\dual V$.
We show after these propositions that this may be generalized to the additive case. We think that presenting the propositions in the non-additive case allows a more explicit reading of what regularity implies. 

 \begin{prop}\label{tousleschemins}
 Let $V/\dual V$ be a non-additive regular data. If $\pathLL p$ is a legal path containing only actions occurring in $V$ then $\pathLL p\in V$.
 \end{prop}

\proof
If $\pathLL p$ ends  with $\daimon$ then $\dual{\pathLL p}$ is a $\daimon$-free legal path containing only actions occurring in $\dual V$. Hence, without loss of generality we can suppose that $\pathLL p$ is $\daimon$-free.

\noindent Let us consider first that $\pathLL p$ is a chronicle  containing only actions occurring in $V$. We prove the result by induction on the number of positive actions which are not justified by their immediate previous negative action in $\pathLL p$.
	\begin{itemize}
		\item If $\pathLL p $ is a trivial chronicle we already know by proposition~\ref{prop:trivialchronique} that $\pathLL p\in V$.
		\item Otherwise, we suppose that if  the chronicle $\pathLL p$ is made of actions of $V$ and contains at most  $n$ positive actions which are not justified by their immediate previous negative action, then $\pathLL p \in V$. Suppose that $\pathLL p$ is a chronicle made of actions of $V$ and contains $n+1$  positive actions which are not justified by their immediate previous negative action. Let $\kappa$ be the last such positive action occurring in $\pathLL p$, that is, $\pathLL p= \pathLL{c}\kappa\pathLL{d}$ where $\kappa\pathLL{d}$ is a trivial chronicle.\\
			-- Either $\kappa$ is initial. We may write  $\dual{\pathLL  p}=\overline{\pathLL{c}\kappa\pathLL{d}}\daimon\in ( \overline{\pathLL{c}}\shuffle\overline{\kappa\pathLL{d}}\daimon )$. 
Note that $\overline{\pathLL{c}}$ and $\overline{\kappa\pathLL{d}}\daimon$ are legal paths.
Since $\overline{\kappa\pathLL{d}}\daimon$ is the dual of a trivial chronicle, it belongs to $\dual V$; by induction hypothesis, $\overline{\pathLL{c}}\in\dual V$, therefore, by shuffle closure,  $\dual{\pathLL{p}}\in\dual V$ and thus $\pathLL{p}\in V$.\\
			-- Or $\kappa$ is justified by an action $\kappa_0$, and we may write  $\pathLL  p=\pathLL{c}_0\kappa_0\pathLL{c}_1\kappa\pathLL{d}$, where $\pathLL{c}_0\kappa_0\pathLL{c}_1$ and $\pathLL{c}_0\kappa_0\kappa\pathLL{d}$ are two chronicles having at most $n$ positive actions which are not justified by their immediate previous negative action. That $\pathLL{c}_0\kappa_0\kappa\pathLL{d}$ is indeed a chronicle is due to the fact that in $\kappa\pathLL{d}$ all actions are immediately justifed by the immediate previous one. Then, by induction hypothesis,  $\overline{\pathLL{c}_0\kappa_0\pathLL{c}_1} \in \dual V$ and $\pathLL{c}_0\kappa_0\kappa\pathLL{d} \in V$, thus $\overline{\pathLL{c}_0\kappa_0\kappa\pathLL{d}}\daimon \in \dual V$.
		It follows that $\dual{\pathLL  p}=\overline{\pathLL{c}_0\kappa_0\pathLL{c}_1\kappa\pathLL{d}}\daimon\in  \overline{\pathLL{c}_0\kappa}_0(\overline{\pathLL{c}_1}\shuffle\overline{\kappa\pathLL{d}}\daimon )$. 
		Furthermore $\overline{\pathLL{c}_0\kappa_0\pathLL{c}_1}$ and $\overline{\pathLL{c}_0\kappa_0\kappa\pathLL{d}}\daimon$ are legal paths.
		By shuffle closure, $\dual{\pathLL p}\in\dual V$ and $\pathLL{p}\in V$.
	\end{itemize}
	
\noindent Let us consider now  that $\pathLL p$ is a legal path  containing only actions occurring in $V$. We prove by induction on the number of negative actions which are not justified by their immediate previous positive action (we call the number of `negative jumps') that $\pathLL p\in V$.
	\begin{itemize}
		\item If $\pathLL p $ is a  chronicle we know already that $\pathLL p\in V$.

		\item Otherwise, we suppose that: if  the legal path $\pathLL q$ is made of actions of $V$ and contains at most  $n$ negative jumps, then $\pathLL q\in V$. Suppose that $\pathLL p$ is a legal path  made of actions of $V$ and that $\pathLL p$ contains $n+1$  negative jumps. Let $\kappa$ be the last such negative action involving a negative jump  in $\pathLL p$, that is, $\pathLL p= \pathLL{w}\kappa\pathLL{c}$ where $\kappa\pathLL{c}$ is a chronicle.\\
			-- Either $\kappa$ is initial, and we may write  $\pathLL  p=\pathLL{w}\kappa\pathLL{c} \in ( \pathLL{w} \shuffle \kappa\pathLL{c} )$. Since $\kappa\pathLL{c}$ is a chronicle, it belongs to $\dual V$; furthermore $\pathLL{w}$ is a legal path hence, by induction hypothesis, $\pathLL{w} \in V$. Therefore, by shuffle closure,  $\pathLL{p} \in V$.\\
			-- Or $\kappa$ is justified by an action $\kappa_0$, and we may write  $\pathLL  p=\pathLL{w}_0\kappa_0\pathLL{w}_1\kappa\pathLL{c} $. let us observe that $\pathLL{w}_0\kappa_0\kappa\pathLL{c}$ is a legal path: 
	(i) $\pathLL{w}_0\kappa_0$ is a legal path, 
	(ii) $\pathLL{w}_0\kappa_0\kappa\pathLL{c}$ is a path as any positive action occurring in $\pathLL{c}$ cannot be justified in $\pathLL{w}_1$ otherwise $\pathLL p$ could not be a path,
	(iii) $\pathLL{w}_0\kappa_0\kappa\pathLL{c}$ is legal since $\kappa_0\kappa\pathLL{c}$ is a chronicle. Then we may apply the induction hypothesis, \ie, $\pathLL{w}_0\kappa_0\pathLL{w}_1\in V$ and   $\pathLL{w}_0\kappa_0\kappa\pathLL{c} \in V$.\\
Finally, $\pathLL  p = \pathLL{w}_0\kappa_0\pathLL{w}_1\kappa\pathLL{c} \in \pathLL{w}_0\kappa_0(\pathLL{w}_1 \shuffle \kappa\pathLL{c})$. Therefore, by shuffle closure, $\pathLL{p}\in V$.
\qed
	\end{itemize}

That a data $V/\dual V$ contains all legal paths written from a set of justified actions is not only a necessary condition to be regular, but it is also sufficient as stated in the following proposition.

\begin{prop}\label{touschemins}
Let $V$ be a set of legal paths of base $\beta$ such that $V/\dual V$ is non-additive. 
Let us say that a set of legal paths of base $\beta$ is {\em complete} if it contains all legal paths of base $\beta$ which may be written from actions in its paths. 
If $V$ and $\dual V$ are complete and non-empty, then $V/\dual V$ is a regular data.
\end{prop}
\begin{proof}
Note that, if $V$ is complete, $V$ is {\em exactly} the set of legal paths of base $\beta$ which may be written from actions in paths of $V$.
Let $\pathLL p \in V$ then prefixes and views of $\pathLL p$ are made of actions of $\pathLL p$ hence are elements of $V$ as $V$ is complete.
Similarly, let $\pathLL p, \pathLL q \in V$ then shuffles of $\pathLL p$ and $\pathLL q$ are made of actions of paths of $V$ hence are elements of $V$.
The same applies for $\dual V$.
Finally, as $V$ is prefix-closed, $\dual V$ is daimon-closed, and vice-versa.
\end{proof}

Let us go back now to the general case where the data $V/\dual V$ may be additive. The following proposition links properties of the non-additive case to the general case.

\begin{lem}\label{lem:AdditiveNonAdditive}
Let $V/\dual V$ be a data set, non necessarily non-additive.
Let $(V_i/\dual{V_i})$ be the family of maximal non-additive data such that $V_i \subset V$.
The data $V/\dual V$ is regular iff each $V_i/\dual{V_i}$ is regular.
\end{lem}
\begin{proof}
It is immediate that $V$ and $\dual V$ are sets of legal paths, prefix closed, daimon closed, view-closed iff each $V_i$ and $\dual{V_i}$ have also these properties.
Finally the shuffle operation is not defined on paths that contain two actions with same focus, hence $V$ and $\dual V$ are shuffle closed iff each $V_i$ and $\dual{V_i}$ are shuffle closed.
\end{proof}

\begin{prop}\label{prop:RegularThenLudicable}
Let $V$ be a   set of legal paths. If $V/\dual V$ is a regular data then $V$ is ludicable. 
\end{prop}
\begin{proof}
It follows from lemma~\ref{lem:AdditiveNonAdditive} that we just need to check the positive saturation and negative saturation on each $V_i$ and $\dual{V_i}$. These properties are trivially satisfied since, as soon as $C$ is a  clique of $V_i$ (resp. $\dual{V_i}$) then all view in $\fullview C$ belongs to $V_i$ (resp. $\dual{V_i}$), by view closure, and therefore each legal path in $\fullview C$ belongs to $V_i$ (resp. $\dual{V_i}$) by shuffle closure.
\end{proof}

Regularity is trivially stable by operations $(\cdot)^\perp$ and $\bigshneg$. Let us consider now the tensor of behaviours.

\begin{prop}\label{prop:shuffleregular}
Let $V$ and $S$ be two sets of legal paths. 
$$
V/\dual V\mbox{ and } S/\dual S \mbox{  are regular iff } V\lshuffle S/\dual{V\lshuffle S}\mbox{ is regular.}
$$
\end{prop}

\proof
We only deal with the case $V$ and $S$ be two non-additive sets of legal paths respectively based on $\sigma\vdash$ and $\tau\vdash$.
Adding a positive action to obtain the same base $\vdash \xi$ is immediate, and the possibly additive case is dealt similarly as in lemma~\ref{lem:AdditiveNonAdditive}.
 We prove in each case that proposition~\ref{touschemins} may be applied.
\begin{itemize}

\item[($\Leftarrow$)]~
Since the empty path $\epsilon$ belongs to $S$, every path $\pathLL{p}_V$ of $V$ may be seen as belonging to the  shuffle $\pathLL{p}_V\shuffle\epsilon$. Therefore, all legal paths which may be written using only the actions from $V$ belong to $V$. Similarly for $S$.

\item[($\Rightarrow$)]~
Let $\pathLL r$ be a legal path made of actions occurring in $V$ or $S$, we prove first that $\pathLL r \in V \lshuffle S$.
Let $\pathLL v$ (resp. $\pathLL s$) be the subsequence of $\pathLL r$ made of actions of $V$ (resp. $S$), hence with foci subaddresses of $\sigma$ (resp. $\tau$).
Note that if $\pathLL r = w_1\kappa^-\kappa^+ w_2$ then $\kappa^-$ and $\kappa ^+$ are subaddresses either the two of $\sigma$ or the two of $\tau$ for $\pathLL r$ to be a path.
Thus $\pathLL v$ (resp.\ $\pathLL s$) is a finite alternated sequence of actions based on $\sigma \vdash$ (resp. $\tau \vdash$). As actions of $\pathLL r$ have distinct foci and, if the daimon is present, it is its last action, this is also the case for $\pathLL v$ and $\pathLL s$.
Furthermore, let $\kappa^+$ be an action of $\pathLL v$ justified by $\kappa^-$ in $\pathLL r$, let $\pathLL r = r_0\kappa^-r_1\kappa^+ r_2$ and $\pathLL v = v_1\kappa v_2$, then $\view{r_0}\kappa^-\kappa^+ = \view{r_0\kappa^-r_1\kappa^+} = \view{v_1\kappa^+}$ thus $\pathLL v$ is a path. Similarly $\pathLL s$ is a path.
Finally, with lemma~\ref{lem:dual_inverse_shuffle} (see annex~\ref{sec:reg_beh}), we have that $\pathLL v$ and $\pathLL s$ are legal paths.
\\
Hence $\pathLL v \in V$ and $\pathLL s \in S$, so $\pathLL r \in V \lshuffle S$.
In other words if $\pathLL r$ is made of actions occurring in $V \lshuffle S$ then $\pathLL r \in V \lshuffle S$.
Furthermore if $\pathLL r$ is a legal path made of actions occurring in $\dual{V \lshuffle S}$ then $\dual{\pathLL r}$ is a legal path made of actions occurring in $V$ and/or $S$ so $\dual{\pathLL r} \in V \lshuffle S$, \ie, $\pathLL r \in \dual{V \lshuffle S}$.

\noindent Therefore, by proposition~\ref{touschemins}, the data $V\lshuffle S/\dual{V\lshuffle S}$ is regular.
\qed
 \end{itemize}

\begin{defi}
A behaviour $\behaviour A$ is said {\tt regular} when $V_{\behaviour A}/\dual{V_{\behaviour A}}$ is regular.
\end{defi}

\begin{prop}\label{prop:TensorStableRegular}
Let $\behaviour A$ and $\behaviour B$ be two  positive alien behaviours distinct from $\behaviour 0$ and based on $\vdash \xi$, $\behaviour A\otimes\behaviour B$ is regular iff $\behaviour A$ and $\behaviour B$ are regular. 
\end{prop}

\proof
We just consider the paradigmatic case: as $\behaviour A$ and $\behaviour B$ are distinct from $\bf 0$, we can suppose that designs of $\behaviour A$ distinct from $\daimon$ have as first action $(+,\xi,I)$ and designs of $\behaviour B$ distinct from $\daimon$ have as first action $(+,\xi,J)$ and $I\cap J=\emptyset$. 

\begin{itemize}
\item We set $C=C_{\behaviour A}\lshuffle C_{\behaviour B}$ where $C=\{ \pathLL p \in \pathLL p_{\behaviour A}\shuffle \pathLL{p_{\behaviour B}} ~;~ \pathLL p$ legal, $\pathLL p_{\behaviour A}\in C_{\behaviour A},\; \pathLL p_{\behaviour B}\in C_{\behaviour B}\}$. We prove first that $C$ is a maximal, positively saturated  clique of $V_{\behaviour A}\lshuffle V_{\behaviour B}$  such that $\dual C$ is finite stable iff $C_{\behaviour A}$ and $C_{\behaviour B}$ are maximal, positively saturated  cliques of respectively $V_{\behaviour A}$ and $ V_{\behaviour B}$  such that $\dual C_{\behaviour A}$ and $\dual C_{\behaviour B}$ are  finite stable:
	\begin{itemize}
	\item The equivalence between the maximality of $C$ and the one of both $C_{\behaviour A}$ and $C_{\behaviour B}$ is immediate, as is the equivalence between the finite stability of $\dual C$ and the finite stability of both $\dual{C_{\behaviour A}}$ and $\dual{C_{\behaviour B}}$. 
	\item Let us check that $C$ is a positively saturated  clique of $V_{\behaviour A}\lshuffle V_{\behaviour B}$ iff $C_{\behaviour A}$ and $C_{\behaviour B}$ are  positively saturated  cliques of respectively $V_{\behaviour A}$ and $ V_{\behaviour B}$.  
	\\
	The condition is necessary: suppose that $\pathLL m\in C$ and $\pathLL n\kappa^-\kappa^+\in C$ while $\pathLL m\kappa^-\daimon\in V_{\behaviour A}\lshuffle V_{\behaviour B}$. Without loss of generality we can suppose that $\kappa^-$ is an action of $V_{\behaviour A}$, therefore $\kappa^+$ also is an action of $V_{\behaviour A}$. 
	Hence there exist paths $\pathLL m_{\behaviour A}\in C_{\behaviour A}$ and $\pathLL n_{\behaviour A}\kappa^-\kappa^+\in C_{\behaviour A}$ and $\pathLL m_{\behaviour B}\in C_{\behaviour B}$ and $\pathLL n_{\behaviour B} \in C_{\behaviour B}$ such that $\pathLL m\in \pathLL m_{\behaviour A}\shuffle\pathLL m_{\behaviour B}$, $\pathLL m\kappa^-\kappa^+\in \pathLL m_{\behaviour A}\kappa^-\kappa^+\shuffle\pathLL m_{\behaviour B}$ and $\pathLL n\kappa^-\kappa^+\in \pathLL n_{\behaviour A}\kappa^-\kappa^+\shuffle\pathLL n_{\behaviour B}$. 
	Since $\pathLL m\kappa^-\daimon\in V_{\behaviour A}\shuffle V_{\behaviour B}$, we have that $\pathLL m_{\behaviour A}\kappa^-\daimon\in V_{\behaviour A}$ and since $C_{\behaviour A}$ is positively saturated, we have that $\pathLL m_{\behaviour A}\kappa^-\kappa^+\in V_{\behaviour A}$. Therefore, $\pathLL m\kappa^-\kappa^+\in V_{\behaviour A}\lshuffle V_{\behaviour B}$. 
	\\
	The condition is also sufficient: suppose that $\pathLL m_{\behaviour A}\in C_{\behaviour A}$ and $\pathLL n_{\behaviour A}\kappa^-\kappa^+\in C_{\behaviour A}$ while $\pathLL m_{\behaviour A}\kappa^-\daimon\in V_{\behaviour A}$. 
	By applying the positive saturation of $C$ to  paths belonging to $\pathLL m_{\behaviour A}\shuffle (+,\xi,J)$, to $ \pathLL n_{\behaviour A}\kappa^-\kappa^+\shuffle (+,\xi,J)$ and $\pathLL m_{\behaviour A}\kappa^-\daimon\shuffle (+,\xi,J)$, we may conclude that $\pathLL m_{\behaviour A}\kappa^-\kappa^+$ belongs to $V_{\behaviour A}$.
	Idem with the behaviour $\behaviour B$.
	\end{itemize}

\item We prove now that $C=C_{\behaviour A}\lshuffle C_{\behaviour B}$ iff $\fullview C=\fullview{C_{\behaviour A}}\otimes\fullview{C_{\behaviour B}}$.
As sets of views, the designs $\fullview C$ and $\fullview{C_{\behaviour A}}\otimes\fullview{C_{\behaviour B}}$ are clearly identical when $ C= C_{\behaviour A}\shuffle C_{\behaviour B}$. Moreover the set of legal paths of $\fullview C$ being the sets of shuffles of  views of $\fullview C$ which are legal, we have that  $C=C_{\behaviour A}\lshuffle C_{\behaviour B}$ as soon as $\fullview C=\fullview{C_{\behaviour A}}\otimes\fullview{C_{\behaviour B}}$.

\item Let $\behaviour A$ and $\behaviour B$ be regular, then $V_{\behaviour A}/\dual{V_{\behaviour A}}$ and $V_{\behaviour B}/\dual{V_{\behaviour B}}$ are regular data. Thus $V_{\behaviour A}\lshuffle V_{\behaviour B}/\dual{V_{\behaviour A}\lshuffle V_{\behaviour B}}$ is regular (proposition~\ref{prop:shuffleregular}), hence $V_{\behaviour A}\lshuffle V_{\behaviour B}$ is ludicable (proposition~\ref{prop:RegularThenLudicable}). We prove that the behaviour with $V_{\behaviour A}\lshuffle V_{\behaviour B}$ as visitable paths is $\behaviour C = \behaviour A\otimes \behaviour B$. 
	Indeed, a design of $|\behaviour C|$ is $\fullview C$ when $C$ is a maximal, positively saturated clique of $V_{\behaviour A}\lshuffle V_{\behaviour B}$, such that $\dual C$ is finite stable, that is exactly a design $\fullview{C_{\behaviour A}}\otimes\fullview{C_{\behaviour B}}$, or in other words a design $\design A\otimes\design B$ when $\design A$ and $\design B$ belong respectively to $|\behaviour A|$ and $|\behaviour B|$ (by proposition~\ref{prop:caracDessinIncarnation}). By internal completeness, such designs are exactly the ones of $|\behaviour A\otimes\behaviour B|$. Thus $\behaviour A \otimes \behaviour B$ is regular. Furthermore $V_{\behaviour A \otimes \behaviour B} = V_{\behaviour A} \lshuffle V_{\behaviour B}$.

\item Suppose now that $\behaviour A \otimes \behaviour B$ is regular. The inverse reasoning of the previous item yields $\behaviour A$ and $\behaviour B$ regular.
\qed
\end{itemize}

Let us notice that regularity is necessary to ensure $V_{\behaviour A\otimes \behaviour B}=V_{\behaviour A}\lshuffle V_{\behaviour B}$. Indeed, equality is not true even when behaviours are not view-closed as seen in example~\ref{example-nonregular}. However the set of visitable paths of a tensor of two behaviours may be characterized without considering hypothesis of regularity. The following proposition~\ref{prop:OperationsOnVisitable} is a joint work with A. Pavaux. One may find in~\cite{PavauxCSL17} the proof in a framework of Ludics {\em \`a la Terui}. A. Pavaux proposed also in~\cite{PavauxCSL17} a different proof of proposition~\ref{prop:TensorStableRegular} based on proposition~\ref{prop:OperationsOnVisitable}. In her paper, she uses these results for studying a representation of data and function types.

\begin{restatable}{prop}{restateOperationsOnVisitable}
\label{prop:OperationsOnVisitable}
Let $\behaviour{P}$ and $\behaviour{Q}$ be alien positive behaviours,\\
$\pathLL r \in V_{\behaviour P \otimes \behaviour Q}$ iff the two following conditions are satisfied:
\begin{itemize}
	\item {\em(Shuffle condition)} $\pathLL r \in V_{{\behaviour{P}}} \lshuffle V_{{\behaviour{Q}}}$,
	\item {\em(Dual condition)} for all path $\overline{\pathLL s \kappa^-}$ in $\fullview{\,\dual{\pathLL r}\,}$, if there exist paths $\pathLL p' = (+,\xi,I)\pathLL p'_1 \in V_{\behaviour{P}}$\\
		and $\pathLL q' = (+,\xi,J)\pathLL q'_1 \in V_{\behaviour{Q}}$ with $\pathLL s \in (+,\xi,I \cup J)(\pathLL p'_1 \shuffle \pathLL q'_1)$, \\
		then either $\pathLL p' \kappa^-\daimon \in V_{\behaviour P}$ or $\pathLL q' \kappa^-\daimon \in V_{\behaviour Q}$.
\end{itemize}
\end{restatable}
\begin{proof}
See annex~\ref{subsec:proofs}.
\end{proof}

We are now able to state the main theorem of this section: MALL formulas are denoted by regular and finite behaviours. A behaviour is finite when it contains only a finite number of designs in the incarnation and these designs are finite, \ie, each such design has a finite number of actions.
\begin{thm}\label{th:regular-finite}
Let $\behaviour E$ be a behaviour. $\behaviour E$ is regular and finite iff   it is generated by the following grammar:
\begin{align*}
\behaviour P &::= \bf 0 ~|~ \bf 1 ~|~ \bigshpos \behaviour N ~|~ {\behaviour N}^\perp ~|~ \behaviour P\otimes \behaviour P ~|~ \behaviour P \oplus \behaviour P
\\
\behaviour N &::= \boldtop ~|~ \boldbot ~|~ \bigshneg \behaviour P ~|~ {\behaviour P}^\perp ~|~ \behaviour N \parr \behaviour N ~|~ \behaviour N \with \behaviour N         
\end{align*}
\end{thm}
\begin{proof}
Let $\behaviour E$ be a regular and finite behaviour, hence the proof may be done by induction on the number of distinct actions present in designs in the incarnation. The main ingredients are the following. If $\behaviour E$ has a negative base, consider its dual $\behaviour E^\perp$, thus of positive base, that has the same number of distinct actions. 
If $\behaviour E$ is different from $\bf 0$ and $\bf 1$ then $\behaviour E$ may be decomposed as a $\oplus$ of connected behaviours $\behaviour E_i$, one for each distinct first action~\cite{DBLP:journals/mscs/Girard01}. 
Obviously each $\behaviour E_i$ is a finite behaviour.
Remark that the family $(V_{\behaviour E_i}/\dual{V_{\behaviour E_i}})$ is the family of maximal non-additive data such that $V_{\behaviour E_i} \subset V_{\behaviour E}$. 
Thus $\behaviour E_i$ is a regular behaviour (consequence of Lemma~\ref{lem:AdditiveNonAdditive}). 
A connected, regular and finite behaviour may be decomposed as a tensor of regular and finite behaviours (Proposition~\ref{prop:TensorStableRegular}), or a shift when the first ramification has only one element.
Conversely,
it is immediate that $\bf 0$ and $\bf 1$ are regular and regularity is stable by connectives $\cdot^\perp$ and $\bigshpos$. Finally it follows from Lemma~\ref{lem:AdditiveNonAdditive} and Proposition~\ref{prop:TensorStableRegular} that regularity is also stable by $\oplus$ and $\otimes$.
\end{proof}

Note that theorem~\ref{th:regular-finite} is a full characterization of MALL behaviours amongst finite behaviours: there exist finite behaviours that are not regular, hence not generated by the previous grammar.
For example we already mentioned that the behaviour $\behaviour B$ defined in example~\ref{example-nonregular}, obviously finite, is not regular. Let us give some insight on the non-regularity of $\behaviour B$. For ease of reading, we recall here that the incarnation of $\behaviour B$ is the daimon-closure of the set $\{\design B_1, \design B_2\}$:
\begin{center}
\begin{tabular}{cc}
$\design B_1 =\!\!\!\!\!\!\!\!\!\!\!\!\!\!\!\!
\scalebox{.8}{
\infer[\gamma_1^+]{\vdash \langle\rangle}{
\infer[\lambda_0^-]{\beta \vdash}
	{
	\infer[\lambda_0^+]{\vdash \beta0}
		{
		\infer[\lambda_1^-]{\beta01 \vdash}
			{
			\infer[\lambda_1^+]{\vdash \beta011}
				{
				\beta0111 \vdash
				}
			}
		&
		\infer[\lambda_2^-]{\beta02 \vdash}
			{
			\infer[\lambda_2^+]{\vdash \beta022}
				{
				\beta0222 \vdash
				}
			}
		}
	}
}
}
$
&
$\design B_2 =\!\!\!\!\!\!\!\!\!\!\!\!\!\!
\scalebox{.8}{
\infer[\gamma_1^+]{\vdash \langle\rangle}{
\infer[\lambda_0^-]{\beta \vdash}
	{
	\infer[\lambda_0^+]{\vdash \beta0}
		{
		\infer[\lambda_1^-]{\beta01 \vdash}
			{
			\infer[\daimon]{\vdash \beta011}
				{
				}
			}
		&
		\beta02 \vdash
		}
	}
}
}
$
\end{tabular}
\end{center}
Obviously we may begin its decomposition: $\behaviour B = \bigshpos \bigshneg \behaviour C$ where the incarnation of $\behaviour C$  is the daimon-closure of the set $\{\design C_1, \design C_2\}$:
\begin{center}
\begin{tabular}{cc}
$\design C_1 =\!\!\!\!\!\!\!\!
\scalebox{.8}{
	\infer[\lambda_0^+]{\vdash \beta0}
		{
		\infer[\lambda_1^-]{\beta01 \vdash}
			{
			\infer[\lambda_1^+]{\vdash \beta011}
				{
				\beta0111 \vdash
				}
			}
		&
		\infer[\lambda_2^-]{\beta02 \vdash}
			{
			\infer[\lambda_2^+]{\vdash \beta022}
				{
				\beta0222 \vdash
				}
			}
		}
}
$
&
$\design C_2 =\!\!\!\!\!\!
\scalebox{.8}{
	\infer[\lambda_0^+]{\vdash \beta0}
		{
		\infer[\lambda_1^-]{\beta01 \vdash}
			{
			\infer[\daimon]{\vdash \beta011}
				{
				}
			}
		&
		\beta02 \vdash
		}
}
$
\end{tabular}
\end{center}
The behaviour $\behaviour C$ is not regular. By the way, as the first proper action is unique and positive but not a shift, if $\behaviour C$ were regular, $\behaviour C$ should be a tensor of behaviours. But the tensor is commutative. In terms of interaction paths, this means that the path $\lambda_0^+\lambda_1^-\lambda_1^+\lambda_2^-\lambda_2^+$ (left then right chronicle in $\design C_1$) is visitable iff the path $\lambda_0^+\lambda_2^-\lambda_2^+\lambda_1^-\lambda_1^+$ (right then left chronicle in $\design C_1$) is visitable. However the first path is visitable but not the second one. In other words, interaction has to visit the left chronicle of $\design C_1$ before the right one: non-commutativity is present there. Next section is devoted to understand in which extent non-commutative tensors may be considered in this framework.

Non-commutativity is not the only kind of non-regularity that we may observe as shown in example~\ref{exa:entangled_behaviour}.
\begin{exa}\label{exa:entangled_behaviour} 
Let us consider the behaviour $\behaviour{D} = \{\design{D}_1,\design{D}_2\}^{\perp\perp}$ where designs $\design{D}_1$ and $\design{D}_2$ are given below. 
\begin{center}
$\design{D}_1 = 
\scalebox{.75}{
\infer{\vdash \xi}
{
	\infer{\xi1 \vdash}
	{
		\infer[\alpha_1]{\vdash \xi10}
			{\xi100 \vdash}
	}
	&
	\infer{\xi2 \vdash}
	{
		\infer[\beta_1]{\vdash \xi20}
			{\xi200 \vdash}
	}
}
}
$
~~~~~
$\design{D}_2 = 
\scalebox{.75}{
\infer{\vdash \xi}
{
	\infer{\xi1 \vdash}
	{
		\infer[\alpha_2]{\vdash \xi10}
			{\xi101 \vdash}
	}
	&
	\infer{\xi2 \vdash}
	{
		\infer[\beta_2]{\vdash \xi20}
			{\xi201 \vdash}
	}
}
}
$
\end{center}
$\behaviour{D}$ is a finite behaviour whose incarnation is the daimon-closure of these two designs, and $\behaviour D$ is not regular. It cannot also be interpreted in terms of non-commutativity. Indeed, during an interaction, action $\alpha_1$ may be followed by action $\beta_1$ (or the converse), but not by action $\beta_2$ (or $\alpha_2$). And this situation is symmetric changing indices $1$ and $2$ in the previous statement. This situation is a kind of entanglement as part of a (standard) tensor.
We let the study of such a situation to further works. 
\end{exa}

\section{Beyond Regular Behaviours}\label{sec:SimpleOrientedTensor}

In previous sections we presented examples of behaviours that are not regular, hence are not generated by the connectives of multiplicative-additive Linear Logic. We propose in this section a study of non-commutative connectives, as they are defined in the literature, plus a new one.

Non-commutativity has been a subject of interest in Logic. J. Lambek~\cite{lambek} developed an intuitionistic non-commutative logic by omitting the exchange rule of the sequent calculus but giving rise to two symbols of implication. With regards to Linear Logic, D. Yetter~\cite{Yetter90} (after a talk given by J.-Y. Girard in 1987) studied a cyclic version: the exchange rule is replaced by a cyclic one over the list of formulas of a sequent. P. Ruet~\cite{DBLP:journals/mscs/Ruet00}, with further works with M. Abrusci~\cite{DBLP:journals/apal/AbrusciR99} and R. Maieli~\cite{DBLP:journals/iandc/MaieliR03}, developed a logic fully integrating commutativity and non-commutativity. Again, this is the exchange rule that is replaced. In all these cases, non-commutativity arises {\em spatially}: the structure over the set/list of formulas in a sequent is non-commutative, however the choice of formula to be decomposed (in a bottom-up reading of proofs) does not depend on this structure. 
In the following, we begin with the study of a non-commutative tensor $\olessthan$ due to J.-Y. Girard~\cite{DBLP:journals/mscs/Girard01} but stable by regularity, hence adding it to MALL connectives does not extend the set of behaviours. This connective is not spatial but {\em prioritize} between operands when a choice between actions has to be done: only one action is kept. Then we study connectives that deal with {\em temporality}: what is at stake is to begin interaction with one operand before the other. This is the case with the sequoid $\oslash$ connective proposed by M. Churchill, J. Laird and G. McCusker~\cite{DBLP:journals/apal/ChurchillLM13}. Regularity is not stable by this connective as well as a new one we propose in a last subsection.

\subsection{A not so significant non-commutative connective.}\label{subsec:oriented_Girard}

If J.-Y. Girard presented non-commutative connectives in his seminal paper of Ludics~\cite{DBLP:journals/mscs/Girard01}, these are of no help for decomposing a behaviour. It is the case for the non-commutative tensor product $\olessthan$.
Let $\design A$ and $\design B$ be two designs, in $\design A \olessthan \design B$ priority is given to chronicles of $\design B$ when a choice has to be done. In terms of execution, \ie, visitable paths, priority is given to the strategy $\design B$ again the strategy of $\design A$.
The connective $\olessthan$ is defined when behaviours are not disjoint, \ie, the ramification of their first actions is not necessarily disjoint:
\begin{defi}
Let $\design A$ and $\design B$ be two positive designs of same base $\vdash \xi$,
	\begin{itemize}
	\item If $\design A = \{\daimon\}$ or $\design B = \{\daimon\}$, $\design A \olessthan \design B = \design A \otimes \design B = \{\daimon\}$.
	\item Otherwise let $\design A = (+,\xi,I)(\bigcup_{j\in J, K_j} (-,\xi.j,K_j)\design A_{j,K_j} \cup \bigcup_{i\not\in J, L_i}(-,\xi.i,L_i)\design A_{i,L_i})$ and $\design B = (+,\xi,J)\design B'$, then $\design A \olessthan \design B = \design B \otimes (+,\xi,I \setminus J)\bigcup_{i\not\in J, L_i}(-,\xi.i,L_i)\design A_{i,L_i}$.
	\end{itemize} 
Let $\behaviour A$ and $\behaviour B$ be two behaviours, $\behaviour A \olessthan \behaviour B = \{\design A \olessthan \design B ~;~ \design A \in \behaviour A, \design B \in \behaviour B\}^{\perp\perp}$.
\end{defi}

The tensor $\olessthan$ is non-commutative, however, as it immediately follows from the definition, the design $\design A \olessthan \design B$ may always be viewed as the tensor of two designs.
As the connective $\olessthan$ distributes over $\oplus$~\cite{DBLP:journals/mscs/Girard01}, we can study this connective $\olessthan$ on connected behaviours, \ie, behaviours with a unique first ramification: in this case in particular internal completeness is satisfied.
We remark in the following lemma a property that relates the non-commutative tensor of two behaviours with a commutative tensor of two behaviours.
Hence augmenting the grammar for regular behaviours with such a non-commutative tensor does not change the language of behaviours. 

\begin{lem}
Let $\behaviour A$ and $\behaviour B$ be two connected positive behaviours of base $\vdash \xi$, let $J$ be such that $\behaviour B = \{\daimon\} \cup (+,\xi,J)\behaviour B'$. We consider the following function $\phi$: let $\design A \in \behaviour A$,
$$\phi(\design A) = \design A \backslash \bigcup_{i \in I\cap J,K_i} (+,\xi,I)(-,\xi.i,K_i)\design A_{i,K_i}$$
Let $\phi(\behaviour A) = \{\phi(\design A) ~;~ \design A \in \behaviour A\}$.
Then $\behaviour A \olessthan \behaviour B = \phi(\behaviour A)^{\perp\perp} \otimes \behaviour B$.
\end{lem}
\begin{proof}
Note that $\phi(\design A)$ is always a design hence $\phi(\behaviour A)^{\perp\perp}$ is a behaviour.
Remark first that, as behaviours are connected, for designs $\design A \in \behaviour A$ and $\design B \in \behaviour B$, we have that $\design A \olessthan \design B = \phi(\design A) \otimes \design B$. \\
Hence $\behaviour A \olessthan \behaviour B 
	= (\phi(\behaviour A) \otimes \behaviour B)^{\perp\perp}
	= (\phi(\behaviour A)^{\perp\perp} \otimes \behaviour B)^{\perp\perp}
	= \phi(\behaviour A)^{\perp\perp} \otimes \behaviour B
$.
\end{proof}


\subsection{The sequoid $\oslash$ game~\cite{DBLP:journals/apal/ChurchillLM13}}\label{subsec:sequoid}

M. Churchill, J. Laird and G. McCusker proposed in~\cite{DBLP:journals/apal/ChurchillLM13} a first-order logic {\tt WS1} and a games model in which proofs denote history-sensitive strategies. We do not consider here the full logic but only represent in terms of Ludics the essence of their main specific non-commutative tensor connective $\oslash$. Briefly speaking, the first move in a play of $\behaviour A \oslash \behaviour B$ has to be done in $\behaviour A$, the following moves are considered as with the standard commutative tensor. We show how such a connective may be defined in our setting by considering an adequate set of visitable paths. Obviously such a set of visitable paths should contain $\daimon$ and the first (positive) action of $\behaviour A \otimes \behaviour B$. Furthermore, it should contain visitable paths of $\behaviour A \otimes \behaviour B$ with the restriction that the second action ``comes'' from $\behaviour A$.

\begin{prop}\label{prop:ludicabilityofOslash}
Let $\behaviour A$ and $\behaviour B$ be two connected disjoint positive behaviours of base $\vdash \xi$,
let $\kappa_0^+$ (resp.\ $\kappa_{\behaviour A}^+$) be the first action of designs in $\behaviour A \otimes \behaviour B$ (resp.\ $\behaviour A$), 
let $S = \{\daimon, \kappa_0^+\} \cup \{\kappa_0^+\kappa_0^-\pathLL s ~;~ \kappa_{\behaviour A}^+\kappa_0^-\daimon \in V_{\behaviour A}, \kappa_0^+\kappa_0^-\pathLL s \in V_{\behaviour A \otimes \behaviour B}\}$.
Then $S$ is ludicable.
\end{prop}
\begin{proof}
$S$ is a set of legal paths as a subset of $V_{\behaviour A \otimes \behaviour B}$.
$S$ is also daimon and prefix closed.
\\
$S$ is pre-ludicable:
	\begin{itemize}
	\item Positive saturation: let $\pathLL p \in S$, let us consider the set $C_{\pathLL p}$ for $S$. Let $\pathLL m \in C_{\pathLL p}$, $\pathLL n\kappa^-\kappa^+ \in C_{\pathLL p}$, $\pathLL m\kappa^-\daimon \in S$, note that $C_{\pathLL p} \subset D_{\pathLL p}$ where $D_{\pathLL p}$ is the positively saturated maximal clique for $\pathLL p$ in $V_{\behaviour A \otimes \behaviour B}$. Thus $\pathLL m \kappa^-\kappa^+ \in D_{\pathLL p}$. Hence $\pathLL m \kappa^-\kappa^+ \in C_{\pathLL p}$ as $\pathLL m \kappa^-\kappa^+$ satisfies the required properties for being in $S$. Thus $C_{\pathLL p}$  is positively saturated.
	\item Negative saturation: let $\pathLL p \in S$, $\pathLL p \kappa^-\daimon$ be a legal path, suppose that for all positively saturated maximal clique $C$ of $S$, $\fullview{C} \perp \fullview{\,\overline{\pathLL p\kappa^-}\,}^c$, then for all positively saturated maximal clique $D$ of $V_{\behaviour A \otimes \behaviour B}$, we have also that $\fullview{D} \perp \fullview{\,\overline{\pathLL p\kappa^-}\,}^c$, then $\pathLL p\kappa^-\daimon \in V_{\behaviour A \otimes \behaviour B}$. Thus $\pathLL p\kappa^-\daimon \in S$ as conditions are fulfilled.
	\end{itemize}
$\dual S$ is pre-ludicable: The proof is similar noticing that $\dual{V_{\behaviour A \otimes \behaviour B}}$ is pre-ludicable.
\end{proof}

It is straightforward to notice that $S$ as defined before is the set of visitable paths of a behaviour $\behaviour A \oslash \behaviour B$ with properties as required in the logic {\tt WS1}~\cite{DBLP:journals/apal/ChurchillLM13}. We do not go further on the study of {\tt WS1} in our framework, we just remark the strength of our approach to prove that a connective is well-defined when given in terms of visitable paths (or plays). Note finally that $\behaviour A \oslash \behaviour B$ is not a regular behaviour hence this connective may really augment the grammar of MALL.


\subsection{An absolute non-commutative connective.}

We propose another non-commutative connective written $\fullotensor$ defined by means of visitable paths. 
We shall here consider an oriented tensor such that a visitable path $\pathLL r$ of $\behaviour A \fullotensor \behaviour B$ consists of a visitable path $\pathLL p$ of $\behaviour A$ possibly followed by a visitable path $\pathLL q$ of $\behaviour B$ (except $\pathLL q$'s first action). 
When $\pathLL q$ is not empty, the path $\pathLL p$ should be maximal with respect to the design $\fullview{\pathLL p}$ it generates, otherwise it could be possible to switch back from $\pathLL q$ to a path in $\behaviour A$, what should be rejected. 
Note that if a visitable path is maximal with respect to its generated design then it is also maximal with respect to length. The converse is not always true as stated in example~\ref{exa:maxVisitIncluded}.


\begin{exa}\label{exa:maxVisitIncluded}
Let us consider the behaviour $\behaviour A = \{\design A\}^{\perp\perp}$ where the design $\design A$ is given below on the left.
Its dual $\behaviour A^\perp$ contains designs $\design B_1$ and $\design B_2$, drawn on the right, and $\daimon$-restrictions of these two designs.
Let us consider the two following paths:
\begin{itemize}
	\item $\pathLL p = \kappa^+\kappa^-\kappa_1^+\kappa_1^-\kappa_2^+\lambda_0^-\lambda_0^+\lambda_1^-\lambda_1^+ = \normalisationSeq{\design A}{\design B_1}$
	\item $\pathLL q = \kappa^+\kappa^-\kappa_1^+\lambda_0^-\lambda_0^+\kappa_1^-\kappa_2^+ = \normalisationSeq{\design A}{\design B_2}$
\end{itemize}

Paths $\pathLL p$ and $\pathLL q$ are maximal in length among visitable paths of $\behaviour A$. Furthermore $\fullview{\pathLL q} \subsetneq \fullview{\pathLL p} = \design A$, hence $\pathLL q$ is not maximal with respect to generated designs.

\renewcommand{\inferstyle}{\scriptstyle}

\begin{center}
$\design A = 
\infer[\scriptstyle \kappa^+]{\vdash \xi}{
\infer[\scriptstyle \kappa^-]{\alpha \vdash}
	{
	\infer[\scriptstyle \kappa_1^+]{\vdash \alpha1, \alpha2}
		{
		\infer[\scriptstyle \kappa_1^-]{\alpha11 \vdash \alpha2}
			{
			\infer[\scriptstyle \kappa_2^+]{\vdash \alpha111, \alpha2}
				{
				\alpha22 \vdash \alpha111
				}
			}
		}
	}
&
\infer[\scriptstyle \lambda_0^-]{\beta \vdash}
	{
	\infer[\scriptstyle \lambda_0^+]{\vdash \beta0}
		{
		\infer[\scriptstyle \lambda_1^-]{\beta00 \vdash}
			{
			\infer[\scriptstyle \lambda_1^+]{\vdash \beta000}
				{
				\beta0000 \vdash
				}
			}
		}
	}
}
$
~~~~~~
$\design B_1 = \!\!\!\!
\infer{\xi \vdash}
{
\infer{\vdash \alpha, \beta}
	{
	\infer{\alpha1 \vdash}
		{
		\infer{\vdash \alpha11}{\alpha111 \vdash}
		}
	&
	\infer{\alpha2 \vdash \beta}
		{
		\infer{\vdash \alpha22, \beta}
			{
			\infer{\beta0 \vdash \alpha22}
				{
				\infer{\vdash \beta00, \alpha22}
					{
					\infer{\beta000 \vdash \alpha22}
						{
						\infer[\daimon]{\vdash \beta0000, \alpha22}{}
						}
					}
				}
			}
		}
	}
}
$
~~~~~~
$\design B_2 = \!\!\!\!
\infer{\xi \vdash}
{
\infer{\vdash \alpha, \beta}
	{
	\infer{\alpha1 \vdash}
		{
		\infer{\vdash \alpha11, \beta}
			{
			\infer{\beta0 \vdash \alpha11}
				{
				\infer{\vdash \beta00, \alpha11}
					{
					\alpha111 \vdash \beta00
					}
				}
			}
		}
	&
	\infer{\alpha2 \vdash}
		{
		\infer[\daimon]{\vdash \alpha22}{}
		}
	}
}
$

\end{center}
\renewcommand{\inferstyle}{\textstyle}

\end{exa}

To prove that such an oriented connective $\fullotensor$ may be well-defined (proposition~\ref{prop:orientedTensor}), we consider the set of legal paths that should be visitable. We prove first that such a set of legal paths is ludicable (lemmas~\ref{lem:S_legalPaths},~\ref{lem:S_preludicable},~\ref
{lem:dualS_preludicable}), hence the existence of the connective follows: The proof of proposition~\ref{prop:orientedTensor} is then immediate.


\begin{restatable}{prop}{restateOrientedTensor}
\label{prop:orientedTensor}
Let $\behaviour A$ and $\behaviour B$ be two connected disjoint positive behaviours of base $\vdash \xi$, let us consider the three following definitions:
	\begin{itemize}
	\item $V_{\behaviour A_{[\behaviour B]}}$ is the set of paths $\kappa w$ of $V_{\behaviour A \otimes \behaviour B}$ such that $w$ contains only actions of  $V_{\behaviour A}$.
	\item $V_{\behaviour A_{[\behaviour B]}}^{max} = \{\pathLL p ~;~ \pathLL p\in V_{\behaviour A_{[\behaviour B]}}, \pathLL p~\daimon$-free, $\nexists \pathLL q \in V_{\behaviour A_{[\behaviour B]}}, \fullview{\pathLL p} \subsetneq \fullview{\pathLL q}\}$
	\item $V_{\behaviour B}^- = \{\pathLL q ~;~ \exists \kappa^+, \kappa^+\pathLL q\in V_{\behaviour B}\}$
	\end{itemize}
Suppose moreover that the behaviour $\behaviour A$ satisfies the following constraint $(C)$:
	\begin{itemize}
	\item[] $(C)$ For each $\pathLL p \in V_{\behaviour A}$ and $\pathLL q \in V_{\behaviour A}^{max}$ such that $\pathLL p$ and $\pathLL q$ end on the same (positive) action and $\fullview{\overline{\pathLL p}} \subset \fullview{\overline{\pathLL q}}$, then $\pathLL p \in V_{\behaviour A}^{max}$.
	\end{itemize}
Then $S=  V_{\behaviour A_{[\behaviour B]}} \cup V_{\behaviour A_{[\behaviour B]}}^{max}V_{\behaviour B}^-$ is ludicable.\\
We note $\behaviour A \fullotensor \behaviour B$ the behaviour such that $V_{\behaviour A \fullotensor \behaviour B} = S$.
\end{restatable}

Remark that $V_{\behaviour A_{[\behaviour B]}}^{max}$ is the subset of $V_{\behaviour A_{[\behaviour B]}}$ such that paths generate designs that have no extension in the incarnation $|\behaviour A_{[\behaviour B]}|$ of $\behaviour A_{[\behaviour B]}$: such paths are maximal with respect to designs they generate in $|\behaviour A_{[\behaviour B]}|$. 
Remark also that $V_{\behaviour B}^-$ is obtained by deleting the first action of a path of $V_{\behaviour B}$: this first action is already `taken into account' by the first action of paths of $V_{\behaviour A_{[\behaviour B]}}$. 
Example~\ref{exa:orientedBehaviour} proposes such a construction.
Example~\ref{exa:constraintC_counterEx} shows that constraint $(C)$ is required otherwise unexpected visitable paths appear. However there is always a possibility to define the ludicable closure of $S$ (see 
subsection~\ref{subsec:LudicableClosures}). Thus the connective $\fullotensor$ may always be defined even if the structure of visitable paths is not guaranteed when condition $(C)$ is not fulfilled (see next subsection for basic properties in that case).


\begin{exa}\label{exa:orientedBehaviour}
Let $\behaviour A$ be the behaviour $\{\design A\}^{\perp\perp}$ and $\behaviour B$ be the behaviour $\{\design B\}^{\perp\perp}$, then the behaviour $\behaviour A \otimes \behaviour B = \{\design A \otimes \design B\}^{\perp\perp}$ whereas $\behaviour A \fullotensor \behaviour B = \{\design A \otimes \design B, \design C, \design D\}^{\perp\perp}$ (designs $\design A$, $\design B$, $\design C$ and $\design D$ are drawn below). Indeed we have that the behaviour $\behaviour A$ satisfies constraint $(C)$ and:
\begin{itemize}
	\item $V_{\behaviour A}^{max} = \{
		\kappa_0^+\kappa_1^-\kappa_1^+\kappa_2^-\kappa_2^+,
		\kappa_0^+\kappa_2^-\kappa_2^+\kappa_1^-\kappa_1^+
		\}$
		and $V_{\behaviour A} = V_{\behaviour A}^{max\daimon}$
	\item $V_{\behaviour B}^- = \{\lambda_3^-\lambda_3^+\}$
\end{itemize}

The dual behaviour $(\behaviour A \fullotensor \behaviour B)^\perp$ is given as $\{\design X_1, \design X_2\}^{\perp\perp}$.
Note that behaviours $\behaviour A$, $\behaviour B$ and $\behaviour A \otimes \behaviour B$ are regular whereas $\behaviour A \fullotensor \behaviour B$ is not regular: the sequence $(+,\xi,\{1,2,3\})\lambda_3^-\lambda_3^+$ is a chronicle in a design of $|\behaviour A \fullotensor \behaviour B|$ however this chronicle is not visitable in $\behaviour A \fullotensor \behaviour B$. This example is sufficient to prove that $\behaviour A \fullotensor \behaviour B$ cannot be generated by connectives of Linear Logic.

\renewcommand{\inferstyle}{\scriptstyle}
\begin{center}
$
\design A = \!\!\!\!
	\infer[\scriptstyle \kappa_0^+]{\vdash \xi}
	{
		\infer[\scriptstyle \kappa_1^-]{\xi1 \vdash}
		{
			\infer[\scriptstyle \kappa_1^+]{\vdash \xi11}
			{
				\xi111 \vdash
			}
		}
		&
		\infer[\scriptstyle \kappa_2^-]{\xi2 \vdash}
		{
			\infer[\scriptstyle \kappa_2^+]{\vdash \xi22}
			{
				\xi222 \vdash
			}
		}
	}
~~~~~~
\design B = \!
	\infer[\scriptstyle \lambda_0^+]{\vdash \xi}
	{
		\infer[\scriptstyle \lambda_3^-]{\xi3 \vdash}
		{
			\infer[\scriptstyle \lambda_3^+]{\vdash \xi33}
			{
				\xi333 \vdash
			}
		}
	}
$
\\
\vspace{3mm}
$
\design A \otimes \design B = \!\!\!\!
	\infer{\vdash \xi}
	{
		\infer{\xi1 \vdash}
		{
			\infer{\vdash \xi11}
			{
				\xi111 \vdash
			}
		}
		&
		\infer{\xi2 \vdash}
		{
			\infer{\vdash \xi22}
			{
				\xi222 \vdash
			}
		}
		&
		\infer{\xi3 \vdash}
		{
			\infer{\vdash \xi33}
			{
				\xi333 \vdash
			}
		}
	}
~~~~~
\design C = \!\!\!\!
	\infer{\vdash \xi}
	{
		\infer{\xi1 \vdash}
		{
			\infer[\daimon]{\vdash \xi11}
			{
			}
		}
		&
		\infer{\xi2 \vdash}
		{
			\infer{\vdash \xi22}
			{
				\xi222 \vdash
			}
		}
		&
		\xi3 \vdash
	}
~~~~~
\design D = \!\!\!\!
	\infer{\vdash \xi}
	{
		\infer{\xi1 \vdash}
		{
			\infer{\vdash \xi11}
			{
				\xi111 \vdash
			}
		}
		&
		\infer{\xi2 \vdash}
		{
			\infer[\daimon]{\vdash \xi22}
			{
			}
		}
		&
		\xi3 \vdash
	}
$
\\
\vspace{3mm}
$
\design X_1 =
	\infer{\xi \vdash}
	{
		\infer{\vdash \xi1, \xi2, \xi3}
		{
			\infer{\xi11 \vdash \xi2, \xi3}
			{
				\infer{\vdash \xi111, \xi2, \xi3}
				{
					\infer{\xi22 \vdash \xi111, \xi3}
					{
						\infer{\vdash \xi111, \xi222, \xi3}
						{
							\infer{\xi33 \vdash \xi111, \xi222}
							{
								\infer[\daimon]{\vdash \xi111, \xi222, \xi333}{}
							}
						}
					}
				}
			}
		}
	}
$
~~~~~~
$
\design X_2 =
	\infer{\xi \vdash}
	{
		\infer{\vdash \xi1, \xi2, \xi3}
		{
			\infer{\xi22 \vdash \xi1, \xi3}
			{
				\infer{\vdash \xi1, \xi222, \xi3}
				{
					\infer{\xi11 \vdash \xi222, \xi3}
					{
						\infer{\vdash \xi111, \xi222, \xi3}
						{
							\infer{\xi33 \vdash \xi111, \xi222}
							{
								\infer[\daimon]{\vdash \xi111, \xi222, \xi333}{}
							}
						}
					}
				}
			}
		}
	}
$
\end{center}
\renewcommand{\inferstyle}{\textstyle}

\end{exa}


\begin{exa}\label{exa:constraintC_counterEx}
Let us consider the behaviour $\behaviour{A} = \{\design{E}, \design{F}\}^{\perp\perp}$ where designs $\design{E}$ and $\design{F}$ are given below. Its dual behaviour is $\behaviour{A}^\perp = \{\design{G}_{0}, \design{G}_{1}, \design{G}_{2}\}^{\perp\perp}$ with designs $\design{G}_{1}$, $\design{G}_{2}$ and $\design{G}_{3}$ given below.

\renewcommand{\inferstyle}{\scriptstyle}
\begin{center}
$\design{E} =
\infer[\scriptstyle \langle\rangle^+]{\vdash \langle\rangle}
	{
	\infer[\scriptstyle \kappa^-]{\xi \vdash}
		{
		\infer[\scriptstyle \kappa_2^+]{\vdash \xi1, \xi2}
			{
			\infer[\scriptstyle \kappa_{22}^-]{\xi22 \vdash \xi1}
				{
				\infer[\scriptstyle \kappa_1^+]{\vdash \xi222, \xi1}
					{
					\xi11 \vdash \xi222
					}
				}
			}
		}
	&
	\infer[\scriptstyle \lambda^-]{\sigma \vdash}
		{
		\infer[\scriptstyle \lambda^+]{\vdash \sigma1}
			{\sigma11 \vdash}
		}
	}
$
~~~~~~~~
$\design{F} =
\infer[\scriptstyle \langle\rangle^+]{\vdash \langle\rangle}
	{
	\infer[\scriptstyle \kappa^-]{\xi \vdash}
		{
		\infer[\scriptstyle \kappa_1^+]{\vdash \xi1, \xi2}
			{
			\xi11 \vdash \xi2
			}
		}
	&
	\infer[\scriptstyle \lambda^-]{\sigma \vdash}
		{
		\infer[\scriptstyle \lambda^+]{\vdash \sigma1}
			{\sigma11 \vdash}
		}
	}
$
\\
\vspace{3mm}
$\design{G}_{0} =
\infer{\langle\rangle \vdash}
	{
	\infer{\vdash \xi,\sigma}
		{
		\infer{\xi1 \vdash \sigma}
			{
			\infer{\vdash \xi11, \sigma}
				{
				\infer{\sigma1 \vdash \xi11}
					{
					\infer[\daimon]{\vdash \xi11, \sigma11}
						{}
					}
				}
			}
		&
		\infer{\xi2 \vdash}
			{
			\infer{\vdash \xi22}
				{
				\xi222 \vdash
				}
			}
		}
	}
$
~~~~~~~~
$\design{G}_{1} =
\infer{\langle\rangle \vdash}
	{
	\infer{\vdash \xi, \sigma}
		{
		\infer{\xi1 \vdash}
			{
			\infer[\daimon]{\vdash \xi11}
				{
				}
			}
		&
		\infer{\xi2 \vdash \sigma}
			{
			\infer{\vdash \xi22,\sigma}
				{
				\infer{\sigma1 \vdash \xi22}
					{
					\infer{\vdash \xi22, \sigma11}
						{\xi222 \vdash\sigma11}
					}
				}
			}
		}
	}
$
~~~~~~~~
$\design{G}_{2} =
\infer{\langle\rangle \vdash}
	{
	\infer{\vdash \xi, \sigma}
		{
		\infer{\sigma1 \vdash \xi}
			{
			\infer{\vdash \sigma11, \xi}
				{
				\infer{\xi1 \vdash}
					{
					\infer[\daimon]{\vdash \xi11}
						{
						}
					}
				&
				\infer{\xi2 \vdash}
					{
					\infer{\vdash \xi22}
						{
						\xi222 \vdash
						}
					}
				}
			}
		}
	}
$
\end{center}
\renewcommand{\inferstyle}{\textstyle}

\noindent Let us consider the three following paths:
\begin{itemize}
\item $\pathLL q = \langle\rangle^+\kappa^-\kappa_{2}^+\lambda^-\lambda^+\kappa_{22}^-$
\item $\pathLL p = \langle\rangle^+\kappa^-$
\item $\pathLL r = \langle\rangle^+\lambda^-\lambda^+\kappa^-$
\end{itemize}

\noindent We remark that:
\begin{itemize}
\item The three paths $\pathLL p\kappa_1^+$, $\pathLL q\kappa_1^+$, $\pathLL r\kappa_1^+$ are visitable in $\behaviour A$.
\item $\fullview{\pathLL q\kappa_1^+} = \design E \in |\behaviour A|$, $\fullview{\pathLL r\kappa_1^+} = \design F \in |\behaviour A|$, $\fullview{\dual{\pathLL q\kappa_1^+}} = \design G_1 \in |\behaviour A^\perp|$.
\item Thus $\pathLL q\kappa_1^+$ and $\pathLL r\kappa_1^+$ are maximal visitable paths of $\behaviour A$.
\item Paths $\pathLL p\kappa_1^+$ and $\pathLL q\kappa_1^+$ end on the same action and $\fullview{\overline{\pathLL p\kappa_1^+}} \subsetneq \fullview{\overline{\pathLL q\kappa_1^+}}$.
\item However the path $\pathLL p\kappa_1^+$ is not maximal as $\fullview{\pathLL p\kappa_1^+} \subsetneq \fullview{\pathLL r\kappa_1^+}$.
\end{itemize}
Hence the behaviour $\behaviour A$ does not satisfy the constraint $(C)$. Why is it a problem? Because unexpected paths may be visitable when one tries to apply an oriented tensor to it. Let us consider the behaviour $\behaviour B = \{\design B\}^{\perp\perp}$ where
\renewcommand{\inferstyle}{\scriptstyle}
$\design{B} =
\infer[]{\vdash \langle\rangle}
	{
	\infer[\scriptstyle \tau^-]{\tau \vdash}
		{
		\infer[\scriptstyle \tau^+]{\vdash \tau1}
			{\tau11 \vdash}
		}
	}
$.
\renewcommand{\inferstyle}{\textstyle}
Then the set $S=  V_{\behaviour A_{[\behaviour B]}} \cup V_{\behaviour A_{[\behaviour B]}}^{max}V_{\behaviour B}^-$ is not ludicable: $\dual S$ does not satisfy positive saturation. Indeed\footnote{We still note $\langle\rangle^+$ the action $(+,\langle\rangle,\{\xi,\sigma,\tau\})$, hence $\pathLL p$ and $\pathLL q$ are unchanged even if the base is changed.}, by definition of $S$, $\overline{\pathLL q\kappa_1^+\tau^-} \in \dual S$ and also $\overline{\pathLL p\kappa_1^+}\daimon \in \dual S$. As $\fullview{\overline{\pathLL p\kappa_1^+}} \subsetneq \fullview{\overline{\pathLL q\kappa_1^+}}$, any maximal clique $C$ of $\dual S$ that includes $\overline{\pathLL q\kappa_1^+\tau^-}$ includes also $\overline{\pathLL p}$. Thus, if $\dual S$ would satisfy positive saturation, we should have $\overline{\pathLL p\kappa_1^+\tau^-} \in \dual S$, \ie, $\pathLL p\kappa_1^+\tau^-\daimon \in S$, in contradiction with the definition of $S$.

\end{exa}


In the following, we use notations $\pathLL p, \pathLL p', \pathLL p_0, \dots$ for paths of $V_{\behaviour A_{[\behaviour B]}}$ or $V_{\behaviour A_{[\behaviour B]}}^{max}$, $\pathLL q, \pathLL q', \pathLL q_0, \dots$ for paths of $V_{\behaviour B}^-$, and $\pathLL r, \pathLL r', \dots$ for paths of $S$.
It is not difficult to prove that the set $S$ as defined in proposition~\ref{prop:orientedTensor} is made of legal paths.

\begin{lem}\label{lem:S_legalPaths}
$S$ is a set of legal paths.
\end{lem}
\begin{proof}
Paths of $V_{\behaviour A_{[\behaviour B]}}$ are legal: paths of $V_{\behaviour A}$ are legal and paths of $V_{\behaviour A_{[\behaviour B]}}$ differ from paths of $V_{\behaviour A}$ only because of the ramification of the first (hence positive) action that is a larger set.
\\
Let $\pathLL p\pathLL q \in V_{\behaviour A_{[\behaviour B]}}^{max}V_{\behaviour B}^-$ (with conventions of notation given above). $\pathLL p$ is legal as $\pathLL p \in  V_{\behaviour A_{[\behaviour B]}}$.
Note that an action in $\pathLL q$ is necessarily justified by the first action of $\pathLL p$: the base contains a unique positive address, and behaviours $\behaviour A$ and $\behaviour B$ are connected and disjoint.
Let $\pathLL p\pathLL q_0$ be a $\daimon$-free positive-ended prefix of $\pathLL p\pathLL q$, let $\kappa_0^+$ be the first action of $\pathLL p$, then we have that $\view{\pathLL p\pathLL q_0} = \view{\kappa_0^+\pathLL q_0}$. Note that $\kappa_0^+\pathLL q_0 \in V_{\behaviour B_{[\behaviour A]}}$ thus $\kappa_0^+\pathLL q_0$ is a legal path hence $\pathLL p\pathLL q_0$ is a path.
Let $\pathLL p\pathLL q_0$ be now a negative-ended prefix of $\pathLL p\pathLL q$. Remark first that $\overline{\pathLL q_0}$ begins with a positive action and that this positive action is necessarily justified by the first action $\overline{\kappa_0^+}$ of $\overline{\pathLL p}$: the base contains a unique positive address. Thus $\view{\,\overline{\pathLL p\pathLL q_0}\,} = \view{\,\overline{\pathLL p}\,}\view{\,\overline{\pathLL q_0}\,}$. As $\pathLL p$ and $\kappa_0^+\pathLL q_0$ are legal paths, it follows from the previous computation that $\overline{\pathLL p\pathLL q_0}$ is a path.
\end{proof}

Before we prove the two last lemmas necessary for proposition~\ref{prop:orientedTensor}, we consider two technical lemmas, which proofs are in the annex for ease of reading.

\begin{restatable}{lem}{restatecliqueDecomposition}
\label{lem:cliqueDecomposition}
$C$ is a positively saturated maximal clique of $S$ with first action $\kappa_0^+$ iff there is a unique decomposition $C = C_1 \cup C'C''$ such that $C_1$ is a positively saturated maximal clique of $V_{\behaviour A_{[\behaviour B]}}$ with first action $\kappa_0^+$, $C' = C_1 \cap V_{\behaviour A_{[\behaviour B]}}^{max}$ and $C''$ is empty if $C'$ is empty or $\kappa_0^+C''$ is a positively saturated maximal clique of $V_{\behaviour B_{[\behaviour A]}}$.
\end{restatable}
\begin{proof}
See annex~\ref{subsec:proofs:SimpleOrientedTensor}.
\end{proof}

\begin{restatable}{lem}{restatedualcliqueDecomposition}
\label{lem:dualcliqueDecomposition}
$D$ is a positively saturated maximal clique of $\dual S$ with first action $\kappa_0^-$ iff there is a unique decomposition $D = D_1 \cup \bigcup_{\kappa_1^- \in K} D'_{\kappa_1^-} D''_{\kappa_1^-}$ such that $D_1$ is a positively saturated maximal clique of $\dual{V_{\behaviour A_{[\behaviour B]}}}$ with first action $\kappa_0^-$, 
$K$ is the set of (negative) actions $\kappa_1^-$ such that $D_1\kappa_1^- \cap \overline{V_{\behaviour A_{[\behaviour B]}}^{max}} \neq \emptyset$,
$D'_{\kappa_1^-} = D_1\kappa_1^- \cap \overline{V_{\behaviour A_{[\behaviour B]}}^{max}}$
and each $\kappa_0^-D''_{\kappa_1^-}$ is a positively saturated maximal clique of $\dual{V_{\behaviour B_{[\behaviour A]}}}$ such that the first action of paths in $D''_{\kappa_1^-}$ is distinct of first action of paths in $D''_{\kappa_1'^-}$ when $\kappa_1^- \neq \kappa_1'^-$.
\end{restatable}
\begin{proof}
See annex~\ref{subsec:proofs:SimpleOrientedTensor}.
\end{proof}

\begin{lem}\label{lem:S_preludicable}
$S$ is pre-ludicable.
\end{lem}
\proof
Let $\kappa_0^+$ be the first action of paths of $\behaviour A_{[\behaviour B]}$. We write $C$ a positively saturated maximal clique of $S$. By lemma~\ref{lem:cliqueDecomposition}, we can write $C = C_1 \cup C'C''$ with properties as stated in this lemma. 
	\begin{itemize}
	\item By definition of $S$ and the fact that a set of visitable paths of a behaviour is ludicable, $S$ is prefix-closed and daimon-closed.
	\item (positive saturation) Let $\pathLL r$ be a path of $S$:
		\begin{itemize}
		\item Either $\pathLL r \in V_{\behaviour A_{[\behaviour B]}}$. As $V_{\behaviour A_{[\behaviour B]}}$ is ludicable, there exists a positively saturated maximal clique $C_1$ for $V_{\behaviour A_{[\behaviour B]}}$ such that $\pathLL r \in C_1$. Let $C''$ be a positively saturated maximal clique of $V_{\behaviour B}^-$. Let us consider $C = C_1 \cup (C_1 \cap V_{\behaviour A_{[\behaviour B]}}^{max})C''$. Then, by lemma~\ref{lem:cliqueDecomposition}, $C$ is a positively saturated maximal clique of $S$ that contains $\pathLL r$.
		\item Or $\pathLL r= \pathLL p\pathLL q \in V_{\behaviour A_{[\behaviour B]}}^{max}V_{\behaviour B}^-$. Let $C_1$ be a positively saturated maximal clique of $V_{\behaviour A_{[\behaviour B]}}$ that contains $\pathLL p$ and $C''$ be a positively saturated maximal clique of $V_{\behaviour B}^-$ that contains $\kappa_0^+\pathLL q$. Then, by lemma~\ref{lem:cliqueDecomposition}, we remark that $C_1 \cup (C_1 \cap V_{\behaviour A_{[\behaviour B]}}^{max})C''$ is a positively saturated maximal clique of $S$ that contains $\pathLL r$.
		\end{itemize}

	\item (negative saturation)\\
		Let $\pathLL{r} \in S$ and $\pathLL{r}\kappa^-\daimon$ be a legal path such that for all positively saturated maximal clique $C$ of $S$, we have that $\fullview{C} \perp \fullview{\overline{\pathLL r\kappa^-}}^c$. We have to prove that $\pathLL{r}\kappa^-\daimon \in S$:
		\begin{itemize}
		\item Either $\pathLL{r} \in V_{\behaviour A_{[\behaviour B]}}$: 

		If $\kappa^-$ is an action appearing in $V_{\behaviour B}^-$, note that $\kappa^-$ should be immediately justified by $\kappa_0^+$, hence $\kappa_0^+\kappa^-\daimon \in V_{\behaviour B}$. 
Suppose that $\pathLL r \not\in V_{\behaviour A_{[\behaviour B]}}^{max}$.
Let $C_{\pathLL r}$ be the clique for $S$ as defined in proposition~\ref{prop:Cp}: in particular $\pathLL r \in C_{\pathLL r}$. We remark that there is no $\daimon$-free path $\pathLL q$ such that $\pathLL q \in V_{\behaviour A_{[\behaviour B]}}^{max}$ and $\fullview{\pathLL q} \subset \fullview{\pathLL r}$ thus there is no path $\pathLL q\kappa^-$ in $C_{\pathLL r}$. It follows in particular that $\fullview{C_{\pathLL r}} \not\perp \fullview{\overline{\pathLL r\kappa^-}}^c$, contradiction. So $\pathLL r \in V_{\behaviour A_{[\behaviour B]}}^{max}$. Thereby $\pathLL{r}\kappa^-\daimon \in S$.

		Otherwise $\kappa^-$ is an action appearing in $V_{\behaviour A}$. Let $C_1$ (resp.\ $\kappa_0^+C''$) be a positively saturated maximal clique of $V_{\behaviour A_{[\behaviour B]}}$ (resp.\ of $V_{\behaviour B_{[\behaviour A]}}$). Then $C = C_1 \cup (C_1\cap V_{\behaviour A_{[\behaviour B]}}^{max})C''$ is a positively saturated maximal clique of $S$. Note that $\normalisationSeq{\fullview{C}}{\fullview{\overline{\pathLL r\kappa^-}}^c} = \normalisationSeq{\fullview{C_1}}{\fullview{\overline{\pathLL r\kappa^-}}^c}$ hence, as $\fullview{C} \perp \fullview{\overline{\pathLL r\kappa^-}}^c$, we have that $\fullview{C_1} \perp \fullview{\overline{\pathLL r\kappa^-}}^c$. So, as $V_{\behaviour A}$ satisfies negative saturation, $\pathLL r\kappa^-\daimon \in V_{\behaviour A_{[\behaviour B]}}$, thus $\pathLL{r}\kappa^-\daimon \in S$.

		\item Or $\pathLL{r} \in V_{\behaviour A_{[\behaviour B]}}^{max}V_{\behaviour B}^-$: $\pathLL r = \pathLL p \pathLL q$ with $\pathLL p \in V_{\behaviour A_{[\behaviour B]}}^{max}$ and $\pathLL q \in V_{\behaviour B}^-$. 

		If $\kappa^-$ is an action occurring in $V_{\behaviour B}^-$. Let $C_1$ be a positively saturated maximal clique of $V_{\behaviour A_{[\behaviour B]}}$ that contains $\pathLL p$. Let $C' = C_1 \cap V_{\behaviour A_{[\behaviour B]}}^{max}$. Let $\kappa_0^+C''$ be a positively saturated maximal clique of $V_{\behaviour B_{[\behaviour A]}}$. By lemma~\ref{lem:cliqueDecomposition}, $C = C_1 \cup C'C''$ is a positively saturated maximal clique of $S$. Hence $\fullview{C} \perp  \fullview{\overline{\pathLL r\kappa^-}}^c$. Note that $\normalisationSeq{\fullview{C}}{ \fullview{\overline{\pathLL r\kappa^-}}^c} = \pathLL p\normalisationSeq{\fullview{C''}}{ \fullview{\overline{\pathLL q\kappa^-}}^c}$. Thus $\fullview{C''} \perp  \fullview{\overline{\pathLL q\kappa^-}}^c$. Hence $\pathLL q\kappa^-\daimon \in V_{\behaviour B}^-$ as $V_{\behaviour B}$ satisfies negative saturation. It follows that $\pathLL{r}\kappa^-\daimon \in S$.

		Otherwise $\kappa^-$ is an action appearing in $V_{\behaviour A_{[\behaviour B]}}$. Let $C_1$ be a positively saturated maximal clique of $V_{\behaviour A_{[\behaviour B]}}$ that contains $\pathLL p$. Let $C' = C_1 \cap V_{\behaviour A_{[\behaviour B]}}^{max}$. Let $\kappa_0^+C''$ be a positively saturated maximal clique of $V_{\behaviour B_{[\behaviour A]}}$ that contains $\kappa_0^+\pathLL q$. By lemma~\ref{lem:cliqueDecomposition}, $C = C_1 \cup C'C''$ is a positively saturated maximal clique of $S$. Hence $\fullview{C} \perp  \fullview{\overline{\pathLL r\kappa^-}}^c$. Note that $\normalisationSeq{\fullview{C}}{ \fullview{\overline{\pathLL r\kappa^-}}^c} = \normalisationSeq{\fullview{C_1}}{ \fullview{\overline{\pathLL p\kappa^-}}^c}$. Thus $\fullview{C_1} \perp  \fullview{\overline{\pathLL p\kappa^-}}^c$. Hence $\pathLL p\kappa^-\daimon \in V_{\behaviour A_{[\behaviour B]}}$, contradiction with the fact that $\pathLL p \in V_{\behaviour A_{[\behaviour B]}}^{max}$.
\qed
		\end{itemize}
	\end{itemize}


\begin{lem}\label{lem:dualS_preludicable}
$\dual S$ is pre-ludicable.
\end{lem}
\proof
Let $\kappa_0^-$ be the first action of paths of ${\behaviour A_{[\behaviour B]}}^\perp$. 
	\begin{itemize}
	\item By definition of $\dual S$ and the fact that a set of visitable paths of a behaviour is ludicable, $\dual S$ is prefix-closed and daimon-closed.
	\item (positive saturation) Let $\pathLL r$ be a path of $\dual S$, we prove that $D_{\pathLL r}$ is a positively saturated maximal clique for $\dual S$ (where $D_{\pathLL r}$ for $\dual S$ is as defined in proposition~\ref{prop:Cp}). It follows from the definition of $D_{\pathLL r}$ that $D_{\pathLL r}$ is a maximal clique for $\dual S$. Let $\pathLL m \in D_{\pathLL r}$, $\pathLL n\kappa^-\kappa^+ \in D_{\pathLL r}$, $\pathLL m\kappa^-\daimon \in \dual S$, we have to prove that $\pathLL m \kappa^-\kappa^+ \in \dual S$. Clearly we can suppose that $\kappa^+ \neq \daimon$, otherwise the result is trivial.
		\begin{itemize}
		\item Either $\pathLL m \in \dual{V_{\behaviour A_{[\behaviour B]}}}$ and $\pathLL n\kappa^-\kappa^+ \in \dual{V_{\behaviour A_{[\behaviour B]}}}$. As $\dual{V_{\behaviour A_{[\behaviour B]}}}$ satisfies positive saturation, there exists a positively saturated maximal clique $D'$ for $\dual{V_{\behaviour A_{[\behaviour B]}}}$ such that $\pathLL n \kappa^-\kappa^+ \in D'$. 

			If $\pathLL r \in \dual{V_{\behaviour A_{[\behaviour B]}}}$: Let $D'_{\pathLL r}$ for $\dual{V_{\behaviour A_{[\behaviour B]}}}$ be as defined in proposition~\ref{prop:Cp}. Note that we have $\pathLL m \in D'_{\pathLL r}$, $\pathLL n\kappa^-\kappa^+ \in D'_{\pathLL r}$. Furthermore $\pathLL m\kappa^-\daimon \in \dual{V_{\behaviour A_{[\behaviour B]}}}$. Then, as $\dual{V_{\behaviour A_{[\behaviour B]}}}$ satisfies positive saturation, $\pathLL m\kappa^-\kappa^+ \in \dual{V_{\behaviour A_{[\behaviour B]}}}$, thus $\pathLL m\kappa^-\kappa^+ \in \dual S$.

			Otherwise $\pathLL r = \pathLL s\pathLL t$ with $\pathLL s \in \overline{V_{\behaviour A_{[\behaviour B]}}^{max}}$ and $\pathLL t \in \dual{V_{\behaviour B}^-}$. Then $\pathLL s\daimon \in \dual{V_{\behaviour A_{[\behaviour B]}}}$: Let $D''_{\pathLL s\daimon}$ for $\dual{V_{\behaviour A_{[\behaviour B]}}}$ be as defined in proposition~\ref{prop:Cp}. Note that we have $\pathLL m \in D''_{\pathLL s\daimon}$, $\pathLL n\kappa^-\kappa^+ \in D''_{\pathLL s\daimon}$. Furthermore $\pathLL m\kappa^-\daimon \in \dual{V_{\behaviour A_{[\behaviour B]}}}$. Then, as $\dual{V_{\behaviour A_{[\behaviour B]}}}$ satisfies positive saturation, $\pathLL m\kappa^-\kappa^+ \in \dual{V_{\behaviour A_{[\behaviour B]}}}$, thus $\pathLL m\kappa^-\kappa^+ \in \dual S$.

		\item Or $\pathLL m \in \overline{V_{\behaviour A_{[\behaviour B]}}^{max}}\dual{V_{\behaviour B}^-}$ and $\pathLL n\kappa^-\kappa^+ \in \dual{V_{\behaviour A_{[\behaviour B]}}}$. But we have $\pathLL m \kappa^-\daimon \in \dual S$ thus, by definition of $\dual S$, the action $\kappa^-$ should be an action in $\dual{V_{\behaviour B}^-}$: contradiction with the fact that $\pathLL n\kappa^-\kappa^+ \in \dual{V_{\behaviour A_{[\behaviour B]}}}$.

		\item Or $\pathLL m \in \overline{V_{\behaviour A_{[\behaviour B]}}^{max}}\dual{V_{\behaviour B}^-}$ and $\pathLL n\kappa^-\kappa^+ \in \overline{V_{\behaviour A_{[\behaviour B]}}^{max}}\dual{V_{\behaviour B}^-}$. As we have $\pathLL m \kappa^-\daimon \in \dual S$ thus, by definition of $\dual S$, $\kappa^-$ is an action in $\dual{V_{\behaviour B}^-}$. Hence as $\pathLL n\kappa^-\kappa^+ \in \overline{V_{\behaviour A_{[\behaviour B]}}^{max}}\dual{V_{\behaviour B}^-}$, the action $\kappa^+$ is also an action in $\dual{V_{\behaviour B}^-}$.
		Let us write $\pathLL m = \pathLL m_0\pathLL m_1$ and $\pathLL n = \pathLL n_0\pathLL n_1$ where $\pathLL m_0$ and $\pathLL n_0$ are elements of $\overline{V_{\behaviour A_{[\behaviour B]}}^{max}}$ and $\pathLL m_1$ and $\pathLL n_1$ are elements of $\dual{V_{\behaviour B}^-}$.

		Remark that $\pathLL r \in \overline{V_{\behaviour A_{[\behaviour B]}}^{max}}\dual{V_{\behaviour B}^-}$, hence is of the form $\pathLL s\pathLL t$ where $\pathLL s \in \overline{V_{\behaviour A_{[\behaviour B]}}^{max}}$ and $\pathLL t \in \dual{V_{\behaviour B}^-}$. Remark then that $\pathLL t$ begins with a positive action, furthermore $\fullview{\pathLL t}$ is a slice, \ie, it has a unique first action. Then $\pathLL m_0$ and $\pathLL n_0$ ends on the same (negative) action (otherwise their next action should be distinct). Let us consider the clique $D'_{\kappa_0^-\pathLL t}$ for $\dual{V_{\behaviour B_{[\behaviour A]}}}$ as defined in proposition~\ref{prop:Cp}. Remark that $\kappa_0^-\pathLL m_1 \in D'_{\kappa_0^-\pathLL t}$, $\kappa_0^-\pathLL n_1\kappa^-\kappa^+ \in D'_{\kappa_0^-\pathLL t}$ and that $\kappa_0^-\pathLL m_1\kappa^-\daimon \in \dual{V_{\behaviour B_{[\behaviour A]}}}$ hence as $\dual{V_{\behaviour B_{[\behaviour A]}}}$ satisfies positive saturation, $\kappa_0^-\pathLL m_1\kappa^-\kappa^+ \in \dual{V_{\behaviour B_{[\behaviour A]}}}$. Thus $\pathLL m_0\pathLL m_1\kappa^-\kappa^+ \in \overline{V_{\behaviour A_{[\behaviour B]}}^{max}}\dual{V_{\behaviour B_{[\behaviour A]}}}$.

		\item Or $\pathLL m \in \dual{V_{\behaviour A_{[\behaviour B]}}}$ and $\pathLL n\kappa^-\kappa^+ \in \overline{V_{\behaviour A_{[\behaviour B]}}^{max}}\dual{V_{\behaviour B}^-}$.
		Remark that $\kappa^-$ cannot be an action present in $\dual{V_{\behaviour B}^-}$ as $\pathLL m\kappa^-\daimon$ is legal: in such a case $\kappa^-$ should be justified by $\kappa_0^-$, contradiction.
		Hence $\kappa^- \in \overline{V_{\behaviour A_{[\behaviour B]}}}$ and $\kappa^+ \in \dual{V_{\behaviour B}^-}$. As $\pathLL m\kappa^-\daimon \in \dual S$ then $\overline{\pathLL m\kappa^-} \in V_{\behaviour A_{[\behaviour B]}}$, furthermore $\overline{\pathLL n\kappa^-} \in V_{\behaviour A_{[\behaviour B]}}^{max}$, and paths $\overline{\pathLL m\kappa^-}$ and $\overline{\pathLL n\kappa^-}$ end on the same positive action. 

		As $\pathLL n \kappa^-\kappa^+$ is a path in $D_{\pathLL r}$, and being given the structure of $D_{\pathLL r}$, the path $\pathLL r$ is in $\overline{V_{\behaviour A_{[\behaviour B]}}^{max}}\dual{V_{\behaviour B}^-}$. Let $\pathLL r = \pathLL s\pathLL t$ with $\pathLL s \in \overline{V_{\behaviour A_{[\behaviour B]}}^{max}}$ and $\pathLL t \in \dual{V_{\behaviour B}^-}$. 
		We notice that $\kappa^+$ is necessarily the first action of $\pathLL t$ as all actions of $\pathLL t$ are actions occurring in $\dual{V_{\behaviour B}^-}$ and $\pathLL t$ is a legal path. Hence $\kappa^-$ is necessarily the last action of $\pathLL s$.
		Finally, being given the structure of $D_{\pathLL r}$, $\fullview{\pathLL m\kappa^-} \subset \fullview{\pathLL s}$. So constraint $(C)$ may be applied: we have that $\overline{\pathLL m\kappa^-} \in V_{\behaviour A_{[\behaviour B]}}^{max}$.
So $\pathLL m\kappa^-\kappa^+ \in \overline{V_{\behaviour A_{[\behaviour B]}}^{max}}\dual{V_{\behaviour B_{[\behaviour A]}}}$.
		\end{itemize}

	\item (negative saturation)\\
		Let $\pathLL{r} \in \dual S$ and $\pathLL{r}\kappa^-\daimon$ be a legal path such that for all positively saturated maximal clique $D$ of $\dual S$, we have that $\fullview{D} \perp \fullview{\overline{\pathLL r\kappa^-}}^c$. We have to prove that $\pathLL{r}\kappa^-\daimon \in \dual S$. We use lemma~\ref{lem:dualcliqueDecomposition} for decomposing such a saturated maximal clique $D$ of $\dual S$: $D = D_1 \cup \bigcup_{\kappa_1^- \in K} D'_{\kappa_1^-} D''_{\kappa_1^-}$.
		\begin{itemize}
		\item Either $\pathLL{r} \in \dual{V_{\behaviour A_{[\behaviour B]}}}$ then $\pathLL r \in D_1$ that is a positively saturated maximal clique of $\dual{V_{\behaviour A_{[\behaviour B]}}}$. 
		Remark that $\kappa^-$ cannot be an action appearing in $\dual{V_{\behaviour B}^-}$: as $\pathLL{r} \in \dual{V_{\behaviour A_{[\behaviour B]}}}$, $\kappa^-$ should be the first action in $\pathLL{r}\kappa^-\daimon$ from $\dual{V_{\behaviour B}^-}$ but by construction such an action should be positive.
		Hence $\kappa^-$ is an action appearing in $\dual{V_{\behaviour A}}$. Note that $\normalisationSeq{\fullview{D}}{\fullview{\overline{\pathLL r\kappa^-}}^c} = \normalisationSeq{\fullview{D_1}}{\fullview{\overline{\pathLL r\kappa^-}}^c}$ hence, as $\fullview{D} \perp \fullview{\overline{\pathLL r\kappa^-}}^c$, we have that $\fullview{D_1} \perp \fullview{\overline{\pathLL r\kappa^-}}^c$. So, as $\dual{V_{\behaviour A}}$ satisfies negative saturation, $\pathLL r\kappa^-\daimon \in \dual{V_{\behaviour A_{[\behaviour B]}}}$, thus $\pathLL{r}\kappa^-\daimon \in \dual S$.

		\item Or $\pathLL{r} \in \overline{V_{\behaviour A_{[\behaviour B]}}^{max}}\dual{V_{\behaviour B}^-}$: $\pathLL r = \pathLL p \pathLL q$ with $\pathLL p \in \overline{V_{\behaviour A_{[\behaviour B]}}^{max}}$ and $\pathLL q \in \dual{V_{\behaviour B}^-}$. 

		If $\kappa^-$ is an action occurring in $\dual{V_{\behaviour B}^-}$. We prove that $\kappa_0^-\pathLL q\kappa^-\daimon \in \dual{V_{\behaviour B}}$, hence $\pathLL p\pathLL q\kappa^-\daimon \in \dual S$. 
		Note that $\pathLL p$ ends with a negative action: $\pathLL p = \pathLL p'\kappa'^-$. Let $E_1$ be a positively saturated maximal clique of $\dual{V_{\behaviour A_{[\behaviour B]}}}$ that contains $\pathLL p'$.  Let us define the cliques $E'_{\kappa_1^-}$ as in lemma~\ref{lem:dualcliqueDecomposition}. Remark that there should exist a clique $E'_{\kappa_2^-}$ among these $E'_{\kappa_1^-}$ that contains the path $\pathLL p$. 
		Let $\kappa_0^-E''_{\kappa_1^-}$ be positively saturated maximal cliques of $\dual{V_{\behaviour B_{[\behaviour A]}}}$ satisfying constraints of lemma~\ref{lem:dualcliqueDecomposition}, then $E = E_1 \cup \bigcup_{\kappa_1^- \in K} E'_{\kappa_1^-} E''_{\kappa_1^-}$ is a positively saturated maximal clique of $\dual S$. Hence $\fullview{E} \perp  \fullview{\overline{\pathLL r\kappa^-}}^c$. 
		Note that $\normalisationSeq{\fullview{E}}{ \fullview{\overline{\pathLL r\kappa^-}}^c} = \pathLL p\normalisationSeq{\fullview{E''_{\kappa_2^-}}}{ \fullview{\overline{\pathLL q\kappa^-}}^c}$. Thus $\fullview{E''_{\kappa_2^-}} \perp  \fullview{\overline{\pathLL q\kappa^-}}^c$. Hence $\pathLL q\kappa^-\daimon \in \dual{V_{\behaviour B}^-}$. It follows that $\pathLL{r}\kappa^-\daimon \in \dual S$.

		Otherwise $\kappa^-$ is an action appearing in $\overline{V_{\behaviour A_{[\behaviour B]}}}$. 
		Note that $\pathLL q$ begins and ends with a positive action. Hence $\pathLL q$ is not empty otherwise $\pathLL p\kappa^-\daimon$ is not a path. Remark also that the only justifier of an action of $\pathLL q$ that is not in $\pathLL q$ is $\kappa_0^-$, then $\view{\overline{\pathLL p\pathLL q\kappa^-}}$ should be of the form $\overline{\kappa_0^-\pathLL q'\kappa^-}$ where $\pathLL q'$ is a subsequence of $\pathLL q$. As $\overline{\kappa^-}$ cannot be justified by $\overline{\kappa_0^-}$, the sequence $\overline{\pathLL p\pathLL q\kappa^-}$ cannot be a path, contradiction.
\qed
		\end{itemize}
	\end{itemize}


\restateOrientedTensor*
\begin{proof}
The proposition follows from lemmas~\ref{lem:S_legalPaths},~\ref{lem:S_preludicable} and~\ref{lem:dualS_preludicable}.
\end{proof}


\subsection{An absolute non-commutative connective: first algebraic properties.}

In this subsection we consider that condition $(C)$ may not be fulfilled, in that case the connective $\fullotensor$ may be defined as the ludicable closure of the set of oriented visitable paths as given in the previous subsection. 
We can then characterize the set of designs in $\behaviour A \fullotensor \behaviour B$, and more precisely we have internal completeness for this connective (proposition~\ref{prop:orientedBehaviour}). 
For ease of reading, we use the following notations:
\begin{defi}
Let $\behaviour A$ and $\behaviour B$ be behaviours,
	\begin{itemize}
	\item $\behaviour A^+ = \{\design D[\view{\pathLL p}\daimon / \view{\pathLL p\kappa^+}\chronicle d ~;~ \pathLL p\kappa^+ \in V_{\behaviour A}^{max}] ~;~ \design D \in \behaviour A\}$: $\behaviour A^+$ is obtained from $\behaviour A$ by replacing in each design of $\behaviour A$ all chronicles of the form $\view{\pathLL p\kappa^+}\chronicle d$ by the chronicle $\view{\pathLL p}\daimon$ when $\pathLL p\kappa^+$ is in $V_{\behaviour A}^{max}$. $|\behaviour A|^+$ is obtained in a similar way from the incarnation $|\behaviour A|$.
	\item $|\behaviour A|^{max}$ is the set of designs of the incarnation $|\behaviour A|$ which includes at least one path of $V_{\behaviour A}^{max}$.
	\item ${\boldtop}^+_{[\behaviour B]} = \{\{\kappa^+\} ~;~ \kappa^+$ is the first action of designs of $\behaviour B\}^{\perp\perp}$: the behaviour $\boldtop^+_{[\behaviour B]}$ contains all designs that may be built beginning with a first action of a design of $\behaviour B$, hence its unique visitable paths are $\daimon$ and actions $\kappa^+$.
	\end{itemize}
\end{defi}

Remark that $\behaviour A^+$ and ${\boldtop}^+_{[\behaviour B]}$ are positive behaviours.
Note also that, when $V_{\behaviour A}^{max} = \emptyset$, cliques of visitable paths of $\behaviour A \fullotensor \behaviour B$ are in fact cliques of visitable paths of $\behaviour A_{[\behaviour B]}$: there is no way to ``visit'' by interaction a design of $\behaviour B$.

\begin{prop}\label{prop:orientedBehaviour}
$\behaviour A \fullotensor \behaviour B = \behaviour A \otimes \behaviour B ~\cup~ \behaviour A^+ \otimes\boldtop^+_{[\behaviour B]}$
and
$|\behaviour A \fullotensor \behaviour B| = |\behaviour A|^{max} \otimes |\behaviour B| ~\cup~ |\behaviour A|^+ \otimes\boldtop^+_{[\behaviour B]}$
\end{prop}
\begin{proof}
We consider the non-additive case, the additive case follows easily.
Note that if $V_\behaviour A = \{\daimon\}$ then $\behaviour A \fullotensor \behaviour B = \{\{\daimon\}\}$ hence equalities hold.
Suppose that $V_{\behaviour A} \neq \{\daimon\}$. 
Let us note $C_{\behaviour A}$ the set of paths $\pathLL p$ of $V_{\behaviour A}$ for which there exists $\pathLL q \in V_{\behaviour A}^{max}$ such that $\pathLL p$ and $\pathLL q$ end on the same (positive) action and $\fullview{\overline{\pathLL p}} \subset \fullview{\overline{\pathLL q}}$. \Ie, $C_{\behaviour A}$ is the set of paths of $V_{\behaviour A}$ that invalidate condition $(C)$ for $\behaviour A$.
Remark then that the ludicable closure of $S=  V_{\behaviour A_{[\behaviour B]}} \cup V_{\behaviour A_{[\behaviour B]}}^{max}V_{\behaviour B}^-$ is the set $T = V_{\behaviour A_{[\behaviour B]}} \cup \bigcup_{\pathLL p_1 \in C_{\behaviour A}, \pathLL p_1\pathLL p_2 \in V_{\behaviour A}} \pathLL p_1(V_{\behaviour B}^- \lshuffle \pathLL p_2)$.
The incarnation $|\behaviour A \fullotensor \behaviour B|$ is the set of designs $\fullview{D}$ where $D$ is a maximal clique of $T$ such that $\dual D$ is finite-stable and saturated. The result follows as a clique that contains a path in $V_{\behaviour A}^{max}$ contains also extensions of paths in $V_{\behaviour B}^-$, the ludicable closure does not add new designs in the incarnation.
The full behaviour $\behaviour A \fullotensor \behaviour B$ is obtained by adding extensions to designs in the incarnation $|\behaviour A \fullotensor \behaviour B|$.
\end{proof}

\begin{prop}
The connective $\fullotensor$ is associative:
$(\behaviour A \fullotensor \behaviour B) \fullotensor \behaviour C = \behaviour A \fullotensor (\behaviour B \fullotensor \behaviour C)$\\
The behaviour $\bf 1$ is a neutral element: $\bf 1 \fullotensor \behaviour A = \behaviour A \fullotensor \bf 1 = \behaviour A$.
\end{prop}
\begin{proof}
Remark first that $V_{\behaviour A \fullotensor \behaviour B}^{max} = V_{\behaviour A_{[\behaviour B]}}^{max}V_{\behaviour B}^{-max}$. 
Thus $(\behaviour A \fullotensor \behaviour B)^+ = \behaviour A \otimes \behaviour B^+ ~\cup~ \behaviour A^+ \otimes\boldtop^+_{[\behaviour B]}$.
Note also that $\behaviour C^+ \subset \behaviour C \subset \boldtop^+_{[\behaviour C]}$.
\begin{align*}
(\behaviour A \fullotensor \behaviour B) \fullotensor \behaviour C
	&= (\behaviour A \otimes \behaviour B ~\cup~ \behaviour A^+ \otimes\boldtop^+_{[\behaviour B]}) \fullotensor \behaviour C \\
	&= (\behaviour A \otimes \behaviour B ~\cup~ \behaviour A^+ \otimes\boldtop^+_{[\behaviour B]}) \otimes \behaviour C  ~~\cup~~ \behaviour A \otimes \behaviour B^+ \otimes \boldtop^+_{[\behaviour C]} ~\cup~ \behaviour A^+ \otimes\boldtop^+_{[\behaviour B]} \otimes \boldtop^+_{[\behaviour C]}\\
	&= \behaviour A \otimes \behaviour B \otimes \behaviour C ~\cup~ \behaviour A^+ \otimes\boldtop^+_{[\behaviour B]} \otimes \behaviour C ~\cup~ \behaviour A \otimes \behaviour B^+ \otimes \boldtop^+_{[\behaviour C]} ~\cup~ \behaviour A^+ \otimes\boldtop^+_{[\behaviour B]} \otimes \boldtop^+_{[\behaviour C]}\\
	&= \behaviour A \otimes \behaviour B \otimes \behaviour C ~\cup~ \behaviour A \otimes \behaviour B^+ \otimes \boldtop^+_{[\behaviour C]} ~\cup~ \behaviour A^+ \otimes\boldtop^+_{[\behaviour B]} \otimes \boldtop^+_{[\behaviour C]}
\end{align*}
and
\begin{align*}
\behaviour A \fullotensor (\behaviour B \fullotensor \behaviour C)
	&= \behaviour A \fullotensor (\behaviour B \otimes \behaviour C ~\cup~ \behaviour B^+ \otimes \boldtop^+_{[\behaviour C]}) \\
	&= \behaviour A \otimes \behaviour B \otimes \behaviour C ~\cup~ \behaviour A \otimes \behaviour B^+ \otimes \boldtop^+_{[\behaviour C]} ~\cup~ \behaviour A^+ \otimes \boldtop^+_{[\behaviour B]} \otimes \boldtop^+_{[\behaviour C]}
\end{align*}
Hence the result.\\
Remind that $\bf 1$ is neutral for $\otimes$. Hence $\behaviour A \otimes \bf 1 = \bf 1 \otimes \behaviour A = \behaviour A$. Furthermore $\bf 1^+ = \{\daimon\}$ thus $\bf 1^+ \otimes \boldtop^+_{[\behaviour A]} = \{\daimon\}$. It follows that $\bf 1 \fullotensor \behaviour A = \behaviour A$. Finally $\boldtop^+_{[\bf 1]} = \bf 1$ and $\behaviour A^+ \subset \behaviour A$ thus $\behaviour A \fullotensor \bf 1 = \behaviour A$.
\end{proof}

\section{Conclusion}\label{sec:conclusion}

At first sight, Ludics objects seem to be easy to study: designs are nothing else but abstraction of proofs or counter-proofs of multiplicative-additive Linear Logic (MALL). However, this is not the case when one tries to identify among behaviours, \ie, closures of sets of designs, those that are interpretation of MALL formulas.
This paper is a first step toward a full algebraic study of behaviours. First, we make explicit the equivalence between the two presentations of a design, as set of paths versus set of chronicles. 
We give a few properties concerning orthogonality in terms of path traversal, introducing visitable paths, \ie, paths that are visited by orthogonality.
Our main result is a characterization of finite MALL formulas, \ie, formulas built from the linear constants by means of additive and multiplicative connectives.
In particular, we show that such behaviours should be {\em regular}. Regularity is in fact a global property of visitable paths: roughly speaking, legal paths built from actions in the incarnation should be visitable.
Such a study should help understanding the structure of MALL proofs. 
By contrast, many behaviours are not regular. We analyse one case of non-regularity: in short, a proof of $\behaviour A \fullotensor \behaviour B$ is a proof of $\behaviour A$ followed by a proof of $\behaviour B$. We show that such a situation may be fully defined. This should be considered as a first step towards a full study of orientation in Ludics.

Let us remark finally that properties of Ludics that serve for proving that Ludics is a fully abstract model of (slightly modified polarized second-order) MALL, are satisfied for the entire Ludics and not only for behaviours interpreting MALL formulas: interaction between objects, that is cut-elimination, is at the heart of Ludics, thus it allows to consider the Ludics framework as a semantics for computation beyond what is given with MALL: a behaviour may model a type and (open) interaction between behaviours corresponds to composition of types. For future work, we plan to extend our analysis to the whole set of behaviours, defining a grammar for it in such a way that connectives of the grammar may be computationally (or logically) interpreted.

We thank anonymous referees for their thorough reviews and highly appreciate the comments and 
suggestions.

\bibliographystyle{alpha}
\bibliography{Ludics,langue,publi-MQ,GameSemanticsLL,NCLL}

\appendix
\section{}\label{sec:annex}

\subsection{Ludics: Basic Definitions}\label{annex:basicLudics}

\begin{defi}[Base]
A {\tt base} is  a non-empty finite set of sequents: $\Gamma_1 \vdash \Delta_1, \dots, \Gamma_n \vdash \Delta_n$ such that each $\Delta_j$ is a finite set of addresses, at most one $\Gamma_i$ may be empty and the other $\Gamma_i$ contain each exactly one address.
Furthermore if an address appears twice then one occurrence is in one of $\Gamma_i$ of a sequent and the other in one of $\Delta_j$ of another sequent, otherwise an address appears only once.
\end{defi}

 \begin{defi}[Action]
An {\tt action} $\kappa$ is
\begin{itemize}
\item either a positive proper action $(+,\xi,I)$ or a negative proper action $(-,\xi,I)$ where the address $\xi$ is said the {\em focus} of the action, and the finite set of integers $I$ is said its {\em ramification},
\item or the positive action daimon written $\daimon$.
\end{itemize}
The notation $\overline{\kappa}$ may be extended to sequences of actions\footnote{The empty sequence is noted $\epsilon$.} by $\Overline{\epsilon} = \epsilon$ and $\Overline{w\kappa} = \Overline{w}\,\Overline{\kappa}$.
An address $\xi.i$ is {\em justified} by an action $(+,\xi,I)$ when $i \in I$. By extension an action $\kappa = (\pi,\xi.i,J)$ is {\em justified} by an action $\overline{\kappa} = (\overline{\pi},\xi,I)$ when $i \in I$, $\pi \in \{+,-\}$, $\overline{+} = -$ and $\overline{-} = +$.
When $w$ is a $\daimon$-free sequence of actions, we write also $\dual{w} = \overline{w}\daimon$ and $\dual{w\daimon} = \overline{w}$.
\end{defi}

\begin{defi}[Chronicle]
A {\tt chronicle} $\chronicle{c}$ based on $\Gamma \vdash \Delta$ is a non-empty and finite alternate sequence of actions such that
\begin{itemize}
\item {\em Positive proper action:} A  positive proper action is either justified, \ie, its focus is built by one of the previous actions in the sequence, or it is called initial. 
\item {\em Negative action:} A negative action may be initial, in such a case it is the first action of the chronicle and its focus is in $\Gamma$. Otherwise it is justified by the immediate previous positive action.
\item {\em Linearity:} Actions have distinct foci.
\item {\em Daimon:} If present, a daimon ends the chronicle. 
\item {\em Polarity:} If $\Gamma$ is empty, the first action of $\chronicle c$ is positive, otherwise it is negative.
\end{itemize}
\end{defi}

\begin{defi}[Coherence on Chronicles]
Two chronicles $\chronicle{c}_1$ and $\chronicle{c}_2$ are {\tt coherent}, noted $\chronicle{c}_1 \coh \chronicle{c}_2$, when the two following conditions are satisfied:
\begin{itemize}
\item {\em Comparability:}  Either one extends the other or they first differ on negative actions, \ie, if $w\kappa_1 \coh w\kappa_2$ then either $\kappa_1 = \kappa_2$ or $\kappa_1$ and $\kappa_2$ are negative actions.
\item {\em Propagation:} When they first differ on negative actions and these negative actions have distinct foci then the foci of following actions in $\chronicle{c}_1$ and $\chronicle{c}_2$ are pairwise distinct, \ie, if $w(-,\xi_1,I_1)w_1\kappa_1 \coh w(-,\xi_2,I_2)w_2\kappa_2$ with $\xi_1 \neq \xi_2$ then $\kappa_1$ and $\kappa_2$ have distinct foci.
\end{itemize}
\end{defi}

\begin{defi}[Designs, Slices, Nets]
A {\tt design} $\design{D}$, based on $\Gamma\vdash\Delta$, is a set of chronicles based on $\Gamma\vdash\Delta$, such that the following conditions are satisfied:
\begin{itemize}
\item {\em Forest:} The set of chronicles is prefix closed.
\item {\em Coherence:} The set is a clique of chronicles with respect to $\coh$.
\item {\em Positivity:} A chronicle without extension in $\design{D}$ ends with a positive action.
\item {\em Totality:} $\design{D}$ is non-empty when the base is positive, in that case all the chronicles begin with a (unique) positive action.
\end{itemize}
A {\tt slice} is a design $\design{S}$ such that if $w(-,\xi,I_1), w(-,\xi,I_2) \in \design{S}$ then $I_1 = I_2$.\\
A {\tt net} is a finite set of designs on disjoint bases. 
\end{defi}

A design is then a set of chronicles or a forest of actions (when one cares of justification between actions) that satisfies several constraints. It can also be presented as a sequent tree, however with ambiguity due to the possible weakening of addresses created by actions. We describe in example~\ref{exa:designTreeSet} a design based on $\vdash \xi$ as a set of chronicles (on the left) and as a sequent tree (on the right):

\begin{exa}\label{exa:designTreeSet}~\\
\begin{minipage}{.5\textwidth}
$$
 \begin{array}{cl}
 \{ &  (+,\xi,\{1,3\});\\
              & (+,\xi,\{1,3\})(-,\xi.1,\{0\});\\
                   & (+,\xi,\{1,3\})(-,\xi.1,\{0\})(+,\xi.1.0,\{0\});\\
                   &  (+,\xi,\{1,3\})(-,\xi.1,\{1\});\\
                    &  (+,\xi,\{1,3\})(-,\xi.1,\{1\})(+,\xi.1.1,\{0\});\\
                    &  (+,\xi,\{1,3\})(-,\xi.3,\{0\});\\
                    &  (+,\xi,\{1,3\})(-,\xi.3,\{0\})(+,\xi.3.0,\emptyset) \,\,\,\,\}
                        \end{array}
$$
\end{minipage}
\begin{minipage}{.5\textwidth}
$$
	\infer{\vdash \xi}
{
		\infer{\xi.1 \vdash}
			{
			\infer{\vdash \xi.1.0}
				{
				\xi.1.0.0 \vdash
				}
				&
				\infer{\vdash \xi.1.1}
				{
				\xi.1.1.0 \vdash
				}
			}
		&
		\infer{\xi.3 \vdash}
			{
			\infer{\vdash \xi.3.0}
				{
				\vdash
				}
			}
		}
	$$	
\end{minipage}
\end{exa}
 
\begin{defi}[Closed cut-net]
A net of designs $\design{R}$ is a {\tt closed cut-net} if 
\begin{itemize}
\item addresses in bases are either distinct or present twice, once in a left part of a base and once in a right part of another base,
\item the net of designs is acyclic and connected with respect to the graph of bases and cuts.
\end{itemize}
An address present in a left part and in a right part defines a {\tt cut}. 
In a closed cut-net, the (unique) design whose base is positive is called the main design of the cut-net. 
\end{defi}

\begin{defi}[Interaction on closed cut-nets]
Let $\design{R}$ be a closed cut-net. 
The design resulting from the interaction, denoted by $\normalisation{\design{R}}$, is defined in the following way:
let $\design{D}$ be the main design of $\design{R}$, with first action $\kappa$,
\begin{itemize}
\item if $\kappa$ is a daimon, then $\normalisation{\design{R}} = \{\daimon\}$,
\item otherwise $\kappa$ is a proper positive action $(+,\sigma,I)$ such that $\sigma$ is part of a cut with another design with last rule $(-,\sigma, {\mathcal N})$ (aggregating ramifications of actions with the same focus $\sigma$):
	\begin{itemize}
	\item If $I \not\in \mathcal N$, then interaction fails.
	\item Otherwise, interaction follows with the connected part of subdesigns obtained from $I$ with the rest of $\design{R}$.
	\end{itemize}
\end{itemize}
\end{defi}

Following this definition, either interaction fails, or it does not end, or it results in the design $\dai = \{\daimon\}$.
The definition of orthogonality follows:

\begin{defi}[Orthogonal, Behaviour]~
\begin{itemize}
\item Let $\design{D}$ be a design of base $\xi \vdash \sigma_1, \dots, \sigma_n$ (resp. $\vdash \sigma_1, \dots, \sigma_n$),
let $\design{R}$ be the net of designs $(\design{A}, \design{B}_1, \dots, \design{B}_n)$ (resp. $\design{R} = (\design{B}_1, \dots, \design{B}_n)$), where $\design{A}$ has base $\vdash \xi$ and $\design{B}_i$ has base $\sigma_i \vdash$,
then $\design{R}$ belongs to $\design{D}^\perp$ if $\psdes{\design{D}}{\design{R}} = \dai$. 
\item Let $\designset{E}$ be a set of designs of the same base, $\designset{E}^\perp = \bigcap_{\design{D} \in \designset{E}} \design{D}^\perp$.
\item $\designset{E}$ is a {\tt behaviour} if $\designset{E} = \designset{E}^{\perp\perp}$.
A behaviour is {\em positive} (resp. {\em negative}) if the base of its designs is positive (resp. negative).
\end{itemize}
\end{defi}

\begin{defi}[Interaction path]
Let $(\design{D},\design{R}$) be a convergent, \ie, orthogonal, closed cut-net such that all the cut loci belong to the base of $\design{D}$. The {\tt interaction path} of $\design{D}$ with $\design{R}$, denoted $\normalisationSeq{\design{D}}{\design{R}}$, is the sequence of actions of $\design{D}$ visited during the normalization. The construction goes as follows where $n$ is the number of normalization steps so far obtained:
Let $\kappa_1\dots\kappa_n$ be the prefix of $\normalisationSeq{\design{D}}{\design{R}}$ already defined (or the empty sequence if $n=0$).
\begin{itemize}
\item Either the interaction stops: if the main design is a subdesign of $\design{R}$ then $\normalisationSeq{\design{D}}{\design{R}}=\kappa_1\dots\kappa_n$, otherwise the main design is a subdesign of $\design{D}$ then $\normalisationSeq{\design{D}}{\design{R}}=\kappa_1\dots\kappa_n\daimon$.
\item Or, let $\kappa^+$ be the first proper action of the closed cut-net obtained after step $n$, $\normalisationSeq{\design{D}}{\design{R}}$ begins with $\kappa_1\dots\kappa_n\overline{\kappa^+}$ if the main design is a subdesign of $\design{R}$, or it begins with $\kappa_1\dots\kappa_n\kappa^+$ if the main design is a subdesign of $\design{D}$.
\end{itemize}
\noindent We note $\normalisationSeq{\design{R}}{\design{D}}$ the sequence of actions visited in $\design{R}$ during the normalization with $\design{D}$.
\end{defi}

It follows from the definition that $\normalisationSeq{\design D}{\design R} = \dual{\normalisationSeq{\design R}{\design D}}$.

\begin{defi}[Incarnation]\label{incarnationdessein}
Let $\behaviour{B}$ be a behaviour, $\design{D}$ be a design in $\behaviour{B}$. 
\begin{itemize}
\item The {\tt incarnation} $\Dincarnation{D}{\behaviour{B}}$ of $\design{D}$ with respect to the behaviour $\behaviour{B}$ is $\bigcup_{\design{R} \in \behaviour{B}^\perp} \normalisationDes{D}{R}$. 
\item The {\tt incarnation} $\Bincarnation{\behaviour{B}}$ of a behaviour $\behaviour{B}$ is the set $\{\Dincarnation{D}{\behaviour{B}} ~;~ \design D \in \behaviour B\}$.
\end{itemize}
$\Dincarnation{D}{\behaviour{B}}$ is simply noted $\Dincarnation{D}{}$ when $\behaviour{B}$ is clear from the context.
\end{defi}

\begin{defi}[Daimon closure]
Let $\design{D}$ be a design, the {\tt daimon closure} of $\design{D}$, written ${\design{D}}^\daimon$,
is the set of designs obtained from $\design{D}$ by substituting, for some set of  negative-ended chronicles $\chronicle{c}\in\design{D}$, all the chronicles $\chronicle{c}\kappa^+w\in\design{D}$ by the chronicles $\chronicle{c}\daimon$.\\
Let $\designset{E}$ be a set of designs of the same base, the daimon closure of $\designset{E}$ noted ${\designset{E}}^\daimon$ is the set $\bigcup_{\design{D}\in\designset{E}} \design{D}^\daimon$.
\end{defi}

\subsection{Proofs of Subsection~\ref{sec:reg_beh}}\label{subsec:proofs}

We prove first the inversion property of legality with respect to shuffle: a legal shuffle of paths is a shuffle of legal paths.

\begin{lem}\label{lem:dual_inverse_shuffle}
Let $\pathLL{p}$ and $\pathLL{q}$ be two positive-ended paths and $\pathLL{r} \in \pathLL{p} \shuffle \pathLL{q}$ be a path, if $\pathLL{r}$ is legal then also $\pathLL{p}$ and $\pathLL{q}$ are legal.
\end{lem}
\proof
We prove the result by contradiction: let us suppose that $\pathLL{p}$ and $\pathLL{q}$ are two paths, $\pathLL{r} \in \pathLL{p} \shuffle \pathLL{q}$ is a path such that its dual $\dual{\pathLL{r}}$ is a path, and at least one of the duals $\dual{\pathLL{p}}$ and $\dual{\pathLL{q}}$ is not a path. Without loss of generality, we can suppose that $\dual{\pathLL{p}}$ is not a path.
\begin{itemize} 
\item Remark  the following: Let $\pathLL{p}, \pathLL{q}, \pathLL{r}$ defined as above then there exist $\pathLL{p}', \pathLL{q}', \pathLL{r}'$ satisfying the same requirements as for $\pathLL{p}, \pathLL{q}, \pathLL{r}$ and such that $\pathLL{r}' = w\daimon$ where $w$ is a prefix of $\pathLL{r}$. Indeed,  as $\dual{\pathLL{p}}$ is not a path, there exists an action $\kappa^-$ justified by $\kappa^+$ such that $\pathLL{p} = w_1\kappa^+w_2\kappa^-w_3$ and $\overline{\kappa^+}$ does not appear in $\view{\,\overline{w_1\kappa^+w_2\kappa^-}\,}$. Hence $\pathLL{p}' := w_1\kappa^+w_2\kappa^-\daimon$ is a path such that $\dual{\pathLL{p}'}$ is not a path.
Let $\pathLL{r} = x_1\kappa^-x_2$ and note that $w_1\kappa^+w_2$ is a subsequence of $x_1$. Let $\pathLL{q}'$ be the subsequence of %
$x_1\kappa^-$ with actions in $\pathLL{q}$, then $\pathLL{q}'$ is a positive-ended prefix of $\pathLL{q}$, hence a path. 
Finally let $\pathLL{r}' := x_1\kappa^-\daimon$ then $\pathLL{r}'$ and $\dual{\pathLL r'}$ are paths such that $\pathLL{r}' \in \pathLL{p}' \shuffle \pathLL{q}'$ (by following the same construction steps as for $\pathLL{r}$).

\item Hence for proving the lemma, it suffices to consider triples $(\pathLL{p}', \pathLL{q}', \pathLL{r}')$ satisfying the following:
$\pathLL r'\in\pathLL p'\shuffle\pathLL q'$ is such that $\dual{\pathLL r'}$ is a path, $\pathLL p'=w_1\kappa^+w_2\kappa^-\daimon$, $\pathLL r'=w\kappa^-\daimon$ and $\dual{\pathLL p'}$ is not a path: $\overline{\kappa^+}$ does not appear in $\view{\,\overline{w_1\kappa^+w_2\kappa^-}\,}$.

\item Remark also that if lengths of $\pathLL p'$ and $\pathLL q'$ are less or equal to $2$ then $\dual{\pathLL p'}$ and $\dual{\pathLL q'}$ are paths.

\item Let $(\pathLL{p}_0, \pathLL{q}_0, \pathLL{r}_0)$ be such a triple with length of $\pathLL{r}_0$ minimal with respect to all such triples $(\pathLL{p}', \pathLL{q}', \pathLL{r}')$. 
Notice that $\kappa^-$ is not initial, otherwise $\dual{\pathLL{p}_0}$ would be a path. 
As $\dual{\pathLL{r}_0}$ is a path, let us write $\view{\,\dual{\pathLL r_0}\,} = w_0\overline{\alpha_n^+\alpha_n^-\dots\alpha_0^+\alpha_0^-}$ where $\alpha_n^+ = \kappa^+$ and $\alpha_0^- = \kappa^-$.
If $n=0$ then $\kappa^+$ precedes immediately $\kappa^-$ in $\pathLL r_0$, hence $\kappa^+$ precedes immediately $\kappa^-$ in $\pathLL p_0$, contradicting the fact that $\kappa^+$ does not occur in $\view{\,\dual{\pathLL p_0}\,}$. So $n > 0$.
Let us write $\pathLL{r}_0 = w'\alpha_1^-w''\alpha_0^+\kappa^-\daimon$. 
	\begin{itemize}
	\item Suppose $\alpha_0^+$ is an action of $\pathLL{p}_0$, then it is also the case for its justifier $\alpha_1^-$. Define $\pathLL{r}_1 = w'\kappa^-\daimon$. Remark that $\pathLL{r}_1$ is a path and its dual is also a path. Furthermore, we can define $\pathLL{q}_1$ (resp. $\pathLL{p}_1$) as the subsequence of $\pathLL{q}_0$ (resp. $\pathLL{p}_0$) present in $\pathLL{r}_1$. Remark that $\pathLL{r}_1 \in \pathLL{p}_1 \shuffle \pathLL{q}_1$ and $\dual{\pathLL{p}_1}$ is not a path. This contradicts the fact that $\pathLL{r}_0$ is minimal.
	\item Otherwise $\alpha_0^+$ is an action of $\pathLL{q}_0$, then it is also the case for its justifier $\alpha_1^-$. If actions in $w''$ are actions of $\pathLL{q}_0$, we define $\pathLL{r}_1$, $\pathLL{q}_1$, $\pathLL{p}_1$ as before and this yields a contradiction. Else let $\beta^+$ be the last action of $\pathLL{p}_0$ in $w''$. There is also an action   $\gamma^-$ of ${\pathLL p}_0$ which  immediately precedes $\beta^+$ in $w''$.
One can delete from $\pathLL{r}_0$ the actions $\gamma^-$ and $\beta^+$.
Then we get a shorter sequence $\pathLL{r}_1$ together with paths $\pathLL{p}_1$ and $\pathLL{q}_1$ such that $\dual{\pathLL{p}_1}$ is not a path. Hence a contradiction with the hypothesis of minimality.
\qed
	\end{itemize}
\end{itemize}


We prove now proposition~\ref{prop:OperationsOnVisitable} that characterizes, in the general case, the visitable paths of a tensor. This is a generalization of what was at stake for proving that the tensor is stable for regularity (proposition~\ref{prop:TensorStableRegular}).
We decompose the proof in two lemmas as each of them is quite long and technical (lemmas~\ref{lem:TensorThenShuffleDual} and~\ref{lem:ShuffleDualThenTensor}).
In each lemma, we need associativity of interaction in terms of paths (lemma~\ref{lem:assocSeq}). Associativity of interaction in terms of designs is proved in~\cite{DBLP:journals/mscs/Girard01}.
This initial lemma uses a notion of projection:
\begin{defi}[Projection]
Let $\pathLL p$ be a path, the projection of $\pathLL p$ on $X$, written $\proj{\pathLL p}{X}$, is the subsequence of $\pathLL p$ made of actions of $X$, $X$ being a set of actions, a path, a design or a net.
\end{defi}

\begin{lem}\label{lem:assocSeq}
Let $\design R$, $\design S$, $\design T$ be three nets of designs such that $\design S$ and $\design T$ have distinct bases and $\normalisation{\design R, \design S, \design T} = \daimon$,
then $\normalisationSeq{\design T}{\normalisation{\design R,\design S}} = \proj{\normalisationSeq{\design S \design T}{\design R}}{\design T}$.
\end{lem}
\proof
The proof is done by induction on the length of $\normalisationSeq{\design S \design T}{\design R}$. Note first that each proper step of normalization occurs between a subnet of $\design R$ and either a subnet of $\design S$ or a subnet of $\design T$ as $\design S$ and $\design T$ have distinct bases.
\\
If $\normalisationSeq{\design S \design T}{\design R}= \epsilon$ then $\design R = \daimon$ thus $\normalisation{\design R,\design S} = \daimon$ hence $\normalisationSeq{\design T}{\normalisation{\design R,\design S}} = \epsilon$. The result follows.
\\
Suppose the property satisfied for lengths less or equal to $n$ and $\normalisationSeq{\design S \design T}{\design R}$ has length $n+1$. We consider the various cases:
	\begin{itemize}
	\item If $\design S$ is positive: $\design S = (+,\sigma,I) \design S'$.
	\\ Then $\design R = (-,\sigma,I)\design R' \cup \bigcup_{J \neq I} (-,\sigma,J)\design R_J \cup \bigcup_{\xi \neq \sigma, K} (-,\xi,K)\design R_{\xi,K}$. 
	\\ Let $\design R'' = \design R' \cup \bigcup_{\xi \neq \sigma, K} (-,\xi,K)\design R_{\xi,K}$:
		\begin{itemize}
		\item $\normalisation{\design R,\design S} = \normalisation{\design R'',\design S'}$ then
		$\normalisationSeq{\design T}{\normalisation{\design R,\design S}} = \normalisationSeq{\design T}{\normalisation{\design R'',\design S'}}$
		\item $\normalisationSeq{\design S \design T}{\design R} = (+,\sigma,I)\normalisationSeq{\design S' \design T}{\design R''}$ then
		$\proj{\normalisationSeq{\design S \design T}{\design R}}{\design T} = \proj{\normalisationSeq{\design S' \design T}{\design R''}}{\design T}$
		\end{itemize}
			thus the result by induction hypothesis.
	\item If $\design R$ is positive:
		\begin{itemize}
		\item If $\design R = (+,\sigma,I) \design R'$ and $\design S = (-,\sigma,I)\design S' \cup \bigcup_{J \neq I} (-,\sigma,J)\design S_J \cup \bigcup_{\xi \neq \sigma, K} (-,\xi,K)\design S_{\xi,K}$. Let $\design S'' = \design S' \cup \bigcup_{\xi \neq \sigma, K} (-,\xi,K)\design S_{\xi,K}$:
			\begin{itemize}
			\item $\normalisation{\design R,\design S} = \normalisation{\design R',\design S''}$ then
		$\normalisationSeq{\design T}{\normalisation{\design R,\design S}} = \normalisationSeq{\design T}{\normalisation{\design R',\design S''}}$
			\item $\normalisationSeq{\design S \design T}{\design R} = (-,\sigma,I)\normalisationSeq{\design S'' \design T}{\design R'}$ then
		$\proj{\normalisationSeq{\design S \design T}{\design R}}{\design T} = \proj{\normalisationSeq{\design S'' \design T}{\design R'}}{\design T}$
			\end{itemize}
			thus the result by induction hypothesis.
		\item If $\design R = (+,\sigma,I) \design R'$ and $\design T = (-,\sigma,I)\design T' \cup \bigcup_{J \neq I} (-,\sigma,J)\design T_J \cup \bigcup_{\xi \neq \sigma, K} (-,\xi,K)\design T_{\xi,K}$. Let $\design T'' = \design T' \cup \bigcup_{\xi \neq \sigma, K} (-,\xi,K)\design T_{\xi,K}$:
			\begin{itemize}
			\item $\normalisation{\design R,\design S} = (+,\sigma,I)\normalisation{\design R',\design S}$ \\
	then $\normalisationSeq{\design T}{\normalisation{\design R,\design S}} 
			= \normalisationSeq{\design T}{(+,\sigma,I)\normalisation{\design R',\design S}}
			= (-,\sigma,I)\normalisationSeq{\design T''}{\normalisation{\design R',\design S}}$
			\item $\normalisationSeq{\design S \design T}{\design R} = (-,\sigma,I)\normalisationSeq{\design S \design T''}{\design R'}$ \\
	then $\proj{\normalisationSeq{\design S \design T}{\design R}}{\design T} = \proj{\normalisationSeq{\design S \design T''}{\design R'}}{\design T} = \proj{\normalisationSeq{\design S \design T''}{\design R'}}{\design T''}$
			\end{itemize}
			thus the result by induction hypothesis.
		\end{itemize}
	\item If $\design R$ and $\design S$ are negative, hence $\design T$ is positive: let $\design T = (+,\sigma,I) \design T'$ and $\design R = (-,\sigma,I)\design R' \cup \bigcup_{J \neq I} (-,\sigma,J)\design R_J \cup \bigcup_{\xi \neq \sigma, K} (-,\xi,K)\design R_{\xi,K}$.\\
		 Let $\design R'' = \design R' \cup \bigcup_{\xi \neq \sigma, K} (-,\xi,K)\design R_{\xi,K}$:
			\begin{itemize}
			\item $\normalisation{\design R,\design S} = 
				(-,\sigma,I)\normalisation{\design R',\design S}
				\cup \bigcup_{J \neq I} (-,\sigma,J)\normalisation{\design R_J,\design S} 
				\cup \bigcup_{\xi \neq \sigma, K} (-,\xi,K)\normalisation{\design R_{\xi,K},\design S}$ then
		$\normalisationSeq{\design T}{\normalisation{\design R,\design S}} 
			= (-,\sigma,I)\normalisationSeq{\design T''}{\normalisation{\design R',\design S}}$
			\item $\normalisationSeq{\design S \design T}{\design R} = (-,\sigma,I)\normalisationSeq{\design S \design T''}{\design R'}$ \\
		then $\proj{\normalisationSeq{\design S \design T}{\design R}}{\design T} = \proj{\normalisationSeq{\design S \design T''}{\design R'}}{\design T} = \proj{\normalisationSeq{\design S \design T''}{\design R'}}{\design T''}$
			\end{itemize}
			thus the result by induction hypothesis.
\qed
	\end{itemize}


\begin{lem}\label{lem:TensorThenShuffleDual}
Let $\behaviour{P}$ and $\behaviour{Q}$ be alien positive behaviours,\\
If $\pathLL r \in V_{\behaviour P \otimes \behaviour Q}$ then:
\begin{itemize}
	\item {\em(Shuffle condition)} $\pathLL r \in V_{{\behaviour{P}}} \lshuffle V_{{\behaviour{Q}}}$,
	\item {\em(Dual condition)} for all path $\overline{\pathLL s \kappa^-}$ in $\fullview{\,\dual{\pathLL r}\,}$, if there exist paths $\pathLL p' = (+,\xi,I)\pathLL p'_1 \in V_{\behaviour{P}}$\\
		and $\pathLL q' = (+,\xi,J)\pathLL q'_1 \in V_{\behaviour{Q}}$ with $\pathLL s \in (+,\xi,I \cup J)(\pathLL p'_1 \shuffle \pathLL q'_1)$, \\
		then either $\pathLL p' \kappa^-\daimon \in V_{\behaviour P}$ or $\pathLL q' \kappa^-\daimon \in V_{\behaviour Q}$.
\end{itemize}
\end{lem}
%
\proof
Let $\behaviour{P}$ and $\behaviour{Q}$ be alien positive behaviours with base $\vdash \xi$.
Let $\pathLL{r} \in V_{\behaviour{P} \otimes \behaviour{Q}}$. 
\\
{\em(Shuffle condition)}:
	\begin{itemize}
	\item As $\pathLL{r}$ is a visitable path, $\dual{\pathLL{r}}$ is a path. Furthermore there exist designs $\design{D} \in \behaviour{P} \otimes \behaviour{Q}$ and $\design{E} \in (\behaviour{P} \otimes \behaviour{Q})^\perp$ such that $\pathLL{r} = \normalisationSeq{\design{D}}{\design{E}}$. 

	\item Using the independence property (\cite{DBLP:journals/mscs/Girard01}, Th. 20), there exist designs $\design{D}_1 \in \behaviour{P}$ and $\design{D}_2 \in \behaviour{Q}$ such that $\design{D} = \design{D}_1 \otimes \design{D}_2$.

	\item If $\pathLL{r} = \daimon$, remark that $\pathLL{r} \in V_{\behaviour{P}}$ (and also $\pathLL{r} \in V_{\behaviour{Q}}$). So let us consider the other cases, \ie, designs $\design{D}_1$ and $\design{D}_2$ are distinct from the design $\{\daimon\}$. We write $\design D_1 = (+,\xi,I) \design D'_1$ and $\design D_2 = (+,\xi,J) \design D'_2$. Behaviours being alien, we have that $I \cap J = \emptyset$. Let $\design E = (-,\xi,I \cup J)(\design E' \cup \design E'')$.

	\item We have $\normalisationSeq{\design D_1 \otimes \design D_2}{\design E} = (+,\xi,I \cup J)\normalisationSeq{\design D'_1 ~ \design D'_2}{\design E'}$. Moreover, $\normalisation{\design D'_1 , \design D'_2 , \design E'} = \daimon$ and $\design D'_1$ and $\design D'_2$ have distinct bases. Hence lemma~\ref{lem:assocSeq} applies: 
	$\normalisationSeq{\design D'_1}{\normalisation{\design E',\design D'_2}} = \proj{\normalisationSeq{\design D'_1 \design D'_2}{\design E'}}{\design D'_1}$
	and
	$\normalisationSeq{\design D'_2}{\normalisation{\design E',\design D'_1}} = \proj{\normalisationSeq{\design D'_1 \design D'_2}{\design E'}}{\design D'_2}$.

	\item Let $\pathLL r_1 := \proj{\normalisationSeq{\design D'_1 \design D'_2}{\design E'}}{\design D'_1}$ and $\pathLL r_2 := \proj{\normalisationSeq{\design D'_1 \design D'_2}{\design E'}}{\design D'_2}$, then we have that $\pathLL r \in (+,\xi,I \cup J)\pathLL r_1 \shuffle (+,\xi,I \cup J)\pathLL r_2$.
Indeed, for each negative action that occurs in $\normalisationSeq{\design D'_1 \design D'_2}{\design E'}$ and is an address in $\design D'_1$ (resp. $\design D'_2$), then the next action in $\normalisationSeq{\design D'_1 \design D'_2}{\design E'}$ is positive and should also be an address in $\design D'_1$ (resp. $\design D'_2$) as in a path a positive action is in the same chronicle as the negative action that precedes it, hence also in an interaction path. 

	\item Furthermore following the adjunction theorem (\cite{DBLP:journals/mscs/Girard01}, Th. 14), 
we have that 
\begin{displaymath}
(-,\xi,I)\normalisation{\design{E}',\design{D}'_2} \cup \design E'' \in \behaviour P^\perp.
\end{displaymath}
Remark finally that $\normalisationSeq{\design D_1}{(-,\xi,I)\normalisation{\design{E}',\design{D}_2} \cup \design E''} = (+,\xi,I)\normalisationSeq{\design D'_1}{\normalisation{\design{E}',\design{D}_2}} = (+,\xi,I)\pathLL r_1$. 
Thus $(+,\xi,I)\pathLL r_1 \in V_{\behaviour P}$. Similarly, $(+,\xi,J)\pathLL r_2 \in V_{\behaviour Q}$.

	\item So $\pathLL r \in V_{\behaviour{P}} \shuffle V_{\behaviour{Q}}$.
	\end{itemize}
{\em(Dual condition)}. Let $\overline{\pathLL s \kappa^-}$ in $\fullview{\,\dual{\pathLL r}\,}$ such that there exist $\pathLL p' = (+,\xi,I)\pathLL p'_1 \in V_{\behaviour{P}}$ and $\pathLL q' = (+,\xi,J)\pathLL q'_1 \in V_{\behaviour{Q}}$ with $\pathLL s \in (+,\xi,I \cup J)(\pathLL p'_1 \shuffle \pathLL q'_1)$. Without loss of generality, we can suppose that $\kappa^-$ is an action occurring in $\behaviour P$. Then we have to prove that $\pathLL p' \kappa^-\daimon \in V_{\behaviour P}$:
	\begin{itemize}
	\item The justifier of $\kappa^-$ is in $\pathLL s$, hence the justifier of $\kappa^-$ is in $\pathLL p'$. So $\pathLL p' \kappa^-\daimon$ is a path.
Remark that $\pathLL s\kappa^-\daimon$ is a legal path and $\pathLL s\kappa^-\daimon \in (+,\xi,I \cup J)(\pathLL p'_1\kappa^-\daimon \shuffle \pathLL q'_1)$, thus by lemma~\ref{lem:dual_inverse_shuffle} the path 
$\pathLL p'\kappa^-\daimon$ is legal.
	\item We prove now that $\fullview{\pathLL p'\kappa^-\daimon}^c \in \behaviour P$. Remark that $\fullview{\pathLL p'\kappa^-\daimon}^c = \fullview{\pathLL p'}^c$. Finally, as $\pathLL p' \in V_{\behaviour P}$ we have that $\fullview{\pathLL p'}^c \in \behaviour P$. Hence the result.
	\item We prove now by contradiction that $\fullview{ \, \overline{\pathLL p'\kappa^-}\,}^c \in \behaviour P^\perp$. Let $\design D \in \behaviour P$ such that $\design D \not\perp \fullview{ \, \overline{\pathLL p'\kappa^-}\,}^c$. 
Remark that $\design D \neq \daimon$ and $\design D$ should begin with the action $(+,\xi,I)$ otherwise $\design D \perp \fullview{ \, \overline{\pathLL p'\kappa^-}\,}^c$.
We write $\design D =(+,\xi,I)\design D_1$.

Note that, as $\pathLL p' \in V_{\behaviour P}$, we have that $\design D \perp \fullview{ \, \overline{\pathLL p'}\,}^c$. Furthermore paths in $\fullview{ \, \overline{\pathLL p'}\,}^c$ are necessarily of finite length. Then the only possibility for the normalization to diverge is that there exists a legal path $\pathLL v = (+,\xi,I)\pathLL v_1$ such that \\
		-- $\pathLL v$ is a path of $\design D$ \\
		-- $\overline{\pathLL v\kappa^-}$ is a path of $\fullview{ \, \overline{\pathLL p'\kappa^-}\,}^c$ \\
		-- $\pathLL v\kappa^-$ is not a path of $\design D$.

Let $\fullview{\,\dual{\pathLL r}\,}^c = (-,\xi,I \cup J)\design R_1$. Note that $\design R_1$ is also complete with respect to negative actions: $\design R_1 = {\design R_1}^c$.
 
We detail the steps we need to conclude:
	\begin{itemize}
	\item As $\pathLL q' \in V_{\behaviour{Q}}$ then $\fullview{{\pathLL q'}}^c \in \behaviour Q$. Thus $\design D \otimes \fullview{{\pathLL q'}}^c \in \behaviour P \otimes \behaviour Q$.
	\\ As $\pathLL r \in V_{\behaviour P \otimes \behaviour{Q}}$ then $\fullview{\,\dual{\pathLL r}\,}^c \in (\behaviour P \otimes \behaviour Q)^\perp$.
	\\ It follows that $\design D \otimes \fullview{{\pathLL q'}}^c \perp \fullview{\,\dual{\pathLL r}\,}^c$, thus also $\normalisation{\design D_1, \fullview{{\pathLL q'_1}}^c, \design R_1} = \daimon$. Hence $\normalisationSeq{\design D_1}{\normalisation{\design R_1,\fullview{\pathLL q'_1}^c}}$ is well defined.
	\item Let us consider $\fullview{\pathLL p'} \otimes \fullview{\pathLL q'}$. 
We define $\pathLL s_1$ such that $\pathLL s = (+,\xi,I \cup J)\pathLL s_1$. 
As $\pathLL s \in (+,\xi,I \cup J)(\pathLL p'_1 \shuffle \pathLL q'_1)$, we have that $\fullview{\pathLL p'} \otimes \fullview{\pathLL q'} \perp \fullview{\overline{\pathLL s}\daimon}$. Thus $\normalisation{\fullview{\pathLL p'_1}, \fullview{\pathLL q'_1}, \fullview{\overline{\pathLL s_1}\daimon}} = \daimon$.
	\\ Hence by lemma~\ref{lem:assocSeq}, 
$\normalisationSeq{\fullview{\pathLL p'_1}}{\normalisation{\fullview{\,\overline{\pathLL s_1}\daimon\,},\fullview{\pathLL q'_1}}} = \proj{\normalisationSeq{\fullview{\pathLL p'_1} \otimes \fullview{\pathLL q'_1}}{\fullview{\,\overline{\pathLL s_1}\daimon\,}}}{\fullview{\pathLL p'_1}} = \pathLL p'_1$.
	\\ Thus $\fullview{\,\overline{\pathLL p'_1}\daimon\,} \subset \normalisation{\fullview{\,\overline{\pathLL s_1}\daimon\,},\fullview{\pathLL q'_1}}$. Hence, as there are no other actions with subaddresses of $\xi.i$ for $i \in I$ in $\normalisation{\fullview{\,\overline{\pathLL s_1}\daimon\,},\fullview{\pathLL q'_1}}$, we have that $\fullview{\,\overline{\pathLL p'_1}\daimon\,} = \normalisation{\fullview{\,\overline{\pathLL s_1}\daimon\,},\fullview{\pathLL q'_1}}$.
	\item Furthermore, as $\overline{\pathLL s}$ is a path of $\fullview{\,\dual{\pathLL r}\,}^c$, we have that $\overline{\pathLL s}$ is a path of $\design R_1$, hence also $\fullview{\,\overline{\pathLL s_1}\,} \subset \design R_1$.
	\\ Normalization being deterministic, we have then that $\fullview{\,\overline{\pathLL p'_1}\,} \subset \normalisation{\design R_1,\fullview{\pathLL q'_1}}$.
	\item Recall that $\normalisation{\design D_1, \fullview{\,\overline{\pathLL p'_1}\,}^c} = \daimon$, more precisely that $\pathLL v_1 = \normalisationSeq{\design D_1}{\fullview{\,\overline{\pathLL p'_1}\,}^c}$. And that $\normalisation{\design D_1, \design R_1,\fullview{\pathLL q'_1}^c} = \daimon$.
	\\ With the previous item, it follows that the path $\pathLL v_1$ is a prefix of $\normalisationSeq{\design D_1}{\normalisation{\design R_1,\fullview{\pathLL q'_1}^c}}$.
	\item Finally, $\overline{\pathLL v_1 \kappa^-}$ is a path of $\normalisation{\design R_1,\fullview{\pathLL q'_1}^c}$. Hence, for the normalization between $\design D_1$ and $\normalisation{\design R_1,\fullview{\pathLL q'_1}^c}$ to converge, we should also have $\pathLL v_1 \kappa^-$ to be a path of $\design D_1$, \ie, $\pathLL v\kappa^-$ to be a path of $\design D$.
	\item Contradiction with the hypothesis $\design D \not\perp \fullview{ \, \overline{\pathLL p'\kappa^-}\,}^c$.
	Thus $\fullview{ \, \overline{\pathLL p'\kappa^-}\,}^c \in \behaviour P^\perp$.

	\end{itemize}
As $\fullview{\pathLL p'\kappa^-\daimon}^c \in \behaviour P$ and $\fullview{ \, \overline{\pathLL p'\kappa^-}\,}^c \in \behaviour P^\perp$, we have that $\pathLL p'\kappa^-\daimon \in V_{\behaviour P}$.
\qed
\end{itemize}


\begin{lem}\label{lem:ShuffleDualThenTensor}
Let $\behaviour{P}$ and $\behaviour{Q}$ be alien positive behaviours,\\
$\pathLL r \in V_{\behaviour P \otimes \behaviour Q}$ if:
\begin{itemize}
	\item {\em(Shuffle condition)} $\pathLL r \in V_{{\behaviour{P}}} \lshuffle V_{{\behaviour{Q}}}$,
	\item {\em(Dual condition)} for all path $\overline{\pathLL s \kappa^-}$ in $\fullview{\,\dual{\pathLL r}\,}$, if there exist paths $\pathLL p' = (+,\xi,I)\pathLL p'_1 \in V_{\behaviour{P}}$\\
		and $\pathLL q' = (+,\xi,J)\pathLL q'_1 \in V_{\behaviour{Q}}$ with $\pathLL s \in (+,\xi,I \cup J)(\pathLL p'_1 \shuffle \pathLL q'_1)$, \\
		then either $\pathLL p' \kappa^-\daimon \in V_{\behaviour P}$ or $\pathLL q' \kappa^-\daimon \in V_{\behaviour Q}$.
\end{itemize}
\end{lem}
%
\proof
Let $\pathLL r$ satisfy {\em(Shuffle)} and {\em(Dual)} conditions. Suppose that $\pathLL r \in \pathLL p \shuffle \pathLL q$ where $\pathLL p \in V_{\behaviour{P}}$ and $\pathLL q \in V_{\behaviour{Q}}$.
	\begin{itemize}
	\item We first show that $\fullview{\pathLL r}^c \in \behaviour P \otimes \behaviour Q$. As $\pathLL p \in V_{\behaviour{P}}$ and $\pathLL q \in V_{\behaviour{Q}}$, we have that $\fullview{\pathLL p}^c \in \behaviour P$ and $\fullview{\pathLL q}^c \in \behaviour Q$. Hence $\fullview{\pathLL p}^c \otimes \fullview{\pathLL q}^c \in \behaviour P \otimes \behaviour Q$. Finally remark that $\fullview{\pathLL r}^c = \fullview{\pathLL p}^c \otimes \fullview{\pathLL q}^c$.

	\item We show now by contradiction that $\fullview{\,\dual{\pathLL r}\,}^c \in (\behaviour P \otimes \behaviour Q)^\perp$. Let $\design D \in \behaviour P$ and $\design E \in \behaviour Q$ such that $\design D \otimes \design E \not\perp \fullview{\,\dual{\pathLL r}\,}^c$. As $\fullview{\,\dual{\pathLL r}\,}^c$ is complete with respect to negative actions and paths of $\fullview{\,\dual{\pathLL r}\,}^c$ have a finite length, divergence occurs if there exists a path $\pathLL s$ such that
		\\ -- $\pathLL s$ is a path of $\design D \otimes \design E$,
		\\ -- $\overline{\pathLL s\kappa^-}$ is a path of $\fullview{\,\dual{\pathLL r}\,}^c$ hence of $\fullview{\,\dual{\pathLL r}\,}$,
		\\ -- $\pathLL s\kappa^-$ is not a path of $\design D \otimes \design E$.		
		\\
Obviously, neither $\design D$ nor $\design E$ is the daimon, otherwise $\design D \otimes \design E \perp \fullview{\,\dual{\pathLL r}\,}^c$.

We can choose $\design D$ and $\design E$ such that $\pathLL s$ is of minimal length with respect to such pairs of designs non orthogonal to $\fullview{\,\dual{\pathLL r}\,}^c$.
The path $\pathLL s$ defines a path $\pathLL p' = (+,\xi,I)\pathLL p'_1$ (resp.\ $\pathLL q' = (+,\xi,J)\pathLL q'_1$) in $\design D$ (resp.\ in $\design E$) such that $\pathLL s \in (+,\xi,I \cup J)(\pathLL p'_1 \shuffle \pathLL q'_1)$.
		\begin{itemize}
		\item Remark that $\fullview{\pathLL p'} \subset \design D \in \behaviour P$, thus $\fullview{\pathLL p'}^c \in \behaviour P$. Similarly, $\fullview{\pathLL q'}^c \in \behaviour Q$.
		\item As $\overline{\pathLL s}$ is a path in $\fullview{\,\dual{\pathLL r}\,}^c$, we have that $\fullview{\,\overline{\pathLL s}\,} \subset \fullview{\,\dual{\pathLL r}\,}^c$.

		\item We show by contradiction that $\fullview{\,\dual{\pathLL s}\,}^c \in (\behaviour P \otimes \behaviour Q)^\perp$.
		Let $\design D' \in \behaviour P$ and $\design E' \in \behaviour Q$ such that $\design D' \otimes \design E' \not\perp \fullview{\,\dual{\pathLL s}\,}^c$. Divergence occurs necessarily because there exists a path $\pathLL v$ such that
		\\ -- $\pathLL v$ is a path of $\design D' \otimes \design E'$,
		\\ -- $\overline{\pathLL v\kappa'^-}$ is a path of $\fullview{\,\dual{\pathLL s}\,}^c$ hence of $\fullview{\,\dual{\pathLL s}\,}$,
		\\ -- $\pathLL v\kappa'^-$ is not a path of $\design D' \otimes \design E'$.
		\\ As $\fullview{\,\overline{\pathLL s}\,} \subset \fullview{\,\dual{\pathLL r}\,}^c$, we have that $\overline{\pathLL v\kappa'^-}$ is a path of $\fullview{\,\dual{\pathLL r}\,}^c$. Thus $\design D' \otimes \design E' \not\perp \fullview{\,\dual{\pathLL r}\,}^c$. Moreover $\pathLL v$ is strictly shorter than $\pathLL s$. This contradicts the fact that $\pathLL s$ is of minimum length. So $\fullview{\,\dual{\pathLL s}\,}^c \in (\behaviour P \otimes \behaviour Q)^\perp$.	

		\item We show now that $\pathLL p' \in V_{\behaviour P}$.
		Let $\design D' \in \behaviour P$. We use the following notations:
		\\ -- $\design D = (+,\xi,I)\design D_1$
		\\ -- $\design D' = (+,\xi,I)\design D'_1$
		\\ -- $\design E = (+,\xi,I)\design E_1$
		\\ -- $\fullview{\,\dual{\pathLL s}\,}^c = (+,\xi,I)\design S_1$

		As $\fullview{\,\dual{\pathLL s}\,}^c \in (\behaviour P \otimes \behaviour Q)^\perp$, we have that $\normalisation{\design D \otimes \design E, \fullview{\,\dual{\pathLL s}\,}^c} = \daimon$, hence also $\normalisation{\design D_1, \design E_1, \design S_1} = \daimon$.
		\\ By lemma~\ref{lem:assocSeq}, we have that $\normalisationSeq{\design D_1}{\normalisation{\design S_1,\design E_1}} = \proj{\normalisationSeq{\design D_1 \design E_1}{\design S_1}}{\design D_1} = \pathLL p'_1$. Thus $\fullview{\,\dual{\pathLL p'_1}\,} \subset \normalisation{\design S_1,\design E_1}$.
		\\ Furthermore $\fullview{\,\dual{\pathLL s}\,}^c \perp \design D' \otimes \design E$. Thus $\normalisation{\design D', \design E, \fullview{\,\dual{\pathLL s}\,}^c} = \daimon$, hence by associativity 
\begin{displaymath}
\normalisation{\design D', \normalisation{\design E, \fullview{\,\dual{\pathLL s}\,}^c}} = \daimon.
\end{displaymath}
 Thus $\normalisation{\design D'_1, \normalisation{\design E_1, \design S_1}} = \daimon$. 
		\\ Hence as $\fullview{\,\dual{\pathLL p'_1}\,} \subset \normalisation{\design S_1,\design E_1}$, we have $\normalisation{\design D'_1, \fullview{\,\dual{\pathLL p'_1}\,}^c} = \daimon$. It follows that $\normalisation{\design D', \fullview{\,\dual{\pathLL p'}\,}^c} = \daimon$, \ie, $\design D'\perp \fullview{\,\dual{\pathLL p'}\,}^c$.
		\\ So $\fullview{\,\dual{\pathLL p'}\,}^c \in \behaviour P^\perp$. Thus, as $\pathLL p'$ is a path of $\design D \in \behaviour P$, it follows that $\pathLL p' \in V_{\behaviour P}$. 

		\item Similarly, we can prove that $\pathLL q' \in V_{\behaviour Q}$.

		\item Using the constraint, we should have $\pathLL p' \kappa^-\daimon \in V_{\behaviour P}$. As $\design D \in \behaviour P$ and $\pathLL p'$ is a path of $\design D$, we should also have $\pathLL p'\kappa^-$ a path of $\design D$ (necessary condition for visitable paths). Hence a contradiction. So $\fullview{\,\dual{\pathLL r}\,}^c \in (\behaviour P \otimes \behaviour Q)^\perp$.
		\end{itemize}
	\item As $\fullview{\pathLL r}^c \in \behaviour P \otimes \behaviour Q$ and $\fullview{\,\dual{\pathLL r}\,}^c \in (\behaviour P \otimes \behaviour Q)^\perp$, we have that $\pathLL r \in V_{\behaviour P \otimes \behaviour Q}$.
\qed
	\end{itemize} 

%
\restateOperationsOnVisitable*
\begin{proof}
Follows from lemmas~\ref{lem:TensorThenShuffleDual} and~\ref{lem:ShuffleDualThenTensor}.
\end{proof}

\subsection{Proofs of Subsection~\ref{sec:SimpleOrientedTensor}}\label{subsec:proofs:SimpleOrientedTensor}

\restatecliqueDecomposition*
\proof
$(\Leftarrow)$ Let $C = C_1 \cup C'C''$ with $C_1$, $C'$, $C''$ as defined in the lemma. 
	\begin{itemize}
	\item $C$ is a clique: Note that $\fullview{\pathLL p\pathLL q} = \fullview{\pathLL p} \cup \fullview{\kappa_0^+\pathLL q}$ as $\behaviour A$ and $\behaviour B$ are disjoint. Hence the fact that $C$ is a clique follows from the fact that $C_1$ and $C''$ are cliques.

	\item $C$ is a maximal clique: Let $\pathLL r \in S$ and $\pathLL r \coh \pathLL r'$ for all $\pathLL r' \in C$.
	Either $\pathLL r \in V_{\behaviour A_{[\behaviour B]}}$ hence $\pathLL r \in C_1 \subset C$ as $C_1$ is a maximal clique of $V_{\behaviour A_{[\behaviour B]}}$.
	Or $\pathLL r = \pathLL p\pathLL q$ where $\pathLL p \in V_{\behaviour A_{[\behaviour B]}}^{max}$ and $\pathLL q \in V_{\behaviour B}^-$. Thus $\pathLL p \in C_1$ and $\pathLL q \in C''$ as $C_1$ and $C''$ are maximal. Thus $\pathLL r \in C$.

	\item $C$ is positively saturated: Let $\pathLL m \in C$, $\pathLL n\kappa^-\kappa^+ \in C$, $\pathLL m\kappa^-\daimon \in S$:
		\begin{itemize}
		\item If $\pathLL m \in C_1$ and $\pathLL n\kappa^-\kappa^+ \in C_1$. The result follows from the fact that $C_1$ is positively saturated for $V_{\behaviour A_{[\behaviour B]}}$: the path $\pathLL m\kappa^-\kappa^+ \in C_1$ thus $\pathLL m\kappa^-\kappa^+ \in S$.
		\item If $\pathLL m \in C_1$ and $\pathLL n\kappa^-\kappa^+ \in C'C''$, \ie, $\pathLL n\kappa^-\kappa^+ = \pathLL p\pathLL q\kappa^-\kappa^+$ with $\pathLL p \in C'$ and $\pathLL q\kappa^-\kappa^+ \in C''$. As $\pathLL m\kappa^-\daimon \in S$, we have that $\kappa_0^+\kappa^-\daimon \in V_{\behaviour B_{[\behaviour A]}}$. Note also that $\kappa_0^+\pathLL q\kappa^-\kappa^+ \in \kappa_0^+C''$. Thus $\kappa_0^+\kappa^-\kappa^+ \in \kappa_0^+C''$ as $\kappa_0^+C''$ is a positively saturated maximal clique for $V_{\behaviour B_{[\behaviour A]}}$. It follows that $\pathLL m\kappa^-\kappa^+ \in S$.
		\item If $\pathLL m \in C'C''$ and $\pathLL n\kappa^-\kappa^+ \in C_1$: as $\pathLL m\kappa^-\daimon \in S$ then $\pathLL m\in C' \subset C_1$. The result follows from the fact that $C_1$ is positively saturated for $V_{\behaviour A_{[\behaviour B]}}$.
		\item If $\pathLL m \in C'C''$ and $\pathLL n\kappa^-\kappa^+ \in C'C''$, \ie, $\pathLL m = \pathLL p\pathLL q$ with $\pathLL p \in C'$ and $\pathLL q \in C''$, and $\pathLL n\kappa^-\kappa^+ = \pathLL p'\pathLL q'\kappa^-\kappa^+$ with $\pathLL p' \in C'$ and $\pathLL q'\kappa^-\kappa^+ \in C''$. As $\pathLL m\kappa^-\daimon \in S$ then $\kappa_0^+\pathLL q\kappa^-\daimon \in V_{\behaviour B_{[\behaviour A]}}$. Thus $\kappa_0^+\pathLL q\kappa^-\kappa^+ \in \kappa_0^+C''$ as $\kappa_0^+C''$ is a positively saturated maximal clique for $V_{\behaviour B_{[\behaviour A]}}$. It follows that $\pathLL m\kappa^-\kappa^+ \in S$.
		\end{itemize}
	\end{itemize}

\noindent $(\Rightarrow)$ Let $C$ be a positively saturated maximal clique of $S$. 
	\begin{itemize}
	\item $C$ may be factorized: Let $C = C_1 \cup C_2$ with $C_1 \subset V_{\behaviour A_{[\behaviour B]}}$ and $C_2 \subset V_{\behaviour A_{[\behaviour B]}}^{max}V_{\behaviour B}^-$. We remark the following points: let $\pathLL p_1 \in C_1$, $\pathLL p'_2\pathLL q'_2 \in C_2$ and $\pathLL p''_2\pathLL q''_2 \in C_2$. Then $\pathLL p'_2\pathLL q'_2 \coh \pathLL p''_2\pathLL q''_2$ and foci of $\pathLL p'_2$ and $\pathLL p''_2$ are disjoint from foci of $\pathLL q'_2$ and $\pathLL q''_2$. Thus  $\pathLL q'_2 \coh \pathLL q''_2$ hence $\pathLL p'_2\pathLL q''_2 \in C_2$ and $\pathLL p''_2\pathLL q'_2 \in C_2$. It follows that the set $C_2$ may be factorized: We can write $C_2 = C'C''$ with $C' \subset V_{\behaviour A_{[\behaviour B]}}^{max}$ and $C'' \subset V_{\behaviour B}^-$.
	\item As $C$ is a clique, $C_1$, $C'$ and $C''$ are cliques.

	\item By construction if $C'$ is empty, $C''$ is also empty. Suppose $C'$ not empty, let $\pathLL q \in V_{\behaviour B}^-$ such that $\pathLL q \coh C''$ then $C'\pathLL q \coh C'C''$ and also $C'\pathLL q \coh C_1$: $C''$ is a maximal clique or is empty.

	\item Furthermore $\pathLL p_1 \coh \pathLL p'_2 \coh \pathLL p''_2$ hence $\pathLL p'_2 \in C_1$ and $\pathLL p''_2 \in C_1$. It follows that $C' \subset C_1$. Hence by construction, $C' = C_1 \cap V_{\behaviour A_{[\behaviour B]}}^{max}$.

	\item Let $\pathLL p \in V_{\behaviour A_{[\behaviour B]}}$ such that $\pathLL p \coh C_1$ then $\pathLL p \coh C'$ as $C' \subset C_1$, thus $\pathLL p \coh C'C''$, so $\pathLL p \in C$, thus $\pathLL p \in C_1$: $C_1$ is a maximal clique or is empty.

	\item As $C$ is positively saturated, $C_1$ and $\kappa_0^+C''$ are also positively saturated.
\qed
	\end{itemize}

\restatedualcliqueDecomposition*
\proof
There is a unique $\kappa_0^-$ as behaviours are connected.
Note that elements of $V_{\behaviour A_{[\behaviour B]}}^{max}$ are $\daimon$-free then paths of $D'_{\kappa_1^-}$ end with a negative action.
Note also each $D''_{\kappa_1^-}$ has a unique first action as the first action of paths of $D''_{\kappa_1^-}$ should be positive. Finally it follows that $D'_{\kappa_1^-} D''_{\kappa_1^-} \coh D'_{\kappa_1'^-} D''_{\kappa_1'^-}$ for $\kappa_1^- \neq \kappa_1'^-$, and that if there is an action in common between $D''_{\kappa_1^-}$ and $D''_{\kappa_1'^-}$ then $\kappa_1^- = \kappa_1'^-$.
\\
$(\Leftarrow)$ Let $D = D_1 \cup  \bigcup_{\kappa_1^- \in K} D'_{\kappa_1^-} D''_{\kappa_1^-}$ with $D_1$, $D'_{\kappa_1^-}$, $D''_{\kappa_1^-}$ as defined in the lemma. 
	\begin{itemize}
	\item $D$ is a clique: This follows by construction of $D$ and the fact that $D_1$, $D'_{\kappa_1^-}$, $D''_{\kappa_1^-}$ are cliques.

	\item $D$ is a maximal clique: Let $\pathLL r \in \dual S$ and $\pathLL r \coh \pathLL r'$ for all $\pathLL r' \in D$.
	Either $\pathLL r \in \dual{V_{\behaviour A_{[\behaviour B]}}}$ hence $\pathLL r \in D_1 \subset D$ as $D_1$ is a maximal clique of $\dual{V_{\behaviour A_{[\behaviour B]}}}$.
	Or $\pathLL r = \pathLL p\pathLL q$ where $\pathLL p \in \overline{V_{\behaviour A_{[\behaviour B]}}^{max}}$ and $\pathLL q \in \dual{V_{\behaviour B}^-}$. Remark that $\pathLL p \coh \pathLL r'$ for all $\pathLL r' \in D$, thus $\pathLL p \in D_1$ hence $\pathLL p \in D'_{\kappa_1^-}$ for some action $\kappa_1^-$. Now remark that $\pathLL p\pathLL q \coh \pathLL p\pathLL q'$ for all $\pathLL q' \in D''_{\kappa_1^-}$. Thus $\pathLL q \in D''_{\kappa_1^-}$ as $D''_{\kappa_1^-}$ is a maximal clique. Thus $\pathLL r \in D$.

	\item $D$ is positively saturated: Let $\pathLL m \in D$, $\pathLL n\kappa^-\kappa^+ \in D$, $\pathLL m\kappa^-\daimon \in \dual S$:
		\begin{itemize}
		\item If $\pathLL m \in D_1$ and $\pathLL n\kappa^-\kappa^+ \in D_1$: the result follows from the fact that $D_1$ is positively saturated.
		\item If $\pathLL m \in D_1$ and $\pathLL n\kappa^-\kappa^+ \in D'_{\kappa_1^-} D''_{\kappa_1^-}$. As $\pathLL m\kappa^-\daimon \in \dual S$, we have that $\kappa^-$ is an action occurring in $\overline{V_{\behaviour A_{[\behaviour B]}}}$. Hence $\kappa^+$ occurs in $\overline{V_{\behaviour A_{[\behaviour B]}}}$. It follows that $\overline{\kappa_0^+}\kappa^+ \in \dual{V_{\behaviour B_{[\behaviour A]}}}$. Thus $\pathLL m\kappa^-\kappa^+ \in \dual S$.
		\item If $\pathLL m \in D'_{\kappa_1^-} D''_{\kappa_1^-} $ and $\pathLL n\kappa^-\kappa^+ \in D_1$: as $\pathLL m\kappa^-\daimon \in \dual S$ then $\pathLL m\in D_1$. The result follows from the fact that $D_1$ is positively saturated.
		\item If $\pathLL m \in D'_{\kappa_1^-} D''_{\kappa_1^-} $ and $\pathLL n\kappa^-\kappa^+ \in D'_{\kappa_1'^-} D''_{\kappa_1'^-}$, \ie, $\pathLL m = \pathLL m_0\pathLL m_1$ with $\pathLL m_0 \in D'_{\kappa_1^-} $ and $\pathLL m_1 \in D''_{\kappa_1^-} $, and $\pathLL n\kappa^-\kappa^+ = \pathLL n_0\pathLL n_1\kappa^-\kappa^+$ with $\pathLL n_0 \in D'_{\kappa_1'^-}$ and $\pathLL n_1\kappa^-\kappa^+ \in D''_{\kappa_1'^-}$. Furthermore $\pathLL m\kappa^-\daimon \in \dual S$ thus $\kappa_1^- = \kappa_1'^-$: let $\kappa'^+$ be the justifier of $\kappa^-$, then this action $\kappa'^+$ occurs in $\pathLL m_1$ and also in $\pathLL n_1$, thus there is a common view between $\fullview{\pathLL m_1}$ and $\fullview{\pathLL n_1}$, it follows that $\kappa_1^- = \kappa_1'^-$.
		Note now that $\kappa_0^-\pathLL m_1 \in D''_{\kappa_1^-}$ and $\kappa_0^-\pathLL n_1\kappa^-\kappa^+ \in D''_{\kappa_1^-}$ and that $\kappa_0^-\pathLL m_1\kappa^-\daimon \in \dual{V_{\behaviour B_{[\behaviour A]}}}$. Hence as $D''_{\kappa_1^-}$ is positively saturated for $\dual{V_{\behaviour B_{[\behaviour A]}}}$, we have that $\kappa_0^-\pathLL m_1\kappa^-\kappa^+ \in \dual{V_{\behaviour B_{[\behaviour A]}}}$. It follows that $\pathLL m\kappa^-\kappa^+ \in \dual S$.
		\end{itemize}
	\end{itemize}

\noindent $(\Rightarrow)$ Let $D$ be a positively saturated maximal clique of $\dual S$. 
	\begin{itemize}
	\item $D$ may be factorized: Let $D = D_1 \cup D_2$ with $D_1 \subset \dual{V_{\behaviour A_{[\behaviour B]}}}$ and $D_2 \subset \overline{V_{\behaviour A_{[\behaviour B]}}^{max}}\dual{V_{\behaviour B}^-}$. 
We remark the following points: let $\pathLL p_1 \in D_1$, $\pathLL p'_2\pathLL q'_2 \in D_2$ and $\pathLL p''_2\pathLL q''_2 \in D_2$ (with $\pathLL p'_2, \pathLL p''_2 \in \overline{V_{\behaviour A_{[\behaviour B]}}^{max}}$ and $\pathLL q'_2, \pathLL q''_2 \in \dual{V_{\behaviour B}^-}$). We have that $\pathLL p'_2\pathLL q'_2 \coh \pathLL p''_2\pathLL q''_2$.
Thus  $\pathLL p'_2 \coh \pathLL p''_2$. If $\pathLL p'_2$ and $\pathLL p''_2$ end on the same action then $\pathLL q'_2 \coh \pathLL q''_2$. In such a case, $\pathLL p'_2\pathLL q''_2 \in D_2$ and $\pathLL p''_2\pathLL q'_2 \in D_2$: the set $D_2$ may be factorized. 
Let $K$ be the set of last actions of such $\daimon$-free paths $\pathLL p'_2$ and $D'_{\kappa_1^-}$ be the set of paths $\pathLL p'_2$ ending with the same last action $\kappa_1^-$. We define $D''_{\kappa_1^-}$ to be the set of $\pathLL q'_2$ such that $\pathLL p'_2\pathLL q'_2 \in D_2$ and $\pathLL p'_2$ ends on action $\kappa_1^-$.
We can write $D_2 = \bigcup_{\kappa_1^- \in K} D'_{\kappa_1^-} D''_{\kappa_1^-}$ with for all $\kappa_1^-$, $D'_{\kappa_1^-} \subset \overline{V_{\behaviour A_{[\behaviour B]}}^{max}}$ and $D''_{\kappa_1^-} \subset \dual{V_{\behaviour B}^-}$.
	\item As $D$ is a clique, $D_1$, $D'_{\kappa_1^-}$ and $D''_{\kappa_1^-}$ are cliques.

	\item By construction if $D'_{\kappa_1^-}$ is empty, $D''_{\kappa_1^-}$ is also empty. Suppose $D'_{\kappa_1^-}$ not empty, let $\pathLL q \in \dual{V_{\behaviour B}^-}$ such that $\pathLL q \coh D''_{\kappa_1^-}$ then $D'_{\kappa_1^-}\pathLL q \coh D'_{\kappa_1^-} D''_{\kappa_1^-}$ and also $D'_{\kappa_1^-}\pathLL q \coh D_1$: $D''_{\kappa_1^-}$ is a maximal clique or is empty.

	\item With notations as before, we have that $\pathLL p_1 \coh \pathLL p'_2 \coh \pathLL p''_2$ hence $\pathLL p'_2 \in D_1$ and $\pathLL p''_2 \in D_1$. It follows that $D'_{\kappa_1^-} \subset D_1$. Hence by construction, $D'_{\kappa_1^-} = \overline{\dual{D_1}\kappa_1^+ \cap V_{\behaviour A_{[\behaviour B]}}^{max}}$.

	\item Let $\pathLL p \in \dual{V_{\behaviour A_{[\behaviour B]}}}$ such that $\pathLL p \coh D_1$ then $\pathLL p \coh D'_{\kappa_1^-}$ as $D'_{\kappa_1^-} \subset D_1$, thus $\pathLL p \coh D'_{\kappa_1^-} D''_{\kappa_1^-}$, so $\pathLL p \in D$, thus $\pathLL p \in D_1$: $D_1$ is a maximal clique or is empty.

	\item As $D$ is positively saturated, then $D_1$ and $\kappa_0^-D''_{\kappa_1^-}$ are also all positively saturated.
\qed
	\end{itemize}


\end{document}